\documentclass[journal]{IEEEtran}

\usepackage{xcolor}
\usepackage{cite}
\usepackage{amsmath,amssymb,amsfonts}
\usepackage{algorithmic}
\usepackage{graphicx}
\usepackage{textcomp}

\usepackage{mathtools}
\usepackage{amsthm}

\usepackage{booktabs}

\usepackage{siunitx}

\usepackage[acronym, style=alttree, toc=true, shortcuts, xindy, nomain, nonumberlist]{glossaries}

\usepackage{cleveref}

\usepackage{hyphenat}
\usepackage{bm}

\newcommand{\gvec}[1]{\boldsymbol{#1}}
\newcommand{\Mat}[1]{\mathbf{\MakeUppercase{#1}}}

\newcommand{\norm}[1]{\left\lVert#1\right\rVert} 
\newcommand{\set}[1]{\mathcal{#1}} 
\newcommand{\setdef}[2][]{
	\left\{
		\ifblank{#1}{}{#1 \hspace{.1cm} \middle| \hspace{.1cm}}
		#2
	\right\}
} 

\newcommand{\rv}[1]{\mathrm{#1}} 
\newcommand{\lr}[1]{\left(#1\right)} 

\newcommand{\sign}[1]{\text{sign}\left(#1\right)} 
\newcommand{\T}[0]{T}

\DeclarePairedDelimiterXPP\onenorm[1]{}\lVert\rVert{_1}{\ifblank{#1}{\:\cdot\:}{#1}} 
\DeclarePairedDelimiterXPP\twonorm[1]{}\lVert\rVert{_2}{\ifblank{#1}{\:\cdot\:}{#1}} 

\newtheorem{theorem}{Theorem}
\newtheorem{assumption}{Assumption}
\newtheorem{remark}{Remark}
\newtheorem{definition}{Definition}
\newtheorem{lemma}{Lemma}

\newglossary{symbols}{sym}{sbl}{List of Symbols}
\newglossary{notation}{not}{nt}{Notation}

\glsaddkey{dimension}{\glsentrytext{\glslabel}}{\glsentryd}{\GLsentryd}{\glsd}{\Glsd}{\GLSd}
\glsaddkey{body}{\glsentrytext{\glslabel}}{\glsentrybody}{\GLsentrybody}{\glsbody}{\Glsbody}{\GLSbody}

\newacronym{ocp}{OCP}{optimal control problem}
\newacronym{MPC}{MPC}{Model Predictive Control}
\newacronym{MPCa}{MPC algorithm}{MPC algorithm}
\newacronym{RMPC}{RMPC}{Robust Model Predictive Control}
\newacronym{SMPC}{SMPC}{Stochastic Model Predictive Control}
\newacronym{SCMPC}{SCMPC}{Scenario Stochastic Model Predictive Control}
\newacronym{MILP}{MILP}{Mixed Integer Linear Program}
\newacronym{PIT}{PIT}{Pointwise\hyp{}In\hyp{}Time}
\newacronym{POMDP}{POMDP}{Partially Observable Markov Decision Process}
\newacronym{MDP}{MDP}{Markov Decision Process}
\newacronym{KKT}{KKT}{Karush\hyp{}Kuhn\hyp{}Tucker}
\newacronym{SMPCFT}{SMPC+FTP}{\textit{Stochastic Model Predictive Control + fail-safe trajectory planning}}
\newacronym{FT}{FTP}{fail-safe trajectory planning}
\newacronym{tv}{TV}{target vehicle}
\newacronym{ev}{EV}{ego vehicle}

\newglossaryentry{chiUk}{type=symbols,
	sort={control},
	name={\ensuremath{\boldsymbol{\chi}_h^{\bm{U}_h}}},
	description={?}
}
\newglossaryentry{chiX}{type=symbols,
	sort={control},
	name={\ensuremath{\boldsymbol{\chi}}},
	description={?}
}
\newglossaryentry{xik}{type=symbols,
	sort={control},
	name={\ensuremath{\bm{\xi}_k}},
	description={?}
}
\newglossaryentry{xi0}{type=symbols,
	sort={control},
	name={\ensuremath{\bm{\xi}_0}},
	description={?}
}
\newglossaryentry{sk}{type=symbols,
	sort={control},
	name={\ensuremath{s_k}},
	description={?}
}
\newglossaryentry{phik}{type=symbols,
	sort={control},
	name={\ensuremath{\phi_k}},
	description={?}
}
\newglossaryentry{dk}{type=symbols,
	sort={control},
	name={\ensuremath{d_k}},
	description={?}
}
\newglossaryentry{vk}{type=symbols,
	sort={control},
	name={\ensuremath{v_k}},
	description={?}
}
\newglossaryentry{uk}{type=symbols,
	sort={control},
	name={\ensuremath{\bm{u}_k}},
	description={?}
}
\newglossaryentry{ak}{type=symbols,
	sort={control},
	name={\ensuremath{a_k}},
	description={?}
}
\newglossaryentry{deltak}{type=symbols,
	sort={control},
	name={\ensuremath{\delta_k}},
	description={?}
}

\newglossaryentry{xitv}{type=symbols,
	sort={control},
	name={\ensuremath{\bm{\xi}^{\text{TV}}_k}},
	description={?}
}
\newglossaryentry{Rtv}{type=symbols,
	sort={control},
	name={\ensuremath{\mathcal{R}^{\text{TV}}_k}},
	description={?}
}
\newglossaryentry{bRtv}{type=symbols,
	sort={control},
	name={\ensuremath{\overline{\mathcal{R}}^{\text{TV}}_k}},
	description={?}
}
\newglossaryentry{xitv0}{type=symbols,
	sort={control},
	name={\ensuremath{\bm{\xi}^{\text{TV}}_0}},
	description={?}
}
\newglossaryentry{dxEVTV}{type=symbols,
	sort={control},
	name={\ensuremath{ \Delta x^\text{EV,TV}_0 }},
	description={?}
}

\newglossaryentry{hxitv0}{type=symbols,
	sort={control},
	name={\ensuremath{\bm{\hat{\xi}}^{\text{TV}}_0}},
	description={?}
}
\newglossaryentry{hxitv}{type=symbols,
	sort={control},
	name={\ensuremath{\bm{\hat{\xi}}^{\text{TV}}}},
	description={?}
}
\newglossaryentry{xitvref}{type=symbols,
	sort={control},
	name={\ensuremath{\bm{\xi}^{\text{TV}}_{\text{ref},k}}},
	description={?}
}
\newglossaryentry{utv}{type=symbols,
	sort={control},
	name={\ensuremath{\bm{u}^{\text{TV}}_k}},
	description={?}
}
\newglossaryentry{tutv}{type=symbols,
	sort={control},
	name={\ensuremath{\bm{\tilde{u}}^{\text{TV}}_k}},
	description={?}
}
\newglossaryentry{wtv}{type=symbols,
	sort={control},
	name={\ensuremath{\bm{w}^{\text{TV}}_k}},
	description={?}
}
\newglossaryentry{wsens}{type=symbols,
	sort={control},
	name={\ensuremath{\bm{w}^{\text{sens}}_0}},
	description={?}
}
\newglossaryentry{Wsens}{type=symbols,
	sort={control},
	name={\ensuremath{\mathcal{W}^{\text{sens}}}},
	description={?}
}
\newglossaryentry{xtv}{type=symbols,
	sort={control},
	name={\ensuremath{x^{\text{TV}}_k}},
	description={?}
}
\newglossaryentry{ytv}{type=symbols,
	sort={control},
	name={\ensuremath{y^{\text{TV}}_k}},
	description={?}
}
\newglossaryentry{vxtv}{type=symbols,
	sort={control},
	name={\ensuremath{v_{x,k}^{\text{TV}}}},
	description={?}
}
\newglossaryentry{vytv}{type=symbols,
	sort={control},
	name={\ensuremath{v_{y,k}^{\text{TV}}}},
	description={?}
}
\newglossaryentry{xtv0}{type=symbols,
	sort={control},
	name={\ensuremath{x^{\text{TV}}_0}},
	description={?}
}
\newglossaryentry{ytv0}{type=symbols,
	sort={control},
	name={\ensuremath{y^{\text{TV}}_0}},
	description={?}
}
\newglossaryentry{vxtv0}{type=symbols,
	sort={control},
	name={\ensuremath{v_{x,0}^{\text{TV}}}},
	description={?}
}
\newglossaryentry{vytv0}{type=symbols,
	sort={control},
	name={\ensuremath{v_{y,0}^{\text{TV}}}},
	description={?}
}
\newglossaryentry{evlane}{type=symbols,
	sort={control},
	name={\ensuremath{y_{\text{lane},0}^\text{EV}}},
	description={?}
}
\newglossaryentry{tvlane}{type=symbols,
	sort={control},
	name={\ensuremath{y_{\text{lane},0}^\text{TV}}},
	description={?}
}
\newglossaryentry{dclose}{type=symbols,
	sort={control},
	name={\ensuremath{r_\text{close}}},
	description={?}
}
\newglossaryentry{dcloseft}{type=symbols,
	sort={control},
	name={\ensuremath{r^\text{FTP}_\text{close}}},
	description={?}
}
\newglossaryentry{msafe}{type=symbols,
	sort={control},
	name={\ensuremath{\varepsilon_\text{safe}}},
	description={?}
}
\newglossaryentry{fclose}{type=symbols,
	sort={control},
	name={\ensuremath{r_\text{close}}},
	description={?}
}
\newglossaryentry{ffclose}{type=symbols,
	sort={control},
	name={\ensuremath{f_\text{close}}},
	description={?}
}
\newglossaryentry{fcloseft}{type=symbols,
	sort={control},
	name={\ensuremath{r^\text{FTP}_\text{close}}},
	description={?}
}
\newglossaryentry{ffcloseft}{type=symbols,
	sort={control},
	name={\ensuremath{f^\text{FTP}_\text{close}}},
	description={?}
}

\newglossaryentry{dt}{type=symbols,
	sort={control},
	name={\ensuremath{T}}, 
	description={?}
}

\newglossaryentry{wcov}{type=symbols,
	sort={control},
	name={\ensuremath{\bm{\Sigma}^{\text{TV}}_{\bm{w}}}},
	description={?}
}
\newglossaryentry{wcov_t}{type=symbols,
	sort={control},
	name={\ensuremath{\tilde{\Sigma}^{\text{TV}}_{\bm{w}}}},
	description={?}
}
\newglossaryentry{wsenscov}{type=symbols,
	sort={control},
	name={\ensuremath{\Sigma^{\text{sens}}}},
	description={?}
}

\newglossaryentry{Xilane}{type=symbols,
	sort={control},
	name={\ensuremath{\mathcal{D}^{\text{lane}}}},
	description={?}
}

\newglossaryentry{UU}{type=symbols,
	sort={control},
	name={\ensuremath{\mathcal{U}}},
	description={?}
}
\newglossaryentry{XX}{type=symbols,
	sort={control},
	name={\ensuremath{\bm{\Xi}}},
	description={?}
}

\newglossaryentry{XXf}{type=symbols,
	sort={control},
	name={\ensuremath{\bm{\Xi}_\textnormal{f}}},
	description={?}
}
\newglossaryentry{XXsafe}{type=symbols,
	sort={control},
	name={\ensuremath{\bm{\Xi}_\textnormal{safe}}},
	description={?}
}

\newglossaryentry{Usafe}{type=symbols,
	sort={control},
	name={\ensuremath{\bm{U}_\textnormal{safe}}},
	description={?}
}
\newglossaryentry{Usafeh}{type=symbols,
	sort={control},
	name={\ensuremath{\bm{U}_{\textnormal{safe},h}}},
	description={?}
}
\newglossaryentry{Usmpc}{type=symbols,
	sort={control},
	name={\ensuremath{\bm{U}_\text{SMPC}}},
	description={?}
}
\newglossaryentry{Uft}{type=symbols,
	sort={control},
	name={\ensuremath{\bm{U}_\text{FTP}}},
	description={?}
}
\newglossaryentry{Uft2}{type=symbols,
	sort={control},
	name={\ensuremath{\bm{U}'
	_\text{FTP}}},
	description={?}
}
\newglossaryentry{Ubrake}{type=symbols,
	sort={control},
	name={\ensuremath{\bm{U}_\text{brake}}},
	description={?}
}

\newglossaryentry{Nsmpc}{type=symbols,
	sort={control},
	name={\ensuremath{N_\text{SMPC}}},
	description={?}
}

\newglossaryentry{Nft}{type=symbols,
	sort={control},
	name={\ensuremath{N_\text{FTP}}},
	description={?}
}

	\newglossaryentry{state}{type=symbols,
		sort={control},
		name={\ensuremath{\gvec{\xi}}},
		description={State vector}
	}
	
	\newglossaryentry{admissible_states}{type=symbols,
		sort={control},
		name={\ensuremath{\set{\Xi}_k}},
		description={Set of admissible states at time $k$}
	}
	
	\newglossaryentry{position_vector}{type=symbols,
		sort={control},
		name={\ensuremath{\gvec{\zeta}}},
		description={Position vector}
	}
	
	\newglossaryentry{rectangle_vertex}{type=symbols,
		sort={control},
		name={\ensuremath{\gvec{\zeta}^{\text{vert}}}},
		description={Vertex point of the safety rectangle on which the linear constraint is based}
	}
	
	\newglossaryentry{position}{type=symbols,
		sort={control},
		name={\ensuremath{x,y}},
		description={Longitudinal and lateral position}
	}
	
	\newglossaryentry{velo}{type=symbols,
		sort={control},
		name={\ensuremath{v_x, v_y}},
		description={Longitudinal and lateral velocity}
	}
	
	\newglossaryentry{control_input}{type=symbols,
		sort={control},
		name={\ensuremath{\vec{u}}},
		description={Control input vector}
	}
	
	\newglossaryentry{delta_control_input}{type=symbols,
		sort={control},
		name={\ensuremath{\triangle \vec{u}}},
		description={Rate of the control input vector \ensuremath{\vec{u}}}
	}
	
    \newglossaryentry{stacked_inputs}{type=symbols,
		sort={control},
		name={\ensuremath{U}},
		description={Sequence of input vectors}
	}
	
	\newglossaryentry{admissible_inputs}{type=symbols,
		sort={control},
		name={\ensuremath{\mathcal{U}_k}},
		description={Set of admissible inputs at time $k$}
	}
	
	\newglossaryentry{disturbance_Sens}{type=symbols,
		sort={control},
		name={\ensuremath{\vec{w}^{\text{Sens}}}},
		description={Disturbance vector corrupting the TV state measurement}
	}
	
	\newglossaryentry{disturbance_TV}{type=symbols,
		sort={control},
		name={\ensuremath{\vec{w}^{\text{TV}}}},
		description={Disturbance vector corrupting the TV inputs}
	}
	
	\newglossaryentry{disturbance_set}{type=symbols,
		sort={control},
		name={\ensuremath{\set{W}^{\text{Sens}}}},
		description={Compact, bounded set of measurement disturbances}
	}
	
	\newglossaryentry{Sens_covariance}{type=symbols,
		sort={control},
		name={\ensuremath{\set{\Sigma}^{\text{Sens}}}},
		description={Covariance matrix of the TV state measurement disturbance}
	}

	\newglossaryentry{TV_covariance}{type=symbols,
		sort={control},
		name={\ensuremath{\set{\Sigma}^{\text{TV}}_w}},
		description={Covariance matrix of the TV disturbance}
	}

	\newglossaryentry{sampling_time}{type=symbols,
		sort={control},
		name={\ensuremath{\triangle t}},
		description={Sampling time}
	}

	\newglossaryentry{vehicle_dim}{type=symbols,
		sort={control},
		name={\ensuremath{l_x, l_y}},
		description={Vehicle dimensions}
	}
	
	\newglossaryentry{lanewidth}{type=symbols,
		sort={control},
		name={\ensuremath{l_{\text{lane}}}},
		description={Lane width}
	}
	
	\newglossaryentry{lane_nr}{type=symbols,
		sort={control},
		name={\ensuremath{n_{\text{lane}}}},
		description={Number of lanes}
	}

	\newglossaryentry{rectangle_dim}{type=symbols,
		sort={control},
		name={\ensuremath{a,b}},
		description={Minimum safety rectangle dimensions}
	}
	
	\newglossaryentry{rectangle_length}{type=symbols,
		sort={control},
		name={\ensuremath{\tilde{a}}},
		description={Approximated length of the safety rectangle}
	}
	
	\newglossaryentry{delta_rectangle_dim}{type=symbols,
		sort={control},
		name={\ensuremath{\triangle a}},
		description={Velocity dependent longitudinal safety distance}
	}
	
	\newglossaryentry{delta_v_squared}{type=symbols,
		sort={control},
		name={\ensuremath{\triangle v^2}},
		description={Difference between the quadratic velocities of the EV and TV}
	}

    \newglossaryentry{approximation_c}{type=symbols,
		sort={control},
		name={\ensuremath{c}},
		description={Approximation parameter for the velocity dependent safety distance}
	}
	
	\newglossaryentry{steepness_parameter}{type=symbols,
		sort={control},
		name={\ensuremath{c_{\text{sd}}}},
		description={Parameter defining the steepness of the linear constraints}
	}
	
	\newglossaryentry{constraint_identification}{type=symbols,
		sort={control},
		name={\ensuremath{c_{\text{sc}}}},
		description={Parameter for the choice between regular and special case constraints}
	}
	
	\newglossaryentry{constraint}{type=symbols,
		sort={control},
		name={\ensuremath{d_k}},
		description={State constraint at time step $k$}
	}

    \newglossaryentry{pred_horizon}{type=symbols,
		sort={control},
		name={\ensuremath{N}},
		description={Length of the prediction horizon}
	}

    \newglossaryentry{prediction_error}{type=symbols,
		sort={control},
		name={\ensuremath{\vec{e}}},
		description={State prediction error}
	}
	
	\newglossaryentry{error_covariance}{type=symbols,
		sort={control},
		name={\ensuremath{\set{\Sigma}^e}},
		description={State prediction error covariance matrix}
	}
	
	\newglossaryentry{system_dynamics}{type=symbols,
		sort={control},
		name={\ensuremath{\Mat{A},\Mat{B},\Mat{K}}},
		description={Linear vehicle system dynamics}
	}

    \newglossaryentry{tolerance_level}{type=symbols,
		sort={control},
		name={\ensuremath{\kappa}},
		description={Tolerance level, defines the contour line of the bivariate Gaussian distribution}
	}
	
	\newglossaryentry{probability parameter}{type=symbols,
		sort={control},
		name={\ensuremath{\epsilon}},
		description={Probability parameter that defines the level of confidence with which a TV lies within the elliptic region around its predicted state}
	}

	\newglossaryentry{trajectory}{type=symbols,
		sort={control},
		name={\ensuremath{\chi_k^U}},
		description={Trajectory starting at $\gvec{\xi}_k$ and applying $U$}
	}
	
	\newglossaryentry{trajectory_solution}{type=symbols,
		sort={control},
		name={\ensuremath{\chi(k,\gvec{\xi}_0,U)}},
		description={State of a trajectory $\chi_0^U$ at time $k$}
	}
	
	\newglossaryentry{feasible_trajectories}{type=symbols,
		sort={control},
		name={\ensuremath{\Gamma_k}},
		description={Set of feasible trajectories $\chi_k^U$ resulting in $\set{\Xi}_f$}
	}
	
	\newglossaryentry{reachable_set}{type=symbols,
		sort={control},
		name={\ensuremath{\mathcal{R}_k^{\text{TV}}}},
		description={Set of states that can be reached by the TV at time $k$}
	}
	
	\newglossaryentry{set_of_initial_states}{type=symbols,
		sort={control},
		name={\ensuremath{\set{\Xi}^{\text{TV}}_0}},
		description={Set of possible initial TV states}
	}
	
	\newglossaryentry{safety_sets}{type=symbols,
		sort={control},
		name={\ensuremath{\set{\Xi}^{\text{occ}/\text{ex}/\text{safe}}}},
		description={Occupancy, excluded and safe set of states}
	}
	
	\newglossaryentry{inv_set}{type=symbols,
		sort={control},
		name={\ensuremath{\set{\Xi}^{\text{inv}},\set{\Xi}_{\text{f}}}},
		description={Safe control invariant sets with $\set{\Xi}_{\text{f}} \in \set{\Xi}^{\text{safe}}$}
	}
	
    \newglossaryentry{safety_rectangle}{type=symbols,
		sort={control},
		name={\ensuremath{\set{\Xi}^{\text{Rect}}}},
		description={Set of positions describing the SMPC safety rectangle}
	}

    \newglossaryentry{switching_velo}{type=symbols,
		sort={control},
		name={\ensuremath{v_{\text{s}}}},
		description={Switching velocity, defines the maximum engine power}
	}
	
	\newglossaryentry{max_velo_bounds}{type=symbols,
		sort={control},
		name={\ensuremath{v^{\text{up}}_x}},
		description={Maximum longitudinal velocity bound}
	}
	
	\newglossaryentry{velo_bounds}{type=symbols,
		sort={control},
		name={\ensuremath{v^{\text{low/up},A_i}_{x,N}}},
		description={Lower and upper velocity bounds at time $N$ of region $A_i$}
	}
	
	\newglossaryentry{sum_velo_bounds}{type=symbols,
		sort={control},
		name={\ensuremath{v^-, v^+}},
		description={Summarized velocity bounds for all regions $A_i$}
	}
	
    \newglossaryentry{velo_ex}{type=symbols,
		sort={control},
		name={\ensuremath{v^{\text{ex}}_{x,\text{min/max},N}}},
		description={Minimally and maximally possible velocities of the TV at time $N$}
	}
	
	\newglossaryentry{velo_regions}{type=symbols,
		sort={control},
		name={\ensuremath{A_i}},
		description={Region for which the velocity bounds $v^{\text{low/up},A_i}_{x,N}$ hold}
	}

	\newglossaryentry{transition}{type=symbols,
		sort={control},
		name={\ensuremath{\delta}},
		description={Transition width between inner and outer velocity bounds}
	}

	\newglossaryentry{cost}{type=symbols,
		sort={control},
		name={\ensuremath{J}},
		description={Average cost per step}
	}
	
	\newglossaryentry{quad_input_rate}{type=symbols,
		sort={control},
		name={\ensuremath{\mu_{\triangle ^2}, \sigma^2_{\triangle ^2}}},
		description={Mean and variance of the quadratic longitudinal input rate $\triangle u_x^2$}
	}

\newglossaryentry{vector}{type=notation,
	sort={vector},
	name={\ensuremath{\bm{z}_k}},
	description={Vector at timestep $i$}
}
\newglossaryentry{seqi}{type=notation,
	sort={seqi},
	name={\ensuremath{(\cdot)_{(1:k)}}},
	description={Combined variable up to $i$-th step if $(\cdot)^{(i)}$ exists; e.g.: $\bm{z}^{(1:i)}=\begin{bmatrix}\bm{z}^{(1)} & \cdots & \bm{z}^{(i)}\end{bmatrix}^\T$}
}	
\newglossaryentry{seqN}{type=notation,
	sort={seqN},
	name={\ensuremath{(\cdot)}},
	description={Combined variable for \gls{N}/\gls{Nmpc} steps if $(\cdot)^{(i)}$ exists; e.g.: $\bm{z}=\begin{bmatrix}\bm{z}^{(1)} & \cdots & \bm{z}^{(N)}\end{bmatrix}^\T$}
}	
\newglossaryentry{seqNi}{type=notation,
	sort={seqNi},
	name={\ensuremath{(\cdot)_k}},
	description={Combined variable for \gls{Nmpc} steps if $(\cdot)^{(i)}$ exists; e.g.: $\bm{z}_i=\begin{bmatrix}\bm{z}_i^{(i+1)} & \cdots & \bm{z}_i^{(i+N)}\end{bmatrix}^\T$}
}	
\newglossaryentry{matrix}{type=notation,
	sort={matrix},
	name={\ensuremath{Z}},
	description={Matrix}
}		
\newglossaryentry{rv}{type=notation,
	sort={rv},
	name={\ensuremath{\bm{\rv{Z}}_k}},
	description={Random vector at timestep $i$}
}	
\newglossaryentry{set}{type=notation,
	sort={set},
	name={\ensuremath{\set{Z}_k}},
	description={Set of admissable \gls{vector}}
}
\newglossaryentry{I}{type=notation,
	sort={I},
	name={\ensuremath{I_{k\times k}}},
	description={$k$-by-$k$ identity matrix}
}

\begin{document}

\title{Stochastic Model Predictive Control with a Safety Guarantee for Automated Driving: Extended Version}

\author{Tim~Br\"udigam,~
        Michael~Olbrich,~
        Dirk~Wollherr,~
        Marion~Leibold
\thanks{
The authors gratefully acknowledge the financial and scientific support by the BMW Group.}
\thanks{T. Br\"udigam, D. Wollherr, and M. Leibold are with the Chair of Automatic Control Engineering at the Technical University of Munich, Munich, Germany (email: \{tim.bruedigam;~dw;~marion.leibold\}@tum.de).

M. Olbrich is with the Department of Computer Science at the University of Augsburg, Augsburg, Germany (email: michael.olbrich@informatik.uni-augsburg.de).}
}

\maketitle

\begin{abstract}
Automated vehicles require efficient and safe planning to maneuver in uncertain environments. Largely this uncertainty is caused by other traffic participants, e.g., surrounding vehicles. Future motion of surrounding vehicles is often difficult to predict. Whereas robust control approaches achieve safe, yet conservative motion planning for automated vehicles, Stochastic Model Predictive Control (SMPC) provides efficient planning in the presence of uncertainty. Probabilistic constraints are applied to ensure that the maximal risk remains below a predefined level. However, safety cannot be ensured as probabilistic constraints may be violated, which is not acceptable for automated vehicles. Here, we propose an efficient trajectory planning framework with safety guarantees for automated vehicles. SMPC is applied to obtain efficient vehicle trajectories for a finite horizon. Based on the first optimized SMPC input, a guaranteed safe backup trajectory is planned using reachable sets. This backup is used to overwrite the SMPC input if necessary for safety. Recursive feasibility of the safe SMPC algorithm is proved. Highway simulations show the effectiveness of the proposed method regarding performance and safety.
\end{abstract}


\IEEEpeerreviewmaketitle

\section{Introduction}
\label{sec:introduction}

\vspace{-14.8cm}
\mbox{\small This~work~was~accepted~for~publication~in~the~IEEE~Transactions~on~Intelligent~Vehicles.}
\vspace{14.3cm}

\vspace{-14.7cm}
\mbox{\small The~accepted~version~is~available~at~https://doi.org/10.1109/TIV.2021.3074645.}
\vspace{14.3cm}

Within the past decades, research has made significant progress in the area of self-driving cars. Improvements in computer vision-based sensing and the use of this sensor data in control algorithms enable automated vehicles to detect and react to hazards in dynamic traffic and a constantly changing environment. A majority of road accidents is still caused by human errors, therefore, increasing the level of vehicle autonomy has great potential to reduce the overall number of accidents. Automated vehicles are especially relevant in critical situations where a significant number of human drivers is incapable of performing necessary maneuvers in time \cite{Keller20142}. 

The safety of fully automated vehicles depends on the ability of the vehicle control algorithm to handle uncertainty of other traffic participants and the environment. While there are various control methods to plan vehicle trajectories, \gls{MPC} has proved to be a suitable approach for automated vehicle trajectory planning \cite{LevinsonEtalThrun2011, KatrakazasEtalDeka2015}. \gls{MPC} iteratively solves an optimal control problem on a finite prediction horizon, based on minimizing a cost function, while considering constraints on the vehicle input and state. Uncertainties in the prediction model are addressed by \gls{RMPC} \cite{Mayne2014}.

\gls{RMPC} approaches were designed for trajectory planning in automated vehicles \cite{SolopertoEtalMueller2019,DixitEtalFallah2020}, however, robustly accounting for uncertainty yields conservative vehicle behavior. Conservatism resulting from robustly handling uncertainty in MPC is reduced by \gls{SMPC} \cite{Mesbah2016, FarinaGiulioniScattolini2016}, where robust constraints are reformulated into probabilistic constraints. This probabilistic reformulation enables optimistic trajectory planning in a majority of scenarios, but it also allows a small probability of constraint violation, i.e, a probability of collision for vehicles. Various types of \gls{SMPC} have been studied in the area of automated vehicles, showing the advantage of planning optimistic trajectories in a majority of scenarios \cite{CarvalhoEtalBorrelli2014, CesariEtalBorrelli2017}.

In comparison to \gls{SMPC}, trajectory planning based on reachability analysis provides formal safety guarantees \cite{Magdici2016, Althoff2018}. Here, worst-case predictions are obtained for other surrounding vehicles in order to plan fail-safe vehicle trajectories, referred to as \gls{FT}, which is closely related to \gls{RMPC}.

This paper is the extended version of \cite{BruedigamEtalLeibold2021a}. In this work, we tackle the challenge of planning efficient and safe trajectories for automated vehicles. We present a novel \gls{MPC} trajectory planner which combines the advantages of \gls{SMPC} and fail-safe trajectory planning for environments with uncertainty. A trajectory is planned with \gls{SMPC}, providing optimistic and efficient planning. In a regular setting, the first optimized \gls{SMPC} input is then applied to the vehicle and a new \gls{SMPC} optimal control problem is solved at the next time step with a shifted horizon. In addition to \gls{SMPC}, for every time step a fail-safe trajectory is planned, given the first optimized \gls{SMPC} input. The optimistic \gls{SMPC} input is only applied to the vehicle if it is still possible to find a fail-safe backup trajectory after having applied the first \gls{SMPC} input. This ensures that the efficient \gls{SMPC} trajectory is executed as long as a backup exists, therefore guaranteeing safety. The proposed method is referred to as \gls{SMPCFT}.

The contributions of this work are as follows. 
\begin{itemize}
\item Novel \gls{SMPCFT} method providing efficient and safe trajectory planning including lane change decisions.
\item Elaborate case differentiation for safety constraints.
\item Proof of recursive feasibility of the \gls{SMPCFT} method.
\item Simulation study with complex highway traffic situations.
\end{itemize}
We present the proposed \gls{SMPCFT} approach in detail, ensuring safety for the vehicle while exploiting the benefits of efficient \gls{SMPC} trajectory planning. The design of the \gls{SMPCFT} method guarantees recursive feasibility, i.e., if a solution exists at a time step, it is guaranteed that a solution also exists at the next time step. This property is necessary to ensure safe trajectory planning with \gls{MPC}. This extended version present a detailed case differentiation to generate safety constraints with respect to other surrounding vehicles, both for the \gls{SMPC} and the \gls{FT} optimal control problem. A simulation study of two complex scenarios demonstrates the benefits of optimistic trajectory planning in a regular highway scenario, while the ability of \gls{SMPCFT} to guarantee safety is shown in an emergency scenario.

The paper is structured as follows. In Section~\ref{sec:related_work} related work is presented. Then, the basics of \gls{SMPC} and \gls{FT} are introduced in Section~\ref{sec:preliminaries}, while the relevant vehicle models are presented in Section~\ref{sec:vehiclemodels}. The proposed \gls{SMPCFT} approach is derived in Section~\ref{sec:method}, and details for the respective \gls{SMPC} and \gls{FT} optimal control problems are provided in Section~\ref{sec:details}. The simulation results are presented in Section~\ref{sec:results}. A discussion and conclusive remarks are given in Section~\ref{sec:discussion} and Section~\ref{sec:conclusion}.

\subsection{Related Work}
\label{sec:related_work}

Trajectory planning for automated vehicles is a widely studied research area. There are various methods in non-\gls{MPC} related fields, such as using partially observable Markov decision processes (POMDP) \cite{HubmannEtalStiller2018} or reinforcement learning \cite{MirchevskaEtalBoedecker2018}. Learning based methods are also popular for autonomous racing \cite{RosoliaCarvalhoBorrelli2017, KabzanEtalZeilinger2019, WischnewskiBetzLohmann2019, StahlEtalLienkamp2019}. When considering automated road vehicles, planning trajectories with \gls{MPC} has the advantage of iteratively replanning the vehicle trajectory with constraints, accounting for a changing environment. 

Standard \gls{MPC} has been developed for cooperative adaptive cruise control, focusing on cooperative driving \cite{MasseraTerraWolf2017, KazemiEtalFallah2018}. \gls{MPC} is also designed specifically to plan trajectories for a single autonomous vehicle \cite{GutjahrGroellWerling2017} or for combined maneuver and trajectory planning \cite{YiEtalStiller2019}. \gls{MPC} was also combined with potential-field methods in order to avoid static and dynamic obstacles \cite{RasekhipourEtalLitkouhi2017}.

The main focus of this work is trajectory planning with robust or stochastic \gls{MPC}, as well as \acrfull{FT}. Fail-safe trajectory planning is defined as planning collision-free vehicle trajectories, accounting for any legal future motion of surrounding vehicles \cite{PekAlthoff2018}. For bounded uncertainties in real-world applications, \gls{FT} is applied based on finding worst-case sets. Combined with reachability analysis, formal safety guarantees are given \cite{Althoff2018}. The computation of these reachable sets is connected to control invariant sets in \gls{RMPC} as stated in \cite{Gullila2013}. An approach to include reachability analysis into \gls{MPC} is given in \cite{SchuermannKochdumperAlthoff2018}. 

In \cite{Althoff2013} a method is proposed to compute the set of all future locations possibly occupied by traffic participants. The remaining safe space is admissible to plan emergency trajectories. This \gls{FT} is presented in \cite{Magdici2016}. First, given the most likely motion of surrounding vehicles, an optimal trajectory is determined. Then, an emergency trajectory is connected to the last point of the optimal trajectory for which collision avoidance is still guaranteed. The safe space is determined by an over-approximated set of any possible future vehicle motion. The fail-safe trajectory is generated in such a way that the controlled vehicle comes to a standstill. In \cite{PekAlthoff2018} an \gls{FT} method is introduced which generates fail-safe trajectories in real-time. The method is tested in various simulations based on the CommonRoad benchmark framework \cite{AlthoffKoschiManzinger2017}. A motion planning framework is introduced in \cite{ManzingerPekAlthoff2020} which combines reachability analysis with optimization-based trajectory planning. In \cite{DixitEtalFallah2020} an \gls{RMPC} method is suggested which uses a combination of a potential field like function and reachability sets to obtain safe zones on the road. A further \gls{RMPC} method for collision avoidance with moving obstacles is presented in~\cite{SolopertoEtalMueller2019}.

\acrfull{SMPC} has been intensively studied in the context of automated vehicles. These works focus on the trade-off between risk and conservatism, defined by probabilistic constraints, so called chance constraints \cite{SchwarmNikolaou1999}. On the one hand, taking into account unlikely uncertainty realizations drastically reduces efficacy, on the other hand, planning too optimistically increases risk. A major challenge in \gls{SMPC} is reformulating the probabilistic chance constraint into a tractable constraint, which can be handled by a solver.

An \gls{SMPC} particle approach is shown in \cite{BlackmoreEtalWilliams2010} with a simple vehicle braking scenario, where particles approximate the uncertainty. An \gls{SMPC} trajectory planner for automated vehicles in the presence of fixed obstacles is presented in \cite{LenzEtalKnoll2015}. In \cite{CarvalhoEtalBorrelli2014} the environment is modeled by an Interacting Multiple Model Kalman filter. Given the most likely prediction for surrounding vehicles, a vehicle trajectory is then planned with \gls{SMPC} assuming Gaussian uncertainty. Varying risk parameters, denoting the level of accepted risk, are studied, illustrating the trade-off between risk and conservatism. In \cite{SuhChaeYi2018} an \gls{SMPC} lane change controller is presented, where the lane change risk is considered using predicted time-to-collision.

A different \gls{SMPC} approach is utilized in \cite{SchildbachBorrelli2015, CesariEtalBorrelli2017}, focusing on \gls{SCMPC} based on \cite{CampiGaratti2011, SchildbachEtalMorari2014}. In \gls{SCMPC} samples of the uncertainty are drawn, which must then satisfy the constraints to find a tractable chance constraint expression. Arbitrary probability distributions are handled by \gls{SCMPC}, while standard \gls{SMPC} usually requires Gaussian distributions to analytically reformulate the chance constraint.  While \cite{SchildbachBorrelli2015} focuses on simple lane change scenarios, the work is extended in \cite{CesariEtalBorrelli2017} and experimental results are presented. 

A combination of \gls{SMPC} and \gls{SCMPC} is given in \cite{BruedigamEtalWollherr2018b}, exploiting the individual advantages of \gls{SMPC} and \gls{SCMPC}. \gls{SCMPC} is used to cope with maneuver uncertainty of surrounding vehicles, while \gls{SMPC} addresses the maneuver execution uncertainty, described by a Gaussian probability distribution. A further approach to \gls{SMPC} is presented in \cite{BruedigamEtalLeibold2020c}, where a grid-based \gls{SMPC} method is applied to plan vehicle trajectories, based on occupancy grids \cite{SteyerTanzmeisterWollherr2018, SteyerEtalWollherr2019}.

In summary, \gls{SMPC} approaches provide efficient vehicle trajectories for the majority of uncertainty realizations in regular situations. However, for unlikely uncertainty realizations, safety issues occur.

In this work, the benefit of efficiently planning trajectories with \gls{SMPC} is combined with the safety guarantee of \gls{FT}. The \gls{FT} in this work is inspired by the ideas of \cite{Althoff2013,Magdici2016,PekAlthoff2018}. In the following, \gls{SMPC} and \gls{FT} are introduced. Then the proposed \gls{SMPCFT} method is derived in detail.

\subsection{Notation}
\label{sec:notation}

Regular letters denote scalars, bold lowercase letters indicate vectors, and bold uppercase letters are used for matrices, e.g., $a$, $\bm{a}$, $\bm{A}$, respectively. 
The Euclidean norm is denoted by~$\norm{.}_2$ and $\norm{\bm{z}}_{\bm{Z}} = \bm{z}^\top \bm{Z} \bm{z}$. The probability of an event is given by $\text{Pr}(.)$. 
The state at time step $h$ is denoted by $\bm{x}_h$. For \gls{MPC} optimal control problems the prediction step $k$ is used, i.e., $\bm{x}_k$. If clear from context, the time step $h$ is not specifically mentioned in the context of optimal control problems due to clarity. Estimated states are indicated by $\hat{\bm{x}}$.

\section{Preliminaries}
\label{sec:preliminaries}

In the following we briefly introduce the general \gls{MPC} optimal control problems for \gls{SMPC} and \gls{FT} in safety-critical applications. Details on \gls{SMPC} and \gls{FT}, which are relevant for application in automated driving, are then given in Section~\ref{sec:detailsSMPC} and Section~\ref{sec:detailsFST}, respectively.

\gls{MPC} iteratively solves an \gls{ocp} with a finite prediction horizon $N$ subject to input and state constraints. After solving the \gls{MPC} optimal control problem, only the first input $\bm{u}_0$ of the optimized input sequence $\bm{U} = \left(\bm{u}_0,\dots,\bm{u}_{N-1}\right)^\top$ is applied. At the next time step the updated \gls{MPC} optimal control problem is solved again. We distinguish between regular time steps $h$ and prediction steps $k$ within the \gls{MPC} \gls{ocp}. If clear from context, the time step $h$ is not specifically mentioned. In the following we only explicitly denote the prediction time step $k$. The current time step $h$, at which the \gls{MPC} optimal control problem is computed, is omitted to improve clarity.

\subsection{SMPC with Chance Constraints}
While standard \gls{MPC} considers hard constraints, this is problematic if uncertainties are present. Avoiding the worst-case realizations for uncertainties results in an overly conservative solution of the optimal control problem, which is unfavorable for automated driving. Hard constraints subject to uncertainty can be considered by chance constraints. This yields the \gls{SMPC} optimal control problem
\begin{IEEEeqnarray}{rll}
\IEEEyesnumber \label{eq:smpcOCP}
	V^* &= \min_{\bm{U}} \sum_{k=0}^{N-1} l(\bm{\xi}_k,\bm{u}_k) +V_{\text{f}}&\lr{\bm{\xi}_N} \IEEEyessubnumber \\
	\text{s.t. } & \bm{\xi}_{k+1} = \bm{f}\lr{\bm{\xi}_k,\bm{u}_k}& \IEEEyessubnumber \IEEEeqnarraynumspace \label{eq:smpcOCPdynamics}\\
	& \bm{\xi}_k \in \bm{\Xi}_k &~\forall k \in \{1,\dots,N\} , \IEEEyessubnumber\\
	& \bm{u}_k \in \mathcal{U}_k &~\forall k \in \{0,\dots,N-1\} ,\IEEEyessubnumber \IEEEeqnarraynumspace\\
	& \text{Pr}\left(\bm{\xi}_k \in \bm{\Xi}'_{k,\text{safe}}(\bm{w}) \right) \geq \beta &~\forall k \in \{1,\dots,N\} \IEEEyessubnumber \label{eq:smpcOCPcc}
\end{IEEEeqnarray}
with prediction step $k$, states $\bm{\xi}_k$, input sequence $\bm{U}= [\bm{u}_{0}, \bm{u}_{1}, ..., \bm{u}_{N-1}]^\top$, system dynamics $\bm{f}$, and the normally distributed, zero mean uncertainty $\bm{w} \sim \mathcal{N}\left(\bm{0}, \bm{\Sigma}^{w}\right)$ with covariance matrix $\bm{\Sigma}^{w}$. The cost function consists of the stage cost $l(\bm{\xi}_k,\bm{u}_k)$ and the terminal cost $V_{\text{f}}\lr{\bm{\xi}_N}$. States and inputs are bounded by the state and input constraint sets $\bm{\Xi}_k$ and $\mathcal{U}_k$, respectively, and the safety constraint $\bm{\Xi}'_{k,\text{safe}}(\bm{w})$ depends on the uncertainty $\bm{w}$. The probabilistic chance constraint is given by \eqref{eq:smpcOCPcc}. The safety constraint $\bm{\xi}_k \in \bm{\Xi}'_{k,\text{safe}}(\bm{w})$ is required to hold according to the risk parameter $\beta$, i.e., the probability of $\bm{\xi}_k \in \bm{\Xi}'_{k,\text{safe}}(\bm{w})$ must be larger for a larger risk parameter $\beta$. For $\beta < 1$ a certain constraint violation probability is therefore allowed.

The chance constraint \eqref{eq:smpcOCPcc} cannot be handled by a solver directly, but is required to be reformulated into a deterministic expression. Details on the reformulation are given in Section~\ref{sec:detailsSMPC}. 

In this work, uncertainty is considered regarding surrounding vehicles. The safe set $\bm{\Xi}'_{\text{safe}}(\bm{w})$ therefore depends on how the uncertainty $\bm{w}$ affects the surrounding vehicles. This is described in Section~\ref{sec:vehiclemodels}. Note that \gls{SMPC} optimal control problems can consider an expectation value in the cost function, however, this is omitted here as no expectation value will be necessary for the automated driving optimal control problem, i.e., the system dynamics \eqref{eq:smpcOCPdynamics} do not include uncertainty. In the simulation study in Section~\ref{sec:results} the uncertainty in the safe set $\bm{\Xi}'_{k,\text{safe}}(\bm{w})$ is governed by uncertainty in the behavior of surrounding vehicles.

\subsection{Fail-safe Trajectory Planning}
We also consider an \gls{MPC} optimal control problem for \gls{FT}, i.e., a fail-safe \gls{MPC} \gls{ocp}. In contrast to \gls{SMPC}, \gls{FT} considers the worst-case realizations of the uncertainty, resulting in safe, yet conservative optimized inputs. The \gls{FT} optimal control problem is given by
\begin{IEEEeqnarray}{rll}
\IEEEyesnumber \label{eq:ftOCP}
	V^* &= \min_{\bm{U}} \sum_{k=0}^{N-1} l(\bm{\xi}_k,\bm{u}_k) +V_{\text{f}}&\lr{\bm{\xi}_N} \IEEEyessubnumber \\
	\text{s.t. } & \bm{\xi}_{k+1} = \bm{f}\lr{\bm{\xi}_k,\bm{u}_k}& \IEEEyessubnumber \IEEEeqnarraynumspace \\
	& \bm{\xi}_k \in \bm{\Xi}_k &~\forall k \in \{1,\dots,N\} , \IEEEyessubnumber\\
	& \bm{u}_k \in \mathcal{U}_k &~\forall k \in \{0,\dots,N-1\} ,\IEEEyessubnumber \IEEEeqnarraynumspace\\
	& \bm{\xi}_k \in \bm{\Xi}_{k,\text{safe}}(\bm{w}) &~\forall k \in \{1,\dots,N-1\} \IEEEyessubnumber \label{eq:ftOCPcc}\\
	& \bm{\xi}_N \in \bm{\Xi}_{N,\text{safe}}(\bm{w}) & \IEEEyessubnumber \label{eq:ftOCPterm}
\end{IEEEeqnarray}
which is similar to the \gls{SMPC} problem \eqref{eq:smpcOCP}. However, the safety constraint \eqref{eq:ftOCPcc} is not a chance constraint, as in \eqref{eq:smpcOCPcc}, but a hard constraint. In contrast to \gls{SMPC}, for \gls{FT} the safe set $\bm{\Xi}_{k,\text{safe}}(\bm{w})$ is constructed based on reachability analysis to ensure formal safety guarantees. This is strongly connected to the computation of invariant sets in \gls{RMPC}. In addition to constraint \eqref{eq:ftOCPcc}, a terminal constraint \eqref{eq:ftOCPterm} is required, which ensures that the terminal prediction state $\bm{\xi}_N$ allows to remain in a safe state beyond the prediction horizon. Based on this safe terminal set $\bm{\Xi}_{N,\text{safe}}(\bm{w})$ it is guaranteed that there exist system inputs $\bm{u}_{k^+}$ with $k^+>N$ resulting in safe states $\bm{\xi}_{k^+}$. Note that the safe set $\bm{\Xi}'_{k,\text{safe}}(\bm{w})$ in \eqref{eq:smpcOCPcc} is not necessarily computed in the same way as the safe set $\bm{\Xi}_{k,\text{safe}}(\bm{w})$ in \eqref{eq:ftOCPcc},\eqref{eq:ftOCPterm}. Details on the \gls{FT} optimal control problem are provided in Section~\ref{sec:detailsFST}.

\section{Vehicle Models}
\label{sec:vehiclemodels}

\gls{MPC} requires a system model for the controlled vehicle, known as the \gls{ev}, and surrounding vehicles, referred to as \glspl{tv}, in order to predict future states within the optimal control problem. 

\subsection{Ego Vehicle Model}
\label{sec:egovehicle}
We use a kinematic bicycle model to predict the \gls{ev} states on a finite horizon, as suggested in~\cite{KongEtalBorrelli2015}. 
The continuous-time system is given by
\begin{IEEEeqnarray}{rl}
\IEEEyesnumber \label{eq:bicyclemodel}
\dot{s} &= v \cos(\phi + \alpha), \IEEEyessubnumber \label{eq:bm_s}\\
\dot{d} &= v \sin(\phi + \alpha), \IEEEyessubnumber \label{eq:bm_d} \\
\dot{\phi} &= \frac{v}{l_\text{r}}\sin{\alpha}, \IEEEyessubnumber \label{eq:bm_phi} \\
\dot{v} &= a, \IEEEyessubnumber \label{eq:bm_a} \\
\alpha &= \arctan \lr{\frac{l_\text{r}}{l_\text{r}+l_\text{f}} \tan \delta}, \IEEEyessubnumber \label{eq:bm_alpha}
\end{IEEEeqnarray}
where $l_\text{r}$ and $l_\text{f}$ are the distances from the vehicle center of gravity to the rear and front axles, respectively. The state vector is $\bm{\xi} = [s,d,\phi,v]^\top$ and the input vector is $\bm{u} = [a,\delta]^\top$. The vehicle velocity is given by $v$, acceleration and steering angle are denoted by $a$ and $\delta$, respectively. We consider the longitudinal position $s$ of the vehicle along the road, the lateral vehicle deviation $d$ from the centerline of the right lane, and the orientation $\phi$ of the vehicle with respect to the road. The nonlinear vehicle model \eqref{eq:bicyclemodel} is summarized as $\dot{\bm{\xi}} = \bm{f}^\text{c}\lr{\bm{\xi},\bm{u}}$.

Each \gls{MPC} optimal control problem is initialized with a linearization of the nonlinear prediction model \eqref{eq:bicyclemodel} around the current vehicle state $\bm{\xi}^* = \bm{\xi}_0$ and the input $\bm{u}^* = [0,0]^\top$. 
Selecting a non-zero reference input $\bm{u}^*$ often results in large differences $\Delta \bm{u} = \bm{u}_k - \bm{u}^*$ for prediction steps far ahead, increasing the inaccuracy of the linearization. The linearized continuous-time vehicle model is then given by
\begin{IEEEeqnarray}{c}
\dot{\bm{\xi}}^* + \Delta \dot{\bm{\xi}} = \bm{f}^\text{c} \lr{\bm{\xi}^*, \bm{0}} + \bm{A}_\text{l} \lr{\bm{\xi} - \bm{\xi}^*}  + \bm{B}_\text{l} \bm{u}  \label{eq:lin_evsys} \IEEEeqnarraynumspace
\end{IEEEeqnarray}
with the Jacobi matrices
\begin{IEEEeqnarray}{c}
\bm{A}_\text{l} = \left. \left[ \frac{\partial \bm{f}^\text{c}}{\partial \bm{\xi}} \right] \right|_{\lr{\bm{\xi}^*, \bm{u}^*}}, ~~~~ \bm{B}_\text{l} = \left. \left[ \frac{\partial \bm{f}^\text{c}}{\partial \bm{u}} \right] \right|_{\lr{\bm{\xi}^*, \bm{u}^*}}. \IEEEeqnarraynumspace
\end{IEEEeqnarray}

A discrete-time model is required for \gls{MPC}, therefore the linearized prediction model \eqref{eq:lin_evsys} is discretized with sampling time \gls{dt}. This yields the discrete states $\gls{xik}~=~[\gls{sk},\gls{dk},\gls{phik},\gls{vk}]^\top$ and inputs $\gls{uk} = [\gls{ak},\gls{deltak}]^\top$ for prediction step $k$, as well as the linearized, discretized system
\begin{IEEEeqnarray}{rl}
\IEEEyesnumber \label{eq:d_evsys}
\bm{\xi}_{k+1} &= \bm{\xi}_0 + \gls{dt} \bm{f}^\text{c} \lr{\bm{\xi}_0, \bm{0}} + \bm{A}_\text{d} \lr{\bm{\xi}_k - \bm{\xi}_0}  + \bm{B}_\text{d} \bm{u}_k \IEEEyessubnumber \IEEEeqnarraynumspace \\ 
&= \bm{f}^\text{d}\lr{\bm{\xi}_0, \bm{\xi}_k, \bm{u}_k}  \IEEEyessubnumber 
\end{IEEEeqnarray}
where $\bm{A}_\text{d}$ and $\bm{B}_\text{d}$ are matrices of the linearized system obtained from $\bm{A}_\text{l}$, $\bm{B}_\text{l}$ with zero-order hold. The nonlinear term $\bm{f}^\text{c} \lr{\bm{\xi}^*, \bm{u}^*}$ in \eqref{eq:lin_evsys} is approximated by a forward Euler method since $\bm{\xi}_0$ is known. The linearized, discretized matrices $\bm{A}_\text{d}$ and $\bm{B}_\text{d}$ are given in Appendix \ref{sec:appendixd}. In the following, for $k=0$ in \eqref{eq:d_evsys}, i.e., $\bm{\xi}_k = \bm{\xi}_0$, the argument $\bm{\xi}_0$ is only mentioned once, i.e., $\bm{f}^\text{d}\lr{\bm{\xi}_0, \bm{\xi}_0, \bm{u}_0}$ is abbreviated as $\bm{f}^\text{d}\lr{\bm{\xi}_0, \bm{u}_0}$.

The following sections derive an \gls{SMPC} method and constraints to avoid collisions with surrounding vehicles. However, even if no other vehicles are present, certain constraints are required. Acceleration and steering angle are bounded by
\begin{IEEEeqnarray}{rrl}
\IEEEyesnumber \label{eq:constraints_input}
\bm{u}_\text{min} \leq & \gls{uk} &\leq \bm{u}_\text{max} \IEEEyessubnumber \label{eq:con_u}\\
\Delta\bm{u}_\text{min} \leq & \Delta\gls{uk} &\leq \Delta\bm{u}_\text{max} \IEEEyessubnumber \label{eq:con_du}
\end{IEEEeqnarray}
with $\Delta \bm{u}_{k+1} = \bm{u}_{k+1} - \gls{uk}$ and $\bm{u}_\text{max} = [a_\text{max}, \delta_\text{max}]^\top$, $\bm{u}_\text{min} = [a_\text{min}, \delta_\text{min}]^\top$. Further, road and velocity constraints are considered, resulting in
\begin{IEEEeqnarray}{rl}
\IEEEyesnumber \label{eq:constraints_state}
&\gls{dk} \in \gls{Xilane} \IEEEyessubnumber \label{eq:constraints_lane} \\
&0 \leq v_k \leq v_\text{max} \IEEEyessubnumber \label{eq:constraints_velocity}
\end{IEEEeqnarray}
where \gls{Xilane} represents road boundaries and $v_\text{max}$ is the maximal velocity. Negative velocities are not allowed, i.e., $\gls{vk} \geq 0$.

In the following we refer to input constraints by the set of admissible inputs \gls{UU} and state constraints are denoted by the set of admissible states \gls{XX}.

\subsection{Target Vehicle Model}
\label{sec:targetvehicle}
In order to avoid collisions, the \gls{ev} is also required to predict the future states of surrounding \glspl{tv}. The prediction model for the \glspl{tv} used by the \gls{ev} is a linear, discrete-time point-mass model given by
\begin{IEEEeqnarray}{rl}
\IEEEyesnumber \label{eq:tvmodel}
\bm{\xi}^{\text{TV}}_{k+1} &= \bm{A} \gls{xitv}+\bm{B} \gls{utv} \IEEEyessubnumber \label{eq:tvsys}\\
\gls{utv} &= \gls{tutv} + \gls{wtv}\IEEEyessubnumber \label{eq:tvu}
\end{IEEEeqnarray}
where $\gls{xitv} = [\gls{xtv}, \gls{vxtv}, \gls{ytv}, \gls{vytv}]^\top$ is the \gls{tv} state with longitudinal position and velocity \gls{xtv}, \gls{vxtv} and lateral position and velocity \gls{ytv}, \gls{vytv}. The linear \gls{tv} model allows to propagate the uncertainty, which is necessary for the \gls{MPC} approach in the following sections. The \gls{tv} model used in this work is only one possible option. Other linear \gls{tv} prediction models can be utilized. 

The system and input matrices are
\begin{IEEEeqnarray}{c}
\bm{A} = \begin{bmatrix} 1 & \gls{dt} & 0 & 0  \\
											0 & 1 & 0 & 0 \\
											0 & 0 & 1 & \gls{dt} \\
											0 & 0 & 0 & 1 \end{bmatrix},
	~\bm{B} = \begin{bmatrix} 0.5 \gls{dt}^2 & 0 \\
											\gls{dt} & 0 \\ 0 & 0.5 \gls{dt}^2 \\ 0 & \gls{dt} \end{bmatrix}
  \label{eq:tv_AB}
\end{IEEEeqnarray}
with sampling time \gls{dt}. The \gls{tv} input consists of a feedback controller \gls{tutv} and a perturbation on the input, which is assumed to be an independent, identically distributed disturbance vector \gls{wtv}. This setup assumes that the \gls{tv} is following a given reference while deviations are allowed. The \gls{tv} feedback controller is given by
\begin{IEEEeqnarray}{rl}
\IEEEyesnumber
\gls{tutv} &= \bm{K} \lr{\gls{xitv} - \gls{xitvref}}, \IEEEyessubnumber\label{eq:tv_feedback}\\
\bm{K} &= \begin{bmatrix} 0 & k_{12} & 0 & 0 \\ 0 & 0 & k_{21} & k_{22} \end{bmatrix} \IEEEyessubnumber
\end{IEEEeqnarray}
with the \gls{tv} reference \gls{xitvref}. The feedback matrix $\bm{K}$ is obtained by a linear quadratic control strategy. 
As we assume that the \gls{tv} follows a reference velocity in $x$-direction instead of reference $x$-positions, the \gls{tv} input does not need to directly affect \gls{xtv}. The state \gls{xtv} is only used as a measured quantity. 
If the \gls{tv} input computed by \eqref{eq:tv_feedback} exceeds the limits  $\bm{u}^\text{TV}_\text{max} = [a_\text{max}, a_{y,\text{max}}]^\top$ and $\bm{u}^\text{TV}_\text{min} = [a_\text{min}, a_{y,\text{min}}]^\top$, summarized as $\mathcal{U}^\text{TV}$, the \gls{tv} inputs are bounded to satisfy $\mathcal{U}^\text{TV}$.

We assume that \gls{wtv} is subject to a Gaussian distribution with zero mean and covariance matrix \gls{wcov}, which is denoted by $\gls{wtv} \sim \mathcal{N}\lr{0,\gls{wcov}}$. 
We also consider sensor noise in the measurement of the \gls{tv} state, i.e.,
\begin{IEEEeqnarray}{c}
\gls{hxitv0} = \gls{xitv0} + \gls{wsens}  \label{eq:w_sens}
\end{IEEEeqnarray}
where \gls{hxitv0} is the measured initial state of the \gls{tv} by the \gls{ev}. The sensor noise $\gls{wsens} = [w^\text{sens}_{0,x}, w^\text{sens}_{0,v_x}, w^\text{sens}_{0,y}, w^\text{sens}_{0,v_y}]^\top$ is assumed to be a truncated Gaussian noise with $\gls{wsens} \sim \mathcal{N}\lr{0,\gls{wsenscov}}$ and $\gls{wsens} \in \gls{Wsens}$, where \gls{Wsens} is a compact, convex and bounded set.

\section{Stochastic Model Predictive Control with Safety Guarantee}
\label{sec:method}

\gls{SMPC} and fail-safe trajectory planning both have their individual advantages, i.e., efficient trajectories in an uncertain environment and guaranteed safe motion planning, respectively. However, both methods come with certain disadvantages. \gls{SMPC} allows a predefined probability of constraint violation and thus potentially collisions, while \gls{FT} can result in overly conservative trajectories. In the following we will present a combined \gls{SMPC} and \gls{FT} framework, \gls{SMPCFT}, which exploits advantages of both methods to enable efficient and safe trajectories for autonomous vehicles. This section introduces the general setup of the \gls{SMPCFT} framework and gives a proof for recursive feasibility. The subsequent Sections~\ref{sec:detailsSMPC}~and~\ref{sec:detailsFST} provide details on the individual \gls{SMPC} and \gls{FT} algorithms, respectively, which are required for the combined \gls{SMPCFT} approach.

\subsection{\gls{SMPCFT} Method}

Before presenting the \gls{SMPCFT} method, we need to define requirements for a safe ego vehicle state as well as a safe input sequence $\gls{Usafe} = [\bm{u}_{\text{safe},0}, \bm{u}_{\text{safe},1}, ..., \bm{u}_{\text{safe},m}]^\top$ with $m+1$ individual inputs. Note that $m$ is not directly related to the \gls{MPC} prediction horizon.

\begin{definition}[Safe State]
\label{def:safestate}
The state of an ego vehicle, fully located in one lane, is considered to be safe if there is no lateral vehicle motion, i.e., $\phi = 0$, and if the ego vehicle velocity is lower than the velocity of the target vehicle in front on the same lane (or if the ego vehicle velocity is zero). The set of safe states is indicated by $\bm{\Xi}_\textnormal{safe}$. 
\end{definition}

\begin{definition}[Safe Input Sequence]
\label{def:safeinputsequence}
An input sequence \gls{Usafe} is considered safe if consecutively applying all elements of \gls{Usafe} results in a state trajectory that avoids collisions and eventually leads to zero velocity.
\end{definition}

The definition of safe states and safe input sequences results in assumptions for \glspl{tv}.

\begin{assumption}[Traffic Rules]
\label{ass:trafficrules}
The target vehicles adhere to the traffic rules that are assumed for the prediction of future target vehicle behavior.
\end{assumption}

\begin{assumption}[Vehicle Deceleration]
\label{ass:veh_decelerate}
The maximum absolute value of the ego vehicle deceleration is at least as large as the maximum absolute value of the target vehicle deceleration.
\end{assumption}

Given a safe \gls{ev} state, there exists a safe input sequence \gls{Usafe}, consisting of deceleration and zero steering, which results in an \gls{ev} zero velocity state in the current \gls{ev} lane, i.e., zero velocity in $x$-direction and $y$-direction. This is based on the assumptions that the other \glspl{tv} adhere to traffic rules and that the \gls{tv} deceleration is not larger than the \gls{ev} deceleration. \Glspl{tv} behaving against traffic rules cannot be reliably accounted for by any prediction and the deceleration assumption is necessary to avoid colliding with a braking \gls{tv} in front.

In the following we focus on the \gls{SMPCFT} optimal control problems which are solved at each time step $h$. Within the \gls{ocp} only the prediction steps $k$ are relevant. For clarity we therefore omit explicitly denoting the current time step $h$ in the following. Within the \gls{ocp} the current \gls{ev} state at time step $h$ and prediction step $k=0$ is denoted by \gls{xi0}.

At the initialization of each \gls{ocp}, the current \gls{ev} state \gls{xi0} and the current \gls{tv} state \gls{xitv0} are known to the \gls{ev}. Additionally, 
a safe input sequence \gls{Usafe} is available from the \gls{SMPCFT} problem solved at the previous time step. Later we will focus on obtaining a safe input sequence for the \gls{SMPCFT} iteration at the next time step, given the safe input sequence of the current time step.

Now, the \gls{SMPCFT} method is derived, consisting of two parts, \gls{SMPC} and \gls{FT}, i.e., at every time step an \gls{SMPC} \gls{ocp} and an \gls{FT} \gls{ocp} are solved. The general idea is that the first input $\bm{u}_{\text{SMPC},0}$ of the \gls{SMPC} input sequence $\gls{Usmpc}= [\bm{u}_{\text{SMPC},0}, ..., \bm{u}_{\text{SMPC},N-1}]^\top$ must only be applied if, based on the first \gls{SMPC} input $\bm{u}_{\text{SMPC},0}$, a fail-safe trajectory can be found as described below. Compared to regular \gls{SMPC} methods, this approach guarantees that applying the optimistic \gls{SMPC} input $\bm{u}_{\text{SMPC},0}$ does not lead to unsafe behavior. Therefore, at each time step one \gls{SMPC} \gls{ocp} and one \gls{FT} \gls{ocp} is solved. The algorithm outline is shown in Figure~\ref{fig:smpcft}.
\begin{figure}
\centering
\includegraphics[width = \columnwidth]{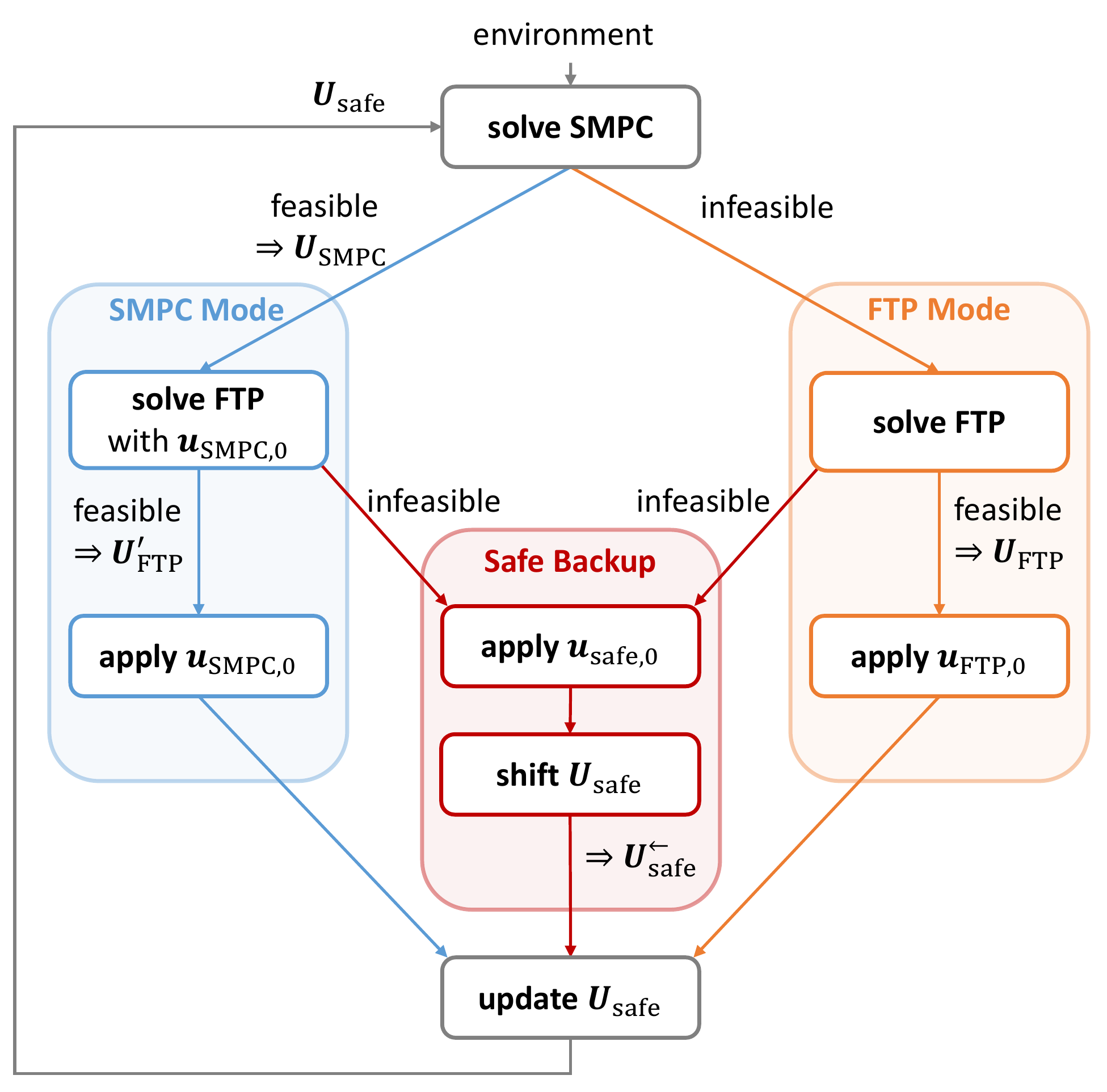}
\caption{\gls{SMPCFT} procedure for each time step. Blue shows the ideal mode with an applied \gls{SMPC} input, orange represents the safe alternative mode with an applied \gls{FT} input, and red indicates an infeasible \gls{FT} problem, which requires applying a safe backup input.} 
\label{fig:smpcft}
\end{figure}

\subsubsection{\gls{SMPC}}

In the first part of \gls{SMPCFT} an \gls{SMPC} problem is solved on a finite horizon \gls{Nsmpc}, yielding the input sequence $\gls{Usmpc} = [\bm{u}_{\text{SMPC},0}, ..., \bm{u}_{\text{SMPC},\gls{Nsmpc}-1}]^\top$. This \gls{SMPC} optimization takes into account the uncertain environment and constraints due to other traffic participants, i.e., target vehicles. 
Collision constraints are formulated as chance-constraints, based on a probabilistic \gls{tv} prediction. Therefore, the planned \gls{SMPC} trajectory provides an efficient and optimistic future trajectory for the \gls{ev}, as it is not required to avoid collision with \glspl{tv} for worst-case scenarios. However, chance-constraints allow a small probability of collision in the future, depending on the predefined \gls{SMPC} risk parameter.

\subsubsection{\gls{FT}}

The second part of \gls{SMPCFT} is based on \gls{FT} to ensure that the planned \gls{ev} trajectory remains safe. First, a worst-case \gls{tv} prediction is performed. Then, a fail-safe \gls{MPC} problem on a finite horizon \gls{Nft} is solved, resulting in an input sequence $\gls{Uft}= [\bm{u}_{\text{FTP},0}, ..., \bm{u}_{\text{FTP},\gls{Nft}-1}]^\top$. The fail-safe trajectory is required to avoid collision with the worst-case \gls{tv} prediction and after applying the full fail-safe input sequence \gls{Uft}, the terminal state $\boldsymbol{\xi}_{\gls{Nft}}$ must be a safe state according to Definition \ref{def:safestate}. 

The exact \gls{FT} formulation depends on the feasibility of the \gls{SMPC} \gls{ocp}.

\paragraph{Feasible SMPC} If the \gls{SMPC} \gls{ocp} yields a solution, \gls{FT} is used to decide whether applying the first \gls{SMPC} input $\bm{u}_{\text{SMPC},0}$ is safe.  Therefore, an \gls{FT} \gls{ocp} is formulated starting with the \gls{ev} state obtained by applying the first \gls{SMPC} input $\bm{u}_{\text{SMPC},0}$, i.e., the initial state for the \gls{FT} \gls{ocp} is 
\begin{IEEEeqnarray}{c}
\bm{\xi}'_0 = \bm{f}\lr{\bm{\xi}_0, \bm{u}_{\text{SMPC},0}}
\end{IEEEeqnarray}
with $\bm{f}\lr{\bm{\xi}_0,\bm{u}_{\text{SMPC},0}}$ 
according to \eqref{eq:d_evsys}. 

If feasible, the \gls{FT} \gls{ocp} yields a fail-safe input sequence \gls{Uft2}, based on $\bm{\xi}'_0$. Therefore, the first element $\bm{u}_{\text{SMPC},0}$ of the \gls{SMPC} input sequence is applied safely, as shown by the blue path in Figure~\ref{fig:smpcft}. The resulting new safe input sequence is given by
\begin{IEEEeqnarray}{rl}
\IEEEyesnumber \label{eq:usafe1}
\gls{Usafe} &= [\gls{Uft2}, \gls{Ubrake}] \IEEEyessubnumber \\
\gls{Ubrake} & = \begin{bmatrix} \begin{bmatrix} a_\text{min}\\ 0 \end{bmatrix}, \begin{bmatrix} a_\text{min} \\ 0 \end{bmatrix}, \dots \end{bmatrix} \IEEEyessubnumber \label{eq:ubrake}
\end{IEEEeqnarray}
where $a_\text{min}$ is the maximal deceleration and \gls{Ubrake} is a braking sequence to bring the \gls{ev} to a standstill. The safe input sequence \gls{Usafe} ensures a safe state after the full fail-safe input sequence \gls{Uft2} was applied and then initiates braking in order to reach zero velocity. Note that $a_\text{min}$ is only applied in \gls{Ubrake} until a standstill is reached, subsequently no deceleration is applied.

\paragraph{Infeasible SMPC} If the \gls{SMPC} \gls{ocp} is infeasible, the \gls{FT} optimal control problem is solved with initial state \gls{xi0} for the \gls{FT} \gls{ocp}. If an \gls{FT} solution \gls{Uft} is found, the first element of \gls{Uft}, i.e., $\bm{u}_{\text{FTP},0}$, is applied, as indicated by the orange path in Figure~\ref{fig:smpcft}. The updated safe input sequence follows from
\begin{IEEEeqnarray}{c}
\gls{Usafe} = [\bm{U}_{\text{FTP},1:\gls{Nft}}, \gls{Ubrake}]  \label{eq:usafe2}
\end{IEEEeqnarray}
with \gls{Ubrake} according to \eqref{eq:ubrake} where
\begin{IEEEeqnarray}{c}
\bm{U}_{\text{FTP},1:\gls{Nft}} = [\bm{u}_{\text{FTP},1}, ..., \bm{u}_{\text{FTP},\gls{Nft}-1}]
\end{IEEEeqnarray}
consists of all input elements of \gls{Uft} except the first input $\bm{u}_{\text{FTP},0}$.

\subsubsection{Infeasible \gls{FT}} \label{sec:safeinputsequence}
In case of an infeasible \gls{FT} \gls{ocp} no new input is generated at the current time step $h$. However, by definition the safe input sequence obtained at the previous time step $h-1$ remains safe for the current time step $h$. Therefore, in case that no solution exists to the \gls{FT} \gls{ocp}, the first element of the still valid, safe input sequence $\gls{Usafe}$ is applied, which is denoted by $\bm{u}_{\text{safe},0}$. This procedure is highlighted in red in Figure~\ref{fig:smpcft}.

Continuously applying the elements of \gls{Usafe} results in a safe trajectory according to Definition~\ref{def:safeinputsequence}. If the \gls{FT} \gls{ocp} remains infeasible for consecutive time steps, multiple subsequent input elements of a single safe input sequence are potentially applied until the \gls{FT} \gls{ocp} becomes feasible again.

This procedure requires shifting \gls{Usafe} after each \gls{SMPCFT} iteration where the \gls{FT} \gls{ocp} was infeasible, i.e., the first input element $\bm{u}_{\text{safe},0}$ of \gls{Usafe} was applied. The shifted updated input sequence is obtained by
\begin{IEEEeqnarray}{c}
\gls{Usafe}^{\leftarrow} = \gls{Usafe} 
\begin{bmatrix}
\bm{0}_{m}  \\ \bm{I}_{m}
\end{bmatrix} = [\bm{u}_{\text{safe},1}, \bm{u}_{\text{safe},2}, ..., \bm{u}_{\text{safe},m}] \IEEEeqnarraynumspace \label{eq:shift}
\end{IEEEeqnarray}
with $\gls{Usafe} \in \mathbb{R}^{2 \times (m+1)}$, identity matrix $\bm{I}_{m} \in \mathbb{R}^{m \times m}$, and $\bm{0}_{m} \in \mathbb{R}^{1\times m}$. The shifted safe input sequence $\gls{Usafe}^{\leftarrow}$ consists of all elements of \gls{Usafe} except the already applied input $\bm{u}_{\text{safe},0}$.

Then, the safe input sequence is updated at the end of the \gls{SMPCFT} iteration by selecting
\begin{IEEEeqnarray}{c}
\gls{Usafe} = \gls{Usafe}^{\leftarrow},\label{eq:usafe3}
\end{IEEEeqnarray}
which initializes the safe input sequence for the next \gls{SMPCFT} iteration.

\subsubsection{Summary of \gls{SMPCFT}}

Within the \gls{SMPCFT} method four cases are considered. These cases are summarized in the following.
\paragraph{\gls{SMPC} and \gls{FT} feasible} The first \gls{SMPC} input $\bm{u}_{\text{SMPC},0}$ is applied and a new safe input sequence \gls{Usafe} is obtained according to \eqref{eq:usafe1}.

\paragraph{\gls{SMPC} infeasible and \gls{FT} feasible} The first \gls{FT} input $\bm{u}_{\text{FTP},0}$ is applied and a new safe input sequence \gls{Usafe} is obtained according to \eqref{eq:usafe2}.

\paragraph{\gls{SMPC} feasible and \gls{FT} infeasible} No new input sequence is obtained. The first input element of the safe input sequence $\bm{u}_{\text{safe},0}$ is applied. The safe input sequence \gls{Usafe} remains valid for the next time step and is updated according to~\eqref{eq:usafe3}. 

\paragraph{\gls{SMPC} infeasible and \gls{FT} infeasible} As in the previous case, no new input sequence is obtained. The input $\bm{u}_{\text{safe},0}$ is applied and $\gls{Usafe}$ is generated based on \eqref{eq:usafe3} for the next time step.

In summary, the \gls{SMPC} solution is applied as long as a fail-safe backup trajectory exists. Safety is guaranteed by solving an \gls{FT} \gls{ocp}, based on the first \gls{SMPC} input. In cases where the \gls{FT} \gls{ocp} is infeasible, the safe input sequence of the previous time step is still valid. Following this procedure, in regular cases the efficient \gls{SMPC} inputs are applied, resulting in good performance, while \gls{FT} guarantees safety for all possible cases, including rare events.

\subsection{Recursive Feasibility}

A disadvantage of various \gls{SMPC} algorithms is that recursive feasibility of the optimization problem cannot be guaranteed. In this section recursive feasibility of the \gls{SMPCFT} method is proved, i.e., if the optimization problem can be solved at step $h$, it can also be solved at step $h+1$ for all $h \in \mathbb{N}$. In this section it is necessary to denote the time step $h$. The safe input sequence updated at time step $h$ is denoted by $\bm{U}_{\textnormal{safe},h}$.

\begin{definition}[Safe Feasible Trajectory]
Let there exist a safe set \gls{XXsafe} and let \gls{XXf} be a control invariant set. Let $\gls{chiUk} = [\bm{\xi}_h, ..., \bm{\xi}_{h+N}]$ denote a trajectory starting at initial state $\bm{\xi}_h$ at time step $h$ with $N$ trajectory steps obtained by applying the input sequence $\bm{U}_h = [\bm{u}_h, ..., \bm{u}_{h+N-1}]$ with $\bm{\xi}_{h+1} = \bm{f}\lr{\bm{\xi}_h,\bm{u}_h}$. Then, the set $\Gamma_h$ of safe feasible trajectories, leading into the control invariant set \gls{XXf}, is defined as
\begin{IEEEeqnarray}{rl}
&\Gamma_h = \IEEEnonumber\\ &\setdef[\gls{chiUk}]{\bm{\xi}_{h+i} \in \gls{XXsafe},~i \in \{0,\dots,N\},~\bm{\xi}_{h+N} \in \gls{XXf}}.  \IEEEeqnarraynumspace \label{eq:safe_traj} 
\end{IEEEeqnarray}
\end{definition}
A safe feasible trajectory satisfies all constraints given by \gls{XXsafe} and ends in the control invariant set \gls{XXf}.

\begin{assumption}[System Models]
\label{ass:sysmodel}
The ego vehicle system models \eqref{eq:bicyclemodel} and \eqref{eq:d_evsys} correspond to the dynamics of the real system. The target vehicle model \eqref{eq:tvmodel} represents an over-approximation of the real target vehicle dynamics.
\end{assumption}

Here, over-approximation means that the possible states reachable with the \gls{tv} model include all possible states obtained with the real \gls{tv} dynamics.

\begin{assumption}[Initial Safe Input Sequence]
\label{ass:FTinitial}
At the initial time step $h=0$ the initial ego vehicle state 
is safe and there exists a known initial safe input sequence $\bm{U}_{\textnormal{safe},\textnormal{init}}$, such that $\gls{chiX}_0^{\bm{U}_{\textnormal{safe},\textnormal{init}}}$ is a safe feasible trajectory, i.e., $\gls{chiX}_0^{\bm{U}_{\textnormal{safe},\textnormal{init}}} \in \Gamma_0$.
\end{assumption}

We can now show recursive feasibility of the proposed method.

\begin{theorem}
\label{t:recfeas}
Let Assumptions \ref{ass:sysmodel} and \ref{ass:FTinitial} hold. Then for the \gls{SMPCFT} approach there exists a feasible trajectory $\gls{chiUk} \in \Gamma_h$ that is guaranteed to be safe at all time steps $h \in \mathbb{N}$.
\end{theorem}

\begin{proof} 
The derivation of the proof is given in Appendix~\ref{sec:appendixa}. 
\end{proof}

Note that the worst-case behavior of the \glspl{tv} depends on the traffic rules. Therefore, safety and recursive feasibility of the \gls{SMPCFT} method can only be guaranteed if surrounding \glspl{tv} adhere to the underlying traffic rules, as stated in Assumption~\ref{ass:trafficrules}. However, no specific traffic rules are required to prove Theorem~\ref{t:recfeas}. 
\section{Trajectory Planning Algorithms}
\label{sec:details}
The two \gls{MPC} optimal control problems, \gls{SMPC} and \gls{FT}, are solved independently. In the following, the respective optimal control problems are derived, specifically focusing on the safety constraints.

\subsection{Stochastic Model Predictive Control}
\label{sec:detailsSMPC}
\gls{SMPC} solves an optimal control problem with chance constraints, accounting for \gls{tv} uncertainty, depending on a risk factor $\beta$. First, a safety area is defined around each predicted \gls{tv} state which accounts for the \gls{ev} and \gls{tv} shape. Then, this safety area is increased to account for \gls{tv} uncertainty, given a predefined risk parameter. Eventually, a linear constraint is generated for each \gls{tv}, depending on the positioning of the \gls{ev} and the \gls{tv}.

\subsubsection{Deterministic Target Vehicle Prediction}
\label{sec:smpcTVpred}
For \gls{SMPC} a simple \gls{tv} prediction is applied, representing the most likely \gls{tv} behavior with $\gls{wtv} = \bm{0}$, i.e., $\gls{utv} = \gls{tutv}$. It is assumed that the current \gls{tv} maneuver continues for the prediction horizon \gls{Nsmpc}. Therefore, \gls{tv} model \eqref{eq:tvmodel} is applied where the \gls{tv} reference \gls{xitvref} depends on the current \gls{tv} maneuver. The reference velocity $v_{x,\text{ref},k}^\text{TV}$ is set to the current \gls{tv} velocity $v_{x,0}^\text{TV}$. The \gls{tv} reference lateral velocity is chosen to be $v_{y,\text{ref},k}^\text{TV} = 0$. The reference lateral position $y_{\text{ref},k}^\text{TV}$ is the current \gls{tv} lane center. A new reference lane is selected if part of the \gls{tv} shape lies in this adjacent lane and the lateral velocity moves the \gls{tv} towards this adjacent lane. If both of these requirements are fulfilled, the adjacent lane center is selected as the lateral position reference.

\subsubsection{Target Vehicle Safety Area}
\label{sec:smpc_rectangle}
Collisions with \glspl{tv} are avoided by ensuring the necessary distance between the \gls{ev} and \gls{tv}. Here, a safety rectangle around the \gls{tv} is defined. While it is possible to choose other shapes, e.g., ellipses as in \cite{BruedigamEtalWollherr2018b}, rectangles allow to easily generate linear constraints, as described in Section~\ref{sec:constraintgeneration}. 

The safety rectangle with length $a_\text{r}$ and width $b_\text{r}$ is illustrated in Figure~\ref{fig:safetyrectangle}.
\begin{figure}
\centering
\includegraphics[width = 0.53\columnwidth]{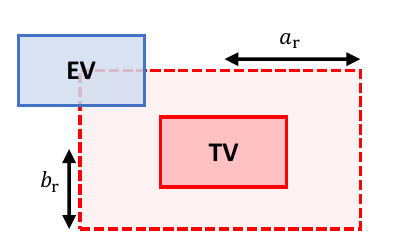}
\caption{Target vehicle safety rectangle.}
\label{fig:safetyrectangle}
\end{figure}
 In order to ensure that the vehicle shapes do not intersect, the vehicle centers need to be distanced at least by the vehicle length $l_\text{veh}$ and width $w_\text{veh}$. For the safety rectangle width this yields 
\begin{IEEEeqnarray}{c}
b_\text{r} = w_\text{veh} + \gls{msafe}
\end{IEEEeqnarray}
where $\gls{msafe}$ is a possible additional safety margin.

Calculating the safety rectangle length $a_\text{r}$ requires a velocity dependent part $\tilde{a}_\text{r}\lr{\bm{\xi}, \bm{\xi}^\text{TV}}$, compensating for a potential velocity difference between the \gls{ev} and the \gls{tv}, resulting in
\begin{IEEEeqnarray}{c}
a_\text{r} = l_\text{veh} + \gls{msafe} + \tilde{a}_\text{r}\lr{\bm{\xi}, \bm{\xi}^\text{TV}}.
\end{IEEEeqnarray}
The velocity dependent part $\tilde{a}_\text{r}$ needs to account for the difference in traveled distance between the \gls{ev} and \gls{tv} if both vehicles initiate maximal braking, e.g., in an emergency braking scenario. Here, in addition to Assumption~\ref{ass:veh_decelerate}, zero reaction time is assumed. The traveled distance $\Delta x$ of a vehicle until standstill is described by 
\begin{IEEEeqnarray}{rl}
\IEEEyesnumber
\Delta x \lr{t_\text{stop}} &= v_{x} t_\text{stop} + 0.5 a_{\text{min}} \lr{t_\text{stop}}^2 \nonumber \\
&= - \frac{1}{a_{\text{min}}} \lr{v_{x}}^2  \IEEEyessubnumber \\
t_\text{stop} &= - \frac{v_{x}}{a_{\text{min}}} \IEEEyessubnumber
\end{IEEEeqnarray}
with maximal longitudinal deceleration $a_{\text{min}}$, time to standstill $t_\text{stop}$, given the initial \gls{ev} and \gls{tv} velocity $v$ and $v^\text{TV}_{x}$, respectively. Based on the difference in traveled distance
\begin{IEEEeqnarray}{c}
\Delta x^\text{EV} \lr{t_\text{stop}^\text{EV}} - \Delta x^\text{TV} \lr{t_\text{stop}^\text{TV}} = - \frac{1}{a_{\text{min}}} \lr{v^2 - \lr{v_{x}^\text{TV}}^2   }, \IEEEeqnarraynumspace
\end{IEEEeqnarray}
assuming similar maximal deceleration for the \gls{ev} and \gls{tv}, the velocity dependent safety distance is obtained by
\begin{IEEEeqnarray}{c}
\tilde{a}_\text{r}\lr{\bm{\xi}, \bm{\xi}^\text{TV}} = - \frac{1}{a_{\text{min}}} \max \left\lbrace 0, \lr{v^2 - \lr{v_{x}^\text{TV}}^2   }  \right\rbrace
\end{IEEEeqnarray}
where the $\max$-operator ensures that the safety rectangle length does not decrease for $v_{x}^\text{TV} > v$.

For the \gls{SMPC} optimal control problem the safety rectangle length and width is calculated for prediction time step $k$, based on the \gls{tv} prediction \gls{xitv} described in Section~\ref{sec:smpcTVpred}. However, only the initial \gls{ev} state $\bm{\xi}_0$ is considered in the velocity depended part $\tilde{a}_\text{r}$. This is necessary in order to generate linear safety constraints. While it would be possible to use predicted \gls{ev} states, these would yield nonlinear constraints. The resulting safety rectangle parameters are 
\begin{IEEEeqnarray}{rl}
\IEEEyesnumber \label{eq:ab_rectangle}
b_{\text{r},k} &= w_\text{veh} + \gls{msafe} \IEEEyessubnumber \\
a_{\text{r},k} &= l_\text{veh} + \gls{msafe} + \tilde{a}_\text{r}\lr{\bm{\xi}_0, \gls{xitv}}. \IEEEyessubnumber
\end{IEEEeqnarray}

\subsubsection{Chance Constraint Reformulation}
\label{sec:cc}
The \gls{tv} safety rectangle given by \eqref{eq:ab_rectangle} only considers the deterministic prediction, but it does not account for \gls{tv} uncertainty. In the following the safety rectangle is enlarged, based on a chance constraint depending on the \gls{tv} uncertainty and a predefined risk parameter $\beta$. The chance constraint, similar to \eqref{eq:smpcOCPcc}, is given by 
\begin{IEEEeqnarray}{c}
\text{Pr}\left(\bm{\xi}_k \in \bm{\Xi}'_{k,\text{safe}}\lr{\bm{w}_k^\text{TV}} \right) \geq \beta \label{eq:cc_ev}
\end{IEEEeqnarray}
where the safe set $\bm{\Xi}'_{k,\text{safe}}\lr{\bm{w}_k^\text{TV}}$ for the \gls{ev} state depends on the previously defined safety rectangle parameters of~\eqref{eq:ab_rectangle} and the \gls{tv} uncertainty $\bm{w}_k^\text{TV}$. In other words, the previously defined safety area is now enlarged to account for \gls{tv} uncertainty. A large risk parameter reduces risk by generating a large safety area around the \gls{tv}.

The chance constraint \eqref{eq:cc_ev} cannot be solved directly. In the following, a deterministic approximation is determined for this probabilistic expression, inspired by other \gls{SMPC} approaches \cite{CarvalhoEtalBorrelli2014, BruedigamEtalWollherr2018b}.

According to \eqref{eq:tvmodel} the \gls{tv} state follows
\begin{IEEEeqnarray}{c}
\bm{\xi}^{\text{TV}}_{k+1} = \bm{A} \gls{xitv} + \bm{B}  \bm{K} \lr{\gls{xitv} - \gls{xitvref}}  + \bm{B} \gls{wtv},
\end{IEEEeqnarray}
while the predicted \gls{tv} state is given by 
\begin{IEEEeqnarray}{c}
\hat{\bm{\xi}}^{\text{TV}}_{k+1} = \bm{A} \hat{\bm{\xi}}^{\text{TV}}_{k} + \bm{B}  \bm{K} \lr{\hat{\bm{\xi}}^{\text{TV}}_{k} - \gls{xitvref}}, \label{eq:tvpred}
\end{IEEEeqnarray}
yielding the prediction error
\begin{IEEEeqnarray}{c}
\bm{e}_k = \hat{\bm{\xi}}^{\text{TV}}_{k} - \gls{xitv}.
\end{IEEEeqnarray}

The \gls{tv} prediction \eqref{eq:tvpred} is now split into a deterministic and a stochastic part
\begin{IEEEeqnarray}{c}
\hat{\bm{\xi}}^{\text{TV}}_{k+1} = \bm{\xi}^{\text{TV}}_{k+1} + \lr{\bm{A} + \bm{B}\bm{K}}\bm{e}_{k} - \bm{B} \gls{wtv} = \bm{\xi}^{\text{TV}}_{k+1} + \bm{e}_{k+1} \IEEEeqnarraynumspace
\end{IEEEeqnarray}
which results in the prediction error update
\begin{IEEEeqnarray}{c}
\bm{e}_{k+1} = \lr{\bm{A} + \bm{B}\bm{K}}\bm{e}_{k} - \bm{B} \gls{wtv}.
\end{IEEEeqnarray}

Given the sensor noise \gls{wsens} according to~\eqref{eq:w_sens}, the initial error follows $\bm{e}_0 \sim \mathcal{N}\lr{0, \bm{\Sigma}_0^{\bm{e}}}$ with $\bm{\Sigma}_0^{\bm{e}} = \bm{\Sigma}^{\text{sens}}$. As both the \gls{tv} uncertainty, with covariance matrix \gls{wcov}, and the sensor noise are assumed to be Gaussian distributions, a recursive computation of the prediction error covariance matrix $\bm{\Sigma}_{k}^{\bm{e}}$ is possible, yielding
\begin{IEEEeqnarray}{c}
\bm{\Sigma}_{k+1}^{\bm{e}} = \bm{B} \gls{wcov} \bm{B}^\top + \lr{\bm{A} + \bm{B}\bm{K}} \bm{\Sigma}_{k}^{\bm{e}} \lr{\bm{A} + \bm{B}\bm{K}}^\top. \label{eq:errcov_propagation}
\end{IEEEeqnarray}

Based on the prediction error covariance matrix $\bm{\Sigma}_{k}^{\bm{e}}$, the \gls{tv} safety rectangle is increased. Given a predefined \gls{SMPC} risk parameter $\beta$, the aim is to find a region around the predicted \gls{tv} state which contains the true \gls{tv} state with probability $\beta$. This then allows to consider the probabilistic safety constraint \eqref{eq:cc_ev} as a deterministic substitute constraint. As the \gls{tv} safety rectangle only considers positions, we define the reduced error $\tilde{\bm{e}}_k = [e_{x,k}, e_{y,k}]^\top$ with the reduced covariance matrix 
\begin{IEEEeqnarray}{c}
\tilde{\bm{\Sigma}}_{k}^{\bm{e}} = 
\begin{bmatrix}
\sigma^2_{x,k} & 0 \\ 0 & \sigma^2_{y,k}
\end{bmatrix}
\end{IEEEeqnarray}
with variances $\sigma^2_{x,k}$ and $\sigma^2_{y,k}$ for the longitudinal and lateral \gls{tv} position, corresponding to the first and third diagonal element of $\bm{\Sigma}_{k}^{\bm{e}}$.

\begin{lemma}
The reduced error covariance matrix $\tilde{\bm{\Sigma}}_{k}^{\bm{e}}$ corresponding to the position coordinates is obtained from $\bm{\Sigma}_{k}^{\bm{e}}$ by omitting the correlation to the respective velocities.
\end{lemma}

\begin{proof}
The proof is given in Appendix~\ref{sec:appendixc}. 
\end{proof}

The reduced error covariance matrix $\tilde{\bm{\Sigma}}_{k}^{\bm{e}}$ is now used to enlarge the safety rectangle to account for uncertainty. Note that the error covariance matrix $\bm{\Sigma}_{k}^{\bm{e}}$ is still required to compute~\eqref{eq:errcov_propagation}.

The bivariate Gaussian distribution described by $\tilde{\bm{\Sigma}}_{k}^{\bm{e}}$ with mean $\bm{\mu} = [\mu_{x},\mu_{y}]^\top = \bm{0}$ consists of independent random variables for longitudinal and lateral position. This allows to find a confidence region around the predicted \gls{tv} state mean, bounded by an ellipsoidal isoline enclosing the highest density region as illustrated in Figure~\ref{fig:isoline}. 
\begin{figure}
\centering
\includegraphics[width = 0.8\columnwidth]{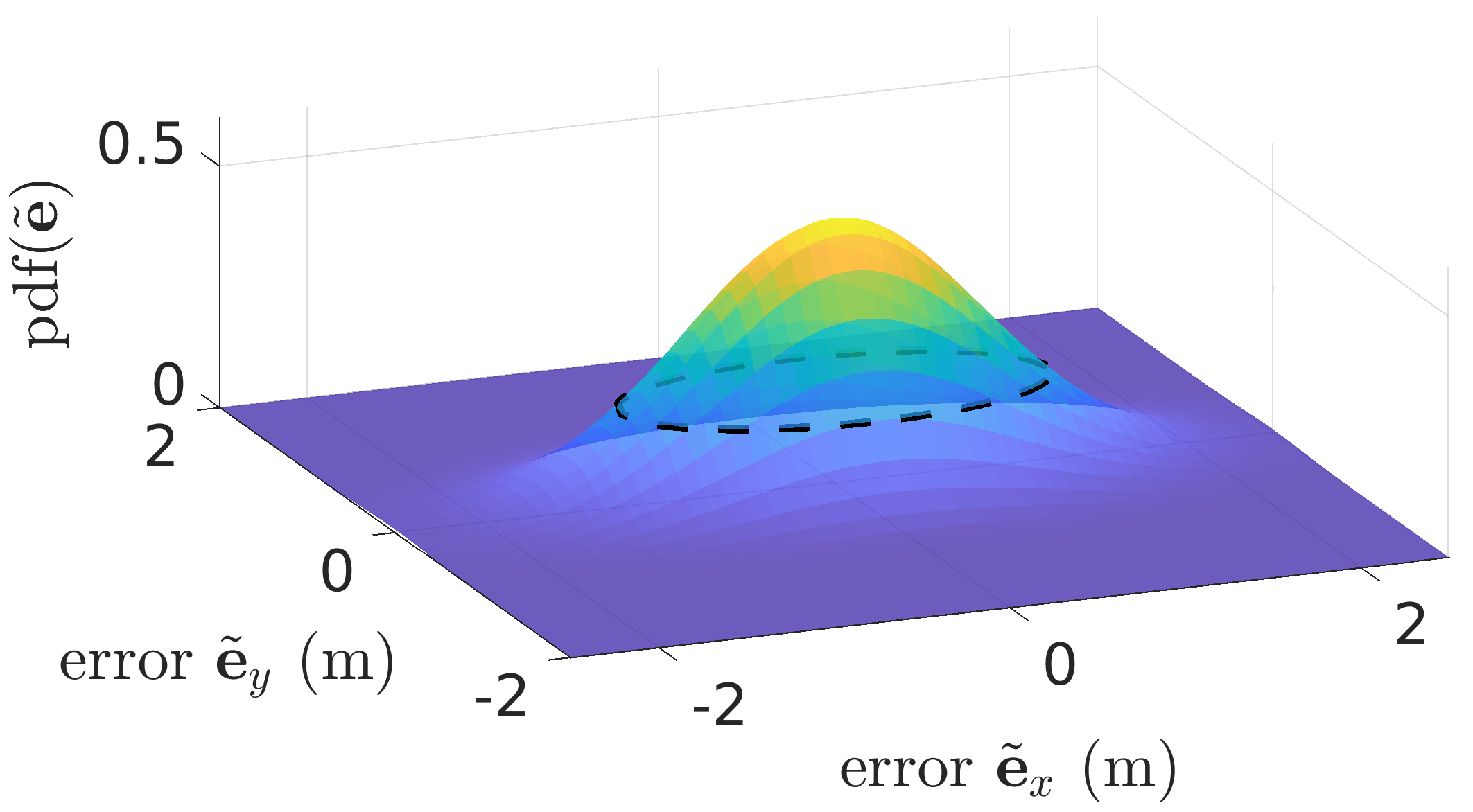}
\caption{Exemplary bivariate Gaussian probability distribution function of the prediction error $\tilde{\bm{e}}$, including an isoline (dotted black line).}
\label{fig:isoline}
\end{figure}
The aim is to find an isoline which contains the prediction error with a probability according to risk parameter $\beta$. The isoline ellipse equation is denoted by 
\begin{IEEEeqnarray}{rl}
\IEEEyesnumber \label{eq:iso_ellipse}
\lr{\tilde{\bm{e}}_k - \bm{\mu}}^\top \lr{\tilde{\bm{\Sigma}_{k}}^{\bm{e}}}^{-1}\lr{\tilde{\bm{e}}_k - \bm{\mu}} &= \kappa \IEEEyessubnumber \\
\frac{\lr{e_{x,k}-\mu_{x}}^2}{\sigma_{x,k}^2} + \frac{\lr{e_{y,k}-\mu_{y}}^2}{\sigma_{y,k}^2} &= \kappa  \IEEEyessubnumber
\end{IEEEeqnarray}
with 
tolerance level $\kappa$. 
The tolerance level $\kappa$ is determined based on the \textit{chi-squared distribution} $\chi^2_n (1-\beta)$, given the risk parameter $\beta$ and 
the number of degrees of freedom $n$. In this case, $n=2$ as the reduced error $\tilde{\bm{e}}_k$ consists of two elements. It then follows that
\begin{IEEEeqnarray}{c}
\kappa = \chi^2_2 (1-\beta),
\end{IEEEeqnarray}
which ensures that the probability of the true \gls{tv} state not lying within the isoline is 
$100(1-\beta)\%$, where the isoline is defined by the tolerance level $\kappa$. The ellipse semi-major and semi-minor axes are then given by 
\begin{IEEEeqnarray}{rl}
\IEEEyesnumber \label{eq:ellipse_axes}
e_{x,k,\kappa}= \sigma_{x,k} \sqrt{\kappa}  \IEEEyessubnumber \\
e_{y,k,\kappa}= \sigma_{y,k} \sqrt{\kappa}. \IEEEyessubnumber
\end{IEEEeqnarray}
These ellipse parameters are now used to increase the \gls{tv} safety rectangle.

\begin{remark}
Instead of using the \textit{chi-squared distribution}, in this case the tolerance level can also be obtained by $\kappa = -2 \ln{\beta}$.
\end{remark}

While an ellipse, according to \eqref{eq:iso_ellipse}, describes the desired confidence region, the constraint generation method used in this work requires a rectangular \gls{tv} safety area. We therefore over-approximate the ellipse by a rectangle. In order to include this uncertainty consideration in the rectangle parameters $a_{\text{r},k}$ and $b_{\text{r},k}$ of \eqref{eq:ab_rectangle}, the rectangle parameters are increased based on the ellipse semi-major axis $e_{x,k,\kappa}$ and semi-minor axis $e_{x,k,\kappa}$, resulting in 
\begin{IEEEeqnarray}{rl}
\IEEEyesnumber \label{eq:ab_rectangle_cc}
b_{\text{r},k} &= w_\text{veh} + \gls{msafe} + e_{y,k,\kappa} \IEEEyessubnumber \\
a_{\text{r},k} &= l_\text{veh} + \gls{msafe} + \tilde{a}_\text{r}\lr{\bm{\xi}_0, \gls{xitv}} + e_{x,k,\kappa}. \IEEEyessubnumber
\end{IEEEeqnarray}
The updated safety rectangle parameters are now utilized to generate the safety constraints for the \gls{SMPC} optimal control problem.

\subsubsection{\gls{SMPC} Constraint Generation}
\label{sec:constraintgeneration}
Given the safety rectangles for each \gls{tv}, linear constraints to avoid collisions can be defined for each prediction step and for each \gls{tv}. Each linear constraint is of the form
\begin{IEEEeqnarray}{c}
0 \geq q_y\lr{\bm{\xi}_0, \gls{xitv}} d_k + q_x\lr{\bm{\xi}_0, \gls{xitv}} s_k + q_t\lr{\bm{\xi}_0, \gls{xitv}} \label{eq:c_smpc}
\end{IEEEeqnarray}
where $q_y$ and $q_x$ are the coefficients for the \gls{ev} states $d_k$ and $s_k$, respectively, and $q_t$ is the intercept. The coefficients $q_y$, $q_x$, and $q_t$ of the linear constraint depend on the current \gls{ev} state $\bm{\xi}_0$ and the predicted mean \gls{tv} states \gls{xitv}. This results in multiple constraint generation cases, where the major cases are displayed in Figure~\ref{fig:smpc_constraintcases}. 
\begin{figure}
\centering
\includegraphics[width = 0.9\columnwidth,trim=0 2 0 1]{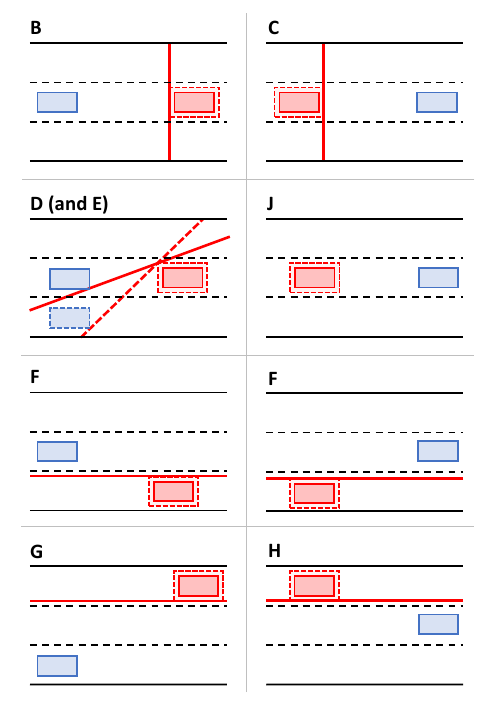}
\caption{Selected constraint generation cases for SMPC. Driving direction is from left to right. The \gls{ev} and \gls{tv} are shown in blue and red, respectively. The dashed red line represents the sfaety area around the \gls{tv}.}
\label{fig:smpc_constraintcases}
\end{figure}
The cases are distinguished based on the initial vehicle configuration at the beginning of the optimal control problem, i.e., $k=0$. As mentioned in Section~\ref{sec:smpc_rectangle}, while the predicted \gls{tv} state $\bm{\xi}^\text{TV}_k$ at prediction step $k$ is considered to build the constraint given a specific case, only the initial \gls{ev} state $\bm{\xi}_0$ is considered in order to allow generating linear constraints. We now introduce a brief, but not necessarily complete overview of constraint cases that are considered, which is an extension to the cases in \cite{CesariEtalBorrelli2017}. A complete overview of cases, requirements, and constraint parameters $q_x$, $q_y$, $q_t$ from \eqref{eq:c_smpc} is found in Appendix~\ref{sec:appendixb}. 

We first consider the case where the \gls{tv} is far away from the \gls{ev}, i.e., a longitudinal distance $\gls{dxEVTV} = s_0 - \gls{xtv0}$ which is larger than $r_\text{lar}$. Therefore, no constraint is generated. 
\begin{itemize}
\item[A)] $|\gls{dxEVTV}|  \geq r_\text{lar}$: no constraint 
\end{itemize}
Then, the remaining cases are addressed with a longitudinal distance $|\gls{dxEVTV}|  < r_\text{lar}$. Further cases are distinguished given how close the \gls{tv} is positioned to the \gls{ev}, based on a function \gls{ffclose}, which is further explained in Appendix~\ref{sec:appendixb}. For simplicity, \gls{ffclose} is assumed to be a constant \gls{dclose} here. In case of a vehicle distance $| \gls{dxEVTV}|  > \gls{fclose}$, a vertical constraint is used behind the \gls{tv} (case B) or in front of the \gls{tv} (case C). 
\begin{itemize}
\item[B)] $-\lr{\gls{dxEVTV}}  > \gls{fclose}$: vertical constraint behind \gls{tv} 
\item[C)] $\lr{\gls{dxEVTV}}  > \gls{fclose}$: vertical constraint in front of \gls{tv} 
\end{itemize}
If the \gls{ev} is located closely behind the \gls{tv}, i.e., $s_0 < \gls{xtv0}$ and $-\lr{\gls{dxEVTV}}  \leq \gls{fclose}$, the constraint depends on the respective current \gls{ev} and \gls{tv} lanes \gls{evlane} and \gls{tvlane}. In general, the plan is to overtake \glspl{tv} on the left side. If the vehicles are in the same lane (case D), an inclined constraint is applied allowing the \gls{ev} to switch to a lane left of the \gls{tv}. The constraint slope is bounded to a horizontal or vertical line. If a \gls{tv} is on an adjacent lane to the left of the \gls{ev}, i.e., $\gls{evlane} + w_\text{lane} = \gls{tvlane}$ with lane width $w_\text{lane}$, an inclined (bounded) constraint is also applied (case E).
\begin{itemize}
\item[D)] $\gls{evlane} = \gls{tvlane}$, $-\lr{\gls{dxEVTV}}  \leq \gls{fclose}$: inclined constraint connecting \gls{ev} and \gls{tv} 
\item[E)] $\gls{evlane} + w_\text{lane} = \gls{tvlane}$, $\lr{\gls{xtv0} - s_0}  \leq \gls{fclose}$: inclined constraint connecting \gls{ev} and \gls{tv} 
\end{itemize}
If the \gls{tv} is on a lane to the right of the \gls{ev}, i.e., $\gls{evlane} > \gls{tvlane}$ (case F), a horizontal constraint is employed at the \gls{tv} safety rectangle. If the \gls{tv} is in front of the \gls{tv} and at least two lanes to the left of the \gls{ev}, i.e., $\gls{evlane} + 2w_\text{lane} \leq \gls{tvlane}$ (case G), or behind the \gls{ev} and at least one lane to the left of the \gls{tv} (case H), i.e., $\gls{evlane} + w_\text{lane} \leq \gls{tvlane}$, again a horizontal constraint is employed.
\begin{itemize}
\item[F)] $\gls{evlane} > \gls{tvlane}$, $| \gls{dxEVTV}|  \leq \gls{fclose}$: horizontal constraint left of the \gls{tv} 
\item[G)] $\gls{evlane} + 2w_\text{lane} \leq \gls{tvlane}$, $-\lr{\gls{dxEVTV}}  \leq \gls{fclose}$: horizontal constraint right of the \gls{tv} 
\item[H)] $\gls{evlane} + w_\text{lane} \leq \gls{tvlane}$, $\lr{\gls{dxEVTV}} \leq \gls{fclose}$: horizontal constraint right of the \gls{tv} 
\end{itemize}
Finally, if the \gls{ev} is positioned in front of the \gls{tv} on the same lane, it is assumed that the \gls{tv} keeps its distance to the \gls{ev}. Therefore, no constraint is generated (case J).
\begin{itemize}
\item[J)] $\gls{evlane} = \gls{tvlane}$, $\lr{\gls{dxEVTV}} \leq \gls{fclose}$: no constraint 
\end{itemize}

The presented cases are now used to formulate safety constraints in the \gls{SMPC} \gls{ocp}.

\subsubsection{\gls{SMPC} Optimal Control Problem}
With the definition of the safety constraints, the deterministic optimal control problem replacing the \gls{SMPC} problem is given by
\begin{IEEEeqnarray}{rl}
\IEEEyesnumber \label{eq:smpc}
	V^* &= \min_{\bm{U}} \sum_{k=1}^{\gls{Nsmpc}} \norm{\Delta \bm{\xi}_k}_{\bm{Q}} + \norm{\bm{u}_{k-1}}_{\bm{R}} + \norm{\Delta\bm{u}_{k-1}}_{\bm{S}}
	 \IEEEyessubnumber \IEEEeqnarraynumspace \\
	\text{s.t. } & \bm{\xi}_{k+1} = \bm{f}^\text{d}\lr{\bm{\xi}_0, \bm{\xi}_k, \bm{u}_k} \IEEEyessubnumber \IEEEeqnarraynumspace \label{eq:smpc_dynamics}\\
	& \bm{\xi}^\text{TV}_{k+1} = \bm{A} \gls{xitv}+\bm{B} \gls{tutv} \IEEEyessubnumber \IEEEeqnarraynumspace \label{eq:smpc_tv_dynamics}\\
	& \bm{\xi}_k \in \gls{XX} ~~~~\forall k \in \{1,\dots,\gls{Nsmpc}\} , \IEEEyessubnumber\\
	& \bm{u}_k \in \mathcal{U} ~~~\forall k \in \{0,\dots,\gls{Nsmpc}-1\} ,\IEEEyessubnumber \IEEEeqnarraynumspace\\
	& 0 \geq q_y\lr{\bm{\xi}_0, \gls{xitv}} d_k + q_x\lr{\bm{\xi}_0, \gls{xitv}}  s_k + q_t\lr{\bm{\xi}_0, \gls{xitv}} \IEEEeqnarraynumspace \nonumber \\
	&~~\forall k \in \{0,\dots,\gls{Nsmpc}\} \IEEEyessubnumber  \label{eq:smpc_cc}
\end{IEEEeqnarray}
where 
$\Delta \bm{\xi}_k = \bm{\xi}_k - \bm{\xi}_{k, \text{ref}}$ with \gls{ev} reference state $\bm{\xi}_{k, \text{ref}}$ and the linear function $\bm{f}^\text{d}$ according to \eqref{eq:d_evsys}. For the input difference $\Delta \bm{u}$, $\bm{u}_{-1}$ is set to the applied input of the previous time step. The cost function sum limits are shifted to include a terminal cost for $\bm{\xi}_N$. The weighting matrices are given by $\bm{Q}$, $\bm{S}$, and $\bm{R}$. We consider constant input constraints $\mathcal{U}$ according to \eqref{eq:constraints_input} and state constraints $\bm{\Xi}$ according to \eqref{eq:constraints_state}.

The resulting \gls{SMPC} optimal control problem \eqref{eq:smpc} is a quadratic program with linear constraints, accounting for uncertainty with the chance constraint reformulation described in Section~\ref{sec:cc}. This optimal control problem can be solved efficiently, where the major calculation steps to obtain the linear constraints \eqref{eq:smpc_cc} are performed before the optimization starts.

\subsection{Failsafe Trajectory Planning}
\label{sec:detailsFST}

While the \gls{SMPC} algorithm only accounts for part of the \gls{tv} uncertainty in order to plan an optimistic trajectory, the backup \gls{FT} algorithm needs to consider worst-case uncertainty realizations. This is achieved based on reachability analysis. First, the worst-case \gls{tv} occupancy prediction is determined. Then, linear constraints are generated based on the \gls{tv} predictions. Eventually, given a safe invariant terminal set, the \gls{FT} optimal control problem is solved.

\subsubsection{Target Vehicle Occupancy Prediction}
\label{sec:occpred}
Similar to the \gls{SMPC} algorithm, a rectangular safety area surrounding each \gls{tv} is defined. However, for the \gls{FT} the maximal reachable area needs to be determined. First, it is necessary to define certain traffic rules to which the \gls{tv} adheres, according to Assumption~\ref{ass:trafficrules}:
\begin{itemize}
\item Road boundaries apply.
\item Negative velocities are forbidden.
\item Collisions with vehicles directly in front of the \gls{tv} (in the same lane) must be avoided.
\item Only a single lane change is allowed (within the prediction horizon).
\item No lane change is allowed if the \gls{tv} velocity is below a predefined minimal lane change velocity $v_\text{LC,min}$.
\item No lane change is allowed if the distance to a vehicle on the new lane becomes too small.
\end{itemize}

As linear dynamics are assumed for the \gls{tv} motion, the minimal and maximal possible \gls{tv} inputs are used to determine the maximal reachable set, inspired by \cite{Althoff2013,Magdici2016,PekAlthoff2018}.

The set of all possible locations reachable for a \gls{tv} at prediction step $k$ is denoted by the reachable set \gls{Rtv}, including the \gls{tv} and shape. While referring to \gls{Rtv} as the reachable set of the \gls{tv}, we additionally enlarge this set accounting for the \gls{ev} shape. This is necessary as the set \gls{Rtv} is later used to avoid collisions by keeping the \gls{ev} center outside of \gls{Rtv}. 
Given the solution $\bm{\zeta}\lr{\hat{\bm{\xi}}^\text{TV}_0,\bm{U}}$ to the \gls{tv} dynamics \eqref{eq:tvmodel} starting at the initial state $\hat{\bm{\xi}}^\text{TV}_0$ applying an input sequence $\bm{U}$
, we define the reachable set 
\begin{IEEEeqnarray}{rl}
	& \gls{Rtv} = \left\lbrace \bm{\zeta}\lr{\hat{\bm{\xi}}^\text{TV}_0,\bm{U}}~ \Big| \right. \nonumber  \\
	&  ~~~~\left. \bm{U}(i) \in \mathcal{U}^\text{TV}~\forall i \in \{0, ...,k-1\},~~\hat{\bm{\xi}}^\text{TV}_0 \in \bm{\Xi}^\text{TV}_0 \right\rbrace.  \IEEEeqnarraynumspace  \IEEEyesnumber \label{eq:reachset} 
\end{IEEEeqnarray}
The initial state for the reachable set \gls{Rtv} is not the \gls{tv} state \gls{xitv0}, but depends on the sensor uncertainty as well as the \gls{tv} and \gls{ev} shape. This initial set is given by 
\begin{IEEEeqnarray}{rl}
	& \bm{\Xi}^\text{TV}_0 = \left\lbrace \hat{\bm{\xi}}^\text{TV}_0~  \Big| \right. \nonumber  \\
	& ~~~~\gls{xitv0} + \min\{\gls{wsens}\} - [l_\text{veh},0,w_\text{veh},0]^\top  \leq  \hat{\bm{\xi}}^\text{TV}_0, \nonumber \\
	& ~~~~\left. \hat{\bm{\xi}}^\text{TV}_0  \leq \gls{xitv0} + \max\{\gls{wsens}\} + [l_\text{veh},0,w_\text{veh},0]^\top	 \right\rbrace.  \IEEEeqnarraynumspace  \IEEEyesnumber \label{eq:xi0TV}
\end{IEEEeqnarray}
The initial set $\bm{\Xi}^\text{TV}_0$ in \eqref{eq:xi0TV} can be interpreted as a rectangle defined by two corners. As we assume a linear \gls{tv} prediction model, the reachable set \gls{Rtv} is calculated for prediction steps $k > 0$ by applying the maximal and minimal inputs $\bm{u}^\text{TV} \in \mathcal{U}^\text{TV}$, while adhering to traffic rules.

The reachable set is only calculated at discrete time steps. In order to account for a continuous system, the final reachable set \gls{bRtv} is obtained by building a rectangular convex hull, covering two consecutive prediction steps, i.e., 
\begin{IEEEeqnarray}{c}
\gls{bRtv} = \mathrm{conv} \left\lbrace \mathcal{R}^{\text{TV}}_{k-1}, \mathcal{R}^{\text{TV}}_k  \right\rbrace
\end{IEEEeqnarray}
where $\mathrm{conv}$ denotes the convex hull operation.

A special case is considered if the \gls{tv} is located behind the \gls{ev} in the same lane. The \gls{tv} must not collide with the \gls{ev} in the same lane, however, the \gls{tv} is allowed to switch lanes in order to pass the \gls{ev}. Here, this is accounted for by treating this special case in the following way. Three placeholder \gls{tv} reachable sets describe the possible \gls{tv} behavior. The first placeholder \gls{tv} reachable set is based in the \gls{ev} lane such that collisions with the \gls{ev} are avoided. The other two placeholder \gls{tv} reachable sets cover the admissible adjacent lanes left and right of the \gls{ev}, representing the reachable sets for a potential \gls{tv} lane change. Figure~\ref{fig:ft_occupancy} shows an example of areas possibly occupied by \glspl{tv} for multiple prediction steps.
\begin{figure}
\centering
\includegraphics[width = \columnwidth]{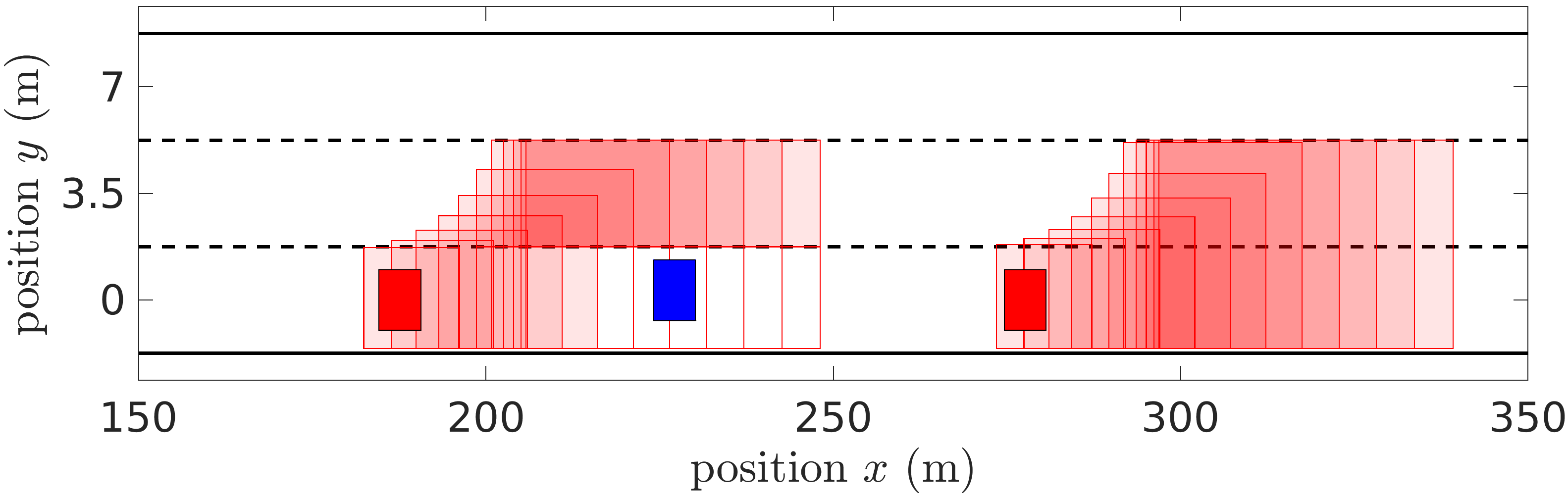}
\caption{Target vehicle occupancy sets for multiple prediction steps with the \gls{ev} in blue and \glspl{tv} in red. Areas shaded in red depict areas possibly occupied by a \gls{tv}. As the \glspl{tv} must avoid collision with vehicles in front, the left \gls{tv} cannot occupy the area in the \gls{ev} lane close to the \gls{ev}. A \gls{tv} double lane change is not considered.}
\label{fig:ft_occupancy}
\end{figure}

\subsubsection{\gls{FT} Constraint Generation}
\label{sec:constraintgeneration_ft}
Once the reachable sets \gls{bRtv} for each \gls{tv} are determined, linear constraints are generated. We again consider different cases regarding varying \gls{ev} and \gls{tv} positions. The cases are similar to those of Section~\ref{sec:constraintgeneration}, with a few variations. \gls{FT} cases are denoted with an asterisk. Same letters indicate similar \gls{SMPC} and \gls{FT} case types. The major \gls{FT} cases are illustrated in Figure~\ref{fig:ft_constraintcases}. Again, a complete overview of the \gls{FT} cases, requirements, and constraint parameters $q_x$, $q_y$, $q_t$ is found in Appendix~\ref{sec:appendixb}.
\begin{figure}
\centering
\includegraphics[width = 0.9\columnwidth,trim=0 2 0 1]{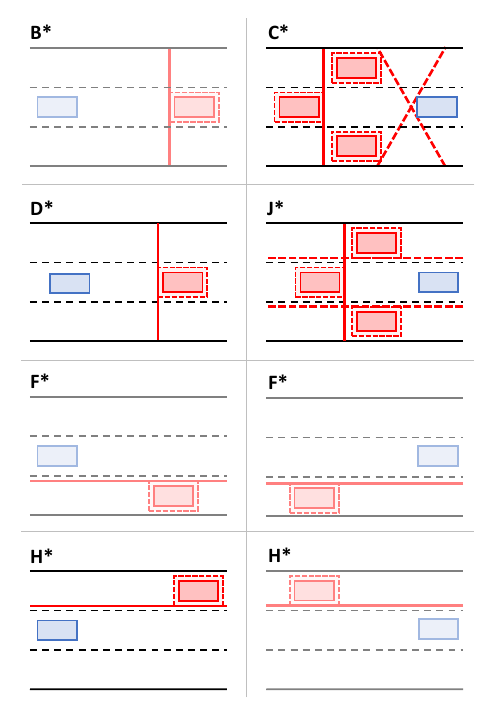}
\caption{Selected constraint generation cases for \gls{FT}. Cases equal to \gls{SMPC} constraint generation cases have lighter colors. Driving direction is from left to right. The \gls{ev} and \gls{tv} are shown in blue and red, respectively. The dashed red line represents the sfaety area around the \gls{tv}.}
\label{fig:ft_constraintcases}
\end{figure}

As for \gls{SMPC}, we initially consider the case where the \gls{tv} is farther away from the \gls{ev} and therefore no constraint is generated (case A${}^*$). 
\begin{itemize}
\item[A${}^*$)] $|\gls{dxEVTV}|  \geq r_\text{lar}$: no constraint
\end{itemize}
Then, all remaining cases are considered with a distance $|\gls{dxEVTV}|  < r_\text{lar}$. First, cases are addressed where the \gls{ev} is behind the \gls{tv} (case B${}^*$) with a distance $-\lr{\gls{dxEVTV}}  > \gls{fcloseft}$, which leads to a vertical constraint behind the \gls{tv}. In contrast to case D of \gls{SMPC}, in \gls{FT} overtaking is not actively planned. Therefore, if the \gls{ev} is in the same lane behind the \gls{tv}, a vertical constraint is employed (case D${}^*$), similar to case B${}^*$. 
\begin{itemize}
\item[B${}^*$)] $-\lr{\gls{dxEVTV}}  > \gls{fcloseft}$: vertical constraint behind \gls{tv}
\item[D${}^*$)] $\gls{evlane} = \gls{tvlane}$, $-\lr{\gls{dxEVTV}}  \leq \gls{fcloseft}$: vertical constraint behind \gls{tv} 
\end{itemize}
We now consider further cases where the \gls{ev} and \gls{tv} are close to each other, i.e., $| \gls{dxEVTV}|  \leq \gls{fcloseft}$. In case the \gls{ev} is located on a different lane than the \gls{tv}, horizontal constraints are generated at the \gls{tv}, independent of the relative longitudinal positioning (cases F${}^*$ and H${}^*$).
\begin{itemize}
\item[F${}^*$)] $\gls{evlane} > \gls{tvlane}$, $| \gls{dxEVTV}|  \leq \gls{fcloseft}$: horizontal constraint left of the \gls{tv}
\item[H${}^*$)] $\gls{evlane} < \gls{tvlane}$, $| \gls{dxEVTV}|  \leq \gls{fcloseft}$: horizontal constraint right of the \gls{tv}
\end{itemize}
The last cases required focuses on the \gls{ev} located in front of the \gls{tv} in the same lane. In \gls{SMPC}, no constraint was generated (case J). However, even though the \gls{tv} is required to avoid collisions with another vehicle in front, for safety reasons we account for \glspl{tv} located behind the \gls{ev} in \gls{FT}. As stated in Section~\ref{sec:occpred}, in these cases up to three placeholder \gls{tv} predictions are made, accounting for possible lane changes and overtaking maneuvers by the \gls{tv}. If the distance between the \gls{ev} and \gls{tv} is smaller, i.e., $\lr{\gls{dxEVTV}} \leq \gls{fcloseft}$, a vertical constraint is employed for the \gls{tv} prediction in the \gls{ev} lane, while horizontal constraints are formed for \gls{tv} predictions in adjacent lanes  (case J${}^*$). If the distance is larger, i.e., $\lr{\gls{dxEVTV}} > \gls{fcloseft}$, inclined constraints are generated for \gls{tv} predictions in lanes next to the \gls{ev}, allowing more \gls{ev} movement (case C${}^*$).
\begin{itemize}
\item[J${}^*$)] $\gls{evlane} = \gls{tvlane}$, $\lr{\gls{dxEVTV}} \leq \gls{fcloseft}$: (maximum of) three \gls{tv} predictions with vertical and horizontal constraints
\item[C${}^*$)] $\gls{evlane} = \gls{tvlane}$, $\lr{\gls{dxEVTV}} > \gls{fcloseft}$: (maximum of) three \gls{tv} predictions with vertical and inclined constraints
\end{itemize}

Overall, the constraints generated for \gls{FT} are more conservative than for \gls{SMPC}. This is due to the \gls{FT} aim of finding a trajectory which ends in a safe state. This would be complicated by incentivizing \gls{FT} to plan overtaking maneuvers. Details on finding a safe terminal state for the \gls{FT} optimal control problem are given in the following.

\subsubsection{Safe Invariant Terminal Set}
In addition to the regular safety constraints, a safe invariant terminal set is required to ensure safe \gls{ev} inputs after the finite \gls{MPC} prediction horizon. If, repeatedly, no solution to the \gls{SMPC} and \gls{FT} optimal control problems is found, all safe backup inputs will be eventually applied. The \gls{FT} inputs are designed in such a way that they remain safe over the prediction horizon. However, after $N_\text{FTP}$ inputs are applied and no new \gls{FT} solution is obtained, an emergency strategy has to be applied to come to a standstill. This is achieved by braking, while maintaining a constant steering angle $\delta = 0$, according to \eqref{eq:usafe1} and \eqref{eq:usafe2}. Therefore, the terminal state of the \gls{FT} optimal control problem needs to fulfill certain requirements. First, the vehicle orientation must be aligned with the road, i.e., $\phi = 0$. This guarantees that braking and a constant steering angle $\delta = 0$ keep the \gls{ev} within its current lane. Second, the distance to a \gls{tv} in front of the \gls{ev} must be large enough that no collision occurs if both vehicles initiate maximal deceleration. This is accounted for by
\begin{IEEEeqnarray}{rl}
\IEEEyesnumber \label{eq:termcon1}
s_N &\leq x_N^\text{TV} - \Delta s_{\gls{Nft},\text{min}} \IEEEyessubnumber \\
v_N &\leq v_{\gls{Nft},\text{max}} \IEEEyessubnumber
\end{IEEEeqnarray}
with minimal terminal safety distance $\Delta s_{\gls{Nft},\text{min}}$ and maximal terminal safety velocity 
\begin{IEEEeqnarray}{c}
v_{\gls{Nft},\text{max}} = v^{\text{TV}}_{\gls{Nft}, \text{min}} - \sqrt{2 \Delta s_{\gls{Nft},\text{min}} a_{x,\text{min}}} \label{eq:termcon2} 
\end{IEEEeqnarray}
where $v^{\text{TV}}_{\gls{Nft}, \text{min}}$ is the lowest predicted longitudinal \gls{tv} velocity. Both \eqref{eq:termcon1} and \eqref{eq:termcon2} combined ensure that the minimal terminal safety distance $\Delta s_{\gls{Nft},\text{min}}$ is large enough such that, given a maximal \gls{ev} velocity $v_{\gls{Nft},\text{max}} $, maximal deceleration of the \gls{ev} guarantees collision avoidance for $k > N_\text{FTP}$. This less intuitive terminal constraint again has the advantage of yielding linear constraints.

\subsubsection{\gls{FT} Optimal Control Problem}
An optimal control problem with a similar structure compared to \eqref{eq:smpc} is applied for the \gls{FT}, yielding
\begin{IEEEeqnarray}{rl}
\IEEEyesnumber \label{eq:ft}
	V^* &= \min_{\bm{U}} \sum_{k=1}^{\gls{Nft}} \norm{\Delta \bm{\xi}_k}_{\bm{Q}} + \norm{\bm{u}_{k-1}}_{\bm{R}} + \norm{\Delta\bm{u}_{k-1}}_{\bm{S}}
	\IEEEyessubnumber \IEEEeqnarraynumspace \\
	\text{s.t. } & \bm{\xi}_{k+1} = \bm{f}^\text{d}\lr{\bm{\xi}_0, \bm{\xi}_k, \bm{u}_k} \IEEEyessubnumber \IEEEeqnarraynumspace \label{eq:ft_dynamics}\\
	& \bm{\xi}_k \in \gls{XX} ~~~~\forall k \in \{1,\dots,\gls{Nft}\} , \IEEEyessubnumber\\
	& \bm{u}_k \in \mathcal{U} ~~~\forall k \in \{0,\dots,\gls{Nft}-1\} ,\IEEEyessubnumber \IEEEeqnarraynumspace\\
	& 0 \geq q_y\lr{\bm{\xi}_0, \gls{bRtv}} d_k + q_x\lr{\bm{\xi}_0, \gls{bRtv}}  s_k \nonumber  \IEEEeqnarraynumspace  \\
	&~~~~~~+ q_t\lr{\bm{\xi}_0, \gls{bRtv}}~~~~\forall k \in \{0,\dots,\gls{Nft}\} \IEEEyessubnumber \label{eq:ft_cc}\\ 
	& s_N \leq x_N^\text{TV} - \Delta s_{\gls{Nft},\text{min}} \IEEEyessubnumber \\ 
	&v_N \leq v_{\gls{Nft},\text{max}} \IEEEyessubnumber
\end{IEEEeqnarray}
with the linear function $\bm{f}^\text{d}$ according to \eqref{eq:d_evsys}. The probabilistic constraint \eqref{eq:smpc_cc} is now changed to constraint \eqref{eq:ft_cc}, accounting for the worst-case \gls{tv} uncertainty realizations. Similar to the \gls{SMPC} optimal control problem, \eqref{eq:ft} is a quadratic program with linear constraints, which can be solved efficiently.

\section{Results}
\label{sec:results}

We evaluate the proposed \gls{SMPCFT} algorithm in different settings. In the following, the simulation setup is introduced first. Then, \gls{SMPCFT} is analyzed and compared to an \gls{SMPC} approach and an \gls{FT} approach in two scenarios.

\subsection{Simulation Setup}

In this simulation section we analyze the scenario illustrated in Figure~\ref{fig:scenario_setup}.
\begin{figure}
\centering
\includegraphics[width = \columnwidth]{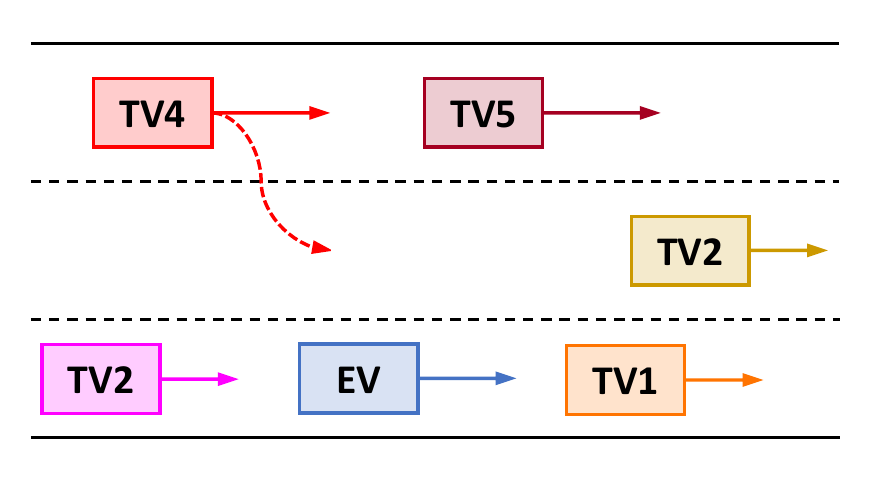}
\caption{Setup for both investigated scenarios (regular and emergency scenario).}
\label{fig:scenario_setup}
\end{figure}
The \gls{ev} is located on the right lane on a three-lane highway. We consider five \glspl{tv} surrounding the \gls{ev} on the highway. The goal for the \gls{ev} is to safely and efficiently maneuver through traffic. The specific aims are to avoid collisions while maintaining a velocity close to a chosen reference velocity. 

Given the initial scenario setup, we consider two different scenarios:
\begin{itemize}
\item[1)] Regular scenario: All \glspl{tv} keep their initial velocities and lanes.
\item[2)] Emergency scenario: One of the \glspl{tv} (TV5) performs an emergency braking maneuver. This causes TV4 to avoid TV5 by moving to the center lane. This is followed by a soft braking maneuver of TV1 to account for possible hazards. Eventually, TV4 moves to the left lane again to pass TV2.
\end{itemize}
The first scenario represents a regular scenario with no unexpected \gls{tv} behavior. The second scenario covers a rare case, where a series of unexpected \gls{tv} actions results in a challenging situation for the autonomous \gls{ev}.

The simulations are carried out in Matlab using the \textit{fmincon} solver on a desktop computer with an AMD Ryzen 7 1700X processor. 
The algorithms are based on the NMPC toolbox \cite{GruenePannek2017}. In the following, setup parameters are introduced which remain constant throughout the different simulations. All quantities are given in SI units. Units are omitted if clear by context.

The road consists of three lanes with lane width $w_\text{lane} = \SI{3.5}{\metre}$. All vehicles have the same rectangular shape with length $l_\text{veh} = \SI{5}{\metre}$ and width $w_\text{veh} = \SI{2}{\metre}$. 

All MPC algorithms use a sampling time $T = \SI{0.2}{\second}$. The MPC optimization horizon is chosen to be the same for each MPC algorithm, i.e., the \gls{SMPC} horizon is $\gls{Nsmpc} = 10$ and the \gls{FT} horizon is $\gls{Nft} = 10$. The linearized, time-discrete \gls{ev} prediction model and constraints follow \eqref{eq:d_evsys}-\eqref{eq:constraints_state} with $l_\text{f} = l_\text{r} = 2$. The \gls{tv} prediction model is given by \eqref{eq:tvmodel}-\eqref{eq:tv_feedback}. The maximum and minimum acceleration and steering angle for the \gls{ev} are $\bm{u}_\text{max} = [5, 0.2]^\top$, $\bm{u}_\text{min} = [-9, -0.2]^\top$, respectively. For the \gls{tv} with a point-mass prediction model the maximal and minimal accelerations in $x$- and $y$-direction are $\bm{u}^{\text{TV}}_\text{max} = [5, 0.4]^\top$, $\bm{u}^{\text{TV}}_\text{min} = [-9, -0.4]^\top$. The \gls{ev} input rate constraints are $-\Delta\bm{u}_\text{min} = \Delta\bm{u}_\text{max} = [9,0.4]^\top$. The lane boundaries follow from the lane width and vehicle width, i.e., the vehicle shape must remain within the road. The maximal velocity is $v_\text{max} = \SI{35}{\metre\per\second}$. The elements of the assumed \gls{tv} feedback controller matrix $\bm{K}$ of \eqref{eq:tv_feedback} are $k_{12} = -0.55$, $k_{21} = -0.63$, and $k_{22} = -1.15$. The \gls{tv} uncertainty covariance matrix is $\gls{wcov_t} = \text{diag}(0.44, 0.09)$ and the sensor noise follows $\gls{wsenscov} = \text{diag}(0.25, 0.25, 0.028, 0.028)$ and $|w^\text{sens}_{0,x}| \leq 0.25$, $|w^\text{sens}_{0,v_x}| \leq 0.25$, $|w^\text{sens}_{0,y}| \leq 0.028$, $|w^\text{sens}_{0,v_y}| \leq 0.028$. 

The safety parameters of Section~\ref{sec:details} for the generation of \gls{tv} safety rectangles and to distinguish between the cases are $\gls{msafe} = \SI{0.01}{\metre}$, $d_\text{lar} = \SI{200}{\metre}$, $\gls{dclose} = \SI{90}{\metre}$, $\gls{dcloseft} = \max\{\SI{10}{\metre},|v_0 N_\text{FTP}T|\}$, $v_{\text{LC, min}} = \SI{10}{\metre\per\second}$, and $\Delta s_{N,\text{min}} = \SI{22.5}{\metre}$.

The weighting matrices of the \gls{SMPC} and \gls{FT} optimal control problems \eqref{eq:smpc} and \eqref{eq:ft} are $\bm{Q} = \text{diag}(0, 0.25, 0.2, 10)$, $\bm{R} = \text{diag}(0.33, 5)$, and $\bm{S} = \text{diag}(0.33, 15)$.

In all scenarios the initial \gls{ev} reference is set to $[d_\text{ref}, \phi_\text{ref}, v_\text{ref}] = [0, 0, 27]$. While the reference orientation and velocity remain constant throughout the simulation, the reference lane is adapted depending on the current \gls{ev} lateral position. The \gls{ev} reference for the lateral position is always set to the current \gls{ev} lane center.

Given this simulation setup, we now investigate the individual scenarios and analyze the proposed \gls{SMPCFT} method.

\subsection{Regular Highway Scenario}
We first analyze a regular highway scenario. The initial states of the vehicles are given in Table~\ref{tab:rs_setup}. The five \glspl{tv} shown in Figure~\ref{fig:scenario_setup} all maintain their initial velocities and lanes, therefore, $\gls{xitvref} = \gls{xitv0}$. The initial longitudinal \gls{tv} position is irrelevant in the computation of \eqref{eq:tv_feedback}.
\begin{table}
\begin{center}
\caption{Initial Vehicle States}
\label{tab:rs_setup}
\begin{tabular}{c  c}
\toprule
vehicle & initial state\\
\midrule
\gls{ev} & $[0,0,0,27]^\top$\\
\addlinespace
\gls{tv}1 & $[70,20,0,0]^\top$\\
\addlinespace
\gls{tv}2 & $[125,20,3.5,0]^\top$\\
\addlinespace
\gls{tv}3 & $[-245,20,0,0]^\top$\\
\addlinespace
\gls{tv}4 & $[-35,32,7,0]^\top$\\
\addlinespace
\gls{tv}5 & $[40,32,7,0]^\top$\\
\addlinespace
\bottomrule
\end{tabular}
\end{center}
\end{table}

In the following the \gls{SMPCFT} solution is shown in detail and comparisons are made to an \gls{SMPC} and an \gls{FT} method.

\subsubsection{\gls{SMPCFT}}

Applying the proposed \gls{SMPCFT} approach to the regular highway scenario yields efficient \gls{ev} behavior in traffic. The \gls{SMPC} risk parameter is chosen to be $\beta = 0.8$. The inputs and important states are shown in Figure~\ref{fig:reg_smpcft_states}. The gray areas in Figure~\ref{fig:reg_smpcft_states} indicate selected sequences of the vehicle motion illustrated in Figure~\ref{fig:reg_smpc}. The vehicle shapes are only shown for every second step to avoid overlapping of vehicle shapes in Figure~\ref{fig:reg_smpc}.
\begin{figure}
\centering
	\includegraphics[width = \columnwidth]{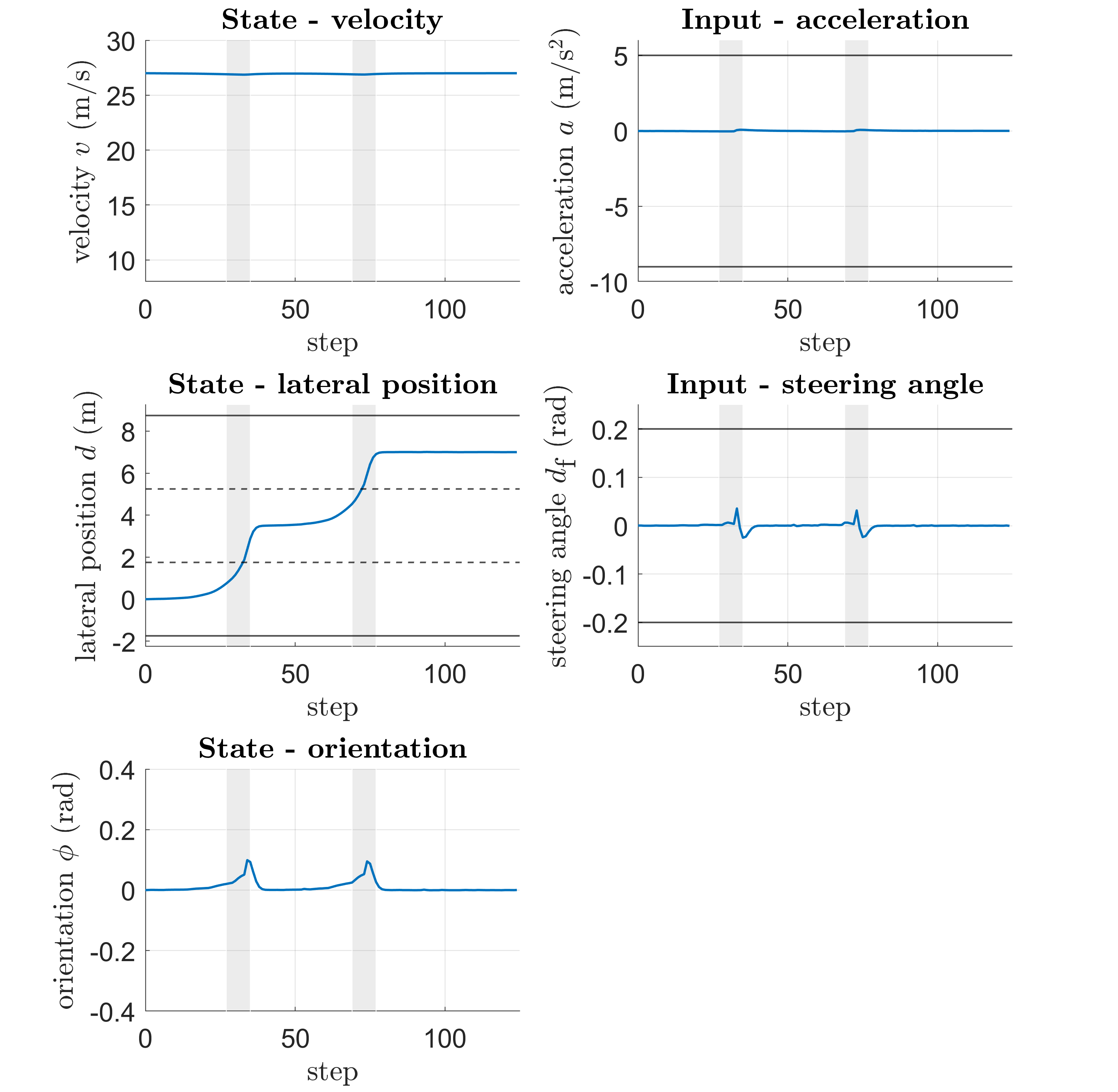}
\caption{\gls{SMPCFT} states and inputs for the regular scenario. Vehicle motion in the gray areas is illustrated in Figure~\ref{fig:reg_smpc}.}
\label{fig:reg_smpcft_states}
\end{figure}
\begin{figure}
\centering
\includegraphics[width = \columnwidth]{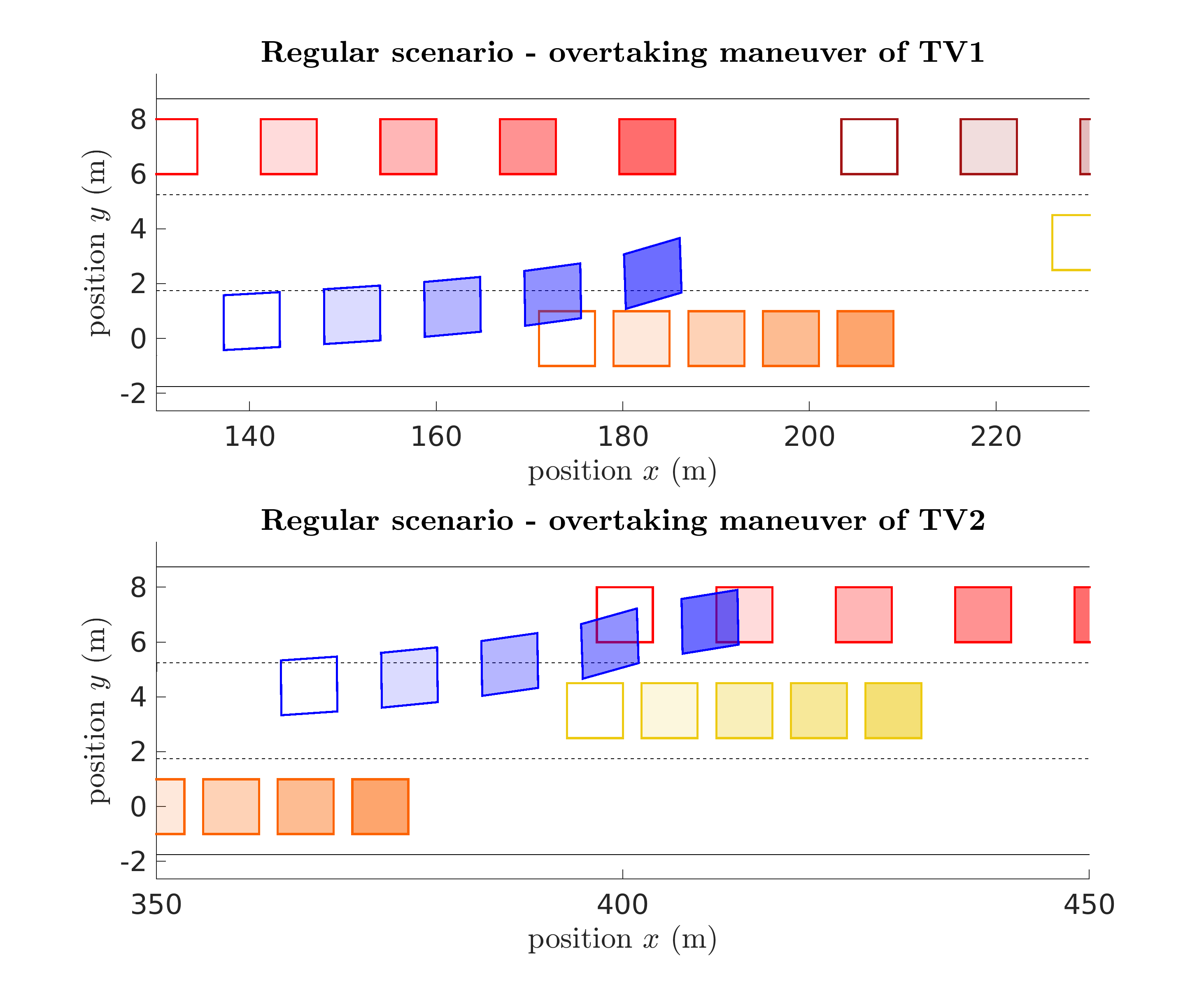}
\caption{Shots of the regular scenario with \gls{SMPCFT}. Fading boxes show past states. The \gls{ev} is shown in blue.}
\label{fig:reg_smpc}
\end{figure}
The \gls{ev} approaches \gls{tv}1 due to the velocity difference. The \gls{ev} then changes lanes to the center lane with a moderate steering angle of $\delta < 0.04$. The center lane is approached with hardly any overshoot. Once \gls{tv}2 is reached, the \gls{ev} again changes lanes. As \gls{tv}4 and \gls{tv}5 are farther ahead, the \gls{ev} smoothly moves to the left lane and eventually passes \gls{tv}2. The vehicle orientation remains at a limited level, i.e., $\phi  < 0.11$. Throughout the scenario, the \gls{ev} maintains the reference velocity, and acceleration inputs are negligible. The average computation time to solve the \gls{SMPC} and \gls{FT} optimal control problems are $\SI{0.11}{\second}$ and $\SI{0.15}{\second}$, respectively. If applied in a setting that requires online computation, it would be possible that the computation cannot be performed successfully in the designated sample time period. This case is treated as if the \gls{FT} optimal control problem is infeasible. Therefore, the previously calculated, still valid safe input sequence would be used.

We will now take a closer look at the constraints for \gls{SMPC} and \gls{FT}. \gls{SMPC} constraints for time step $h = 22$ are illustrated in Figure~\ref{fig:reg_smpc_constraints} for two prediction steps, $k = 1$ and $k = 5$. 
\begin{figure}
\centering
\includegraphics[width = \columnwidth]{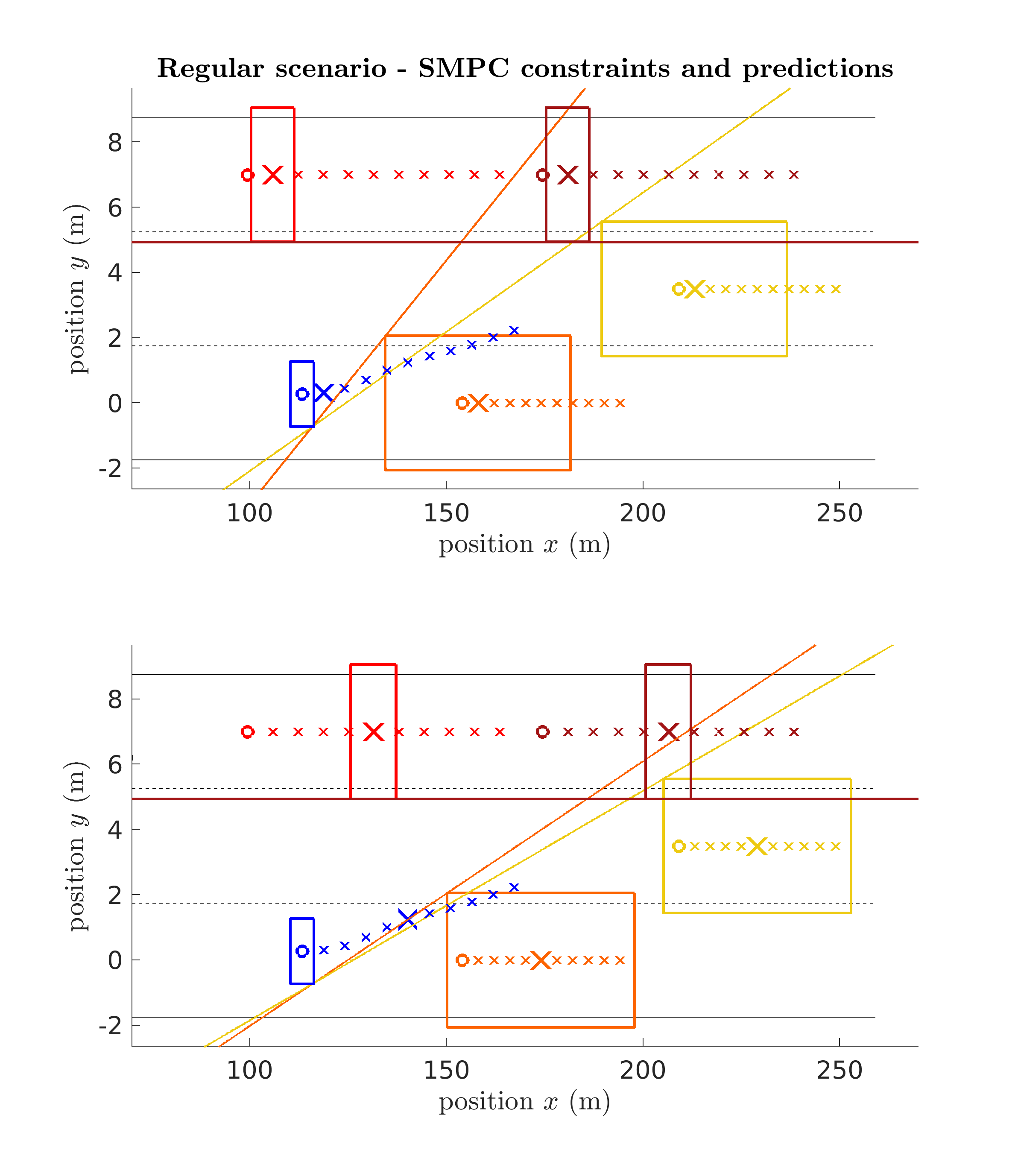}
\caption{\gls{SMPC} constraints for the regular scenario at time step $h=22$ and prediction steps $k = 1$, $k = 5$. The \gls{ev} shape and planned trajectory are shown in blue. \Glspl{tv} as well as respective safety rectangles and constraints have the same color. Initial states are marked by a circle, the prediction step displayed is indicated by a bold cross, other predicted states are represented by smaller crosses.}
\label{fig:reg_smpc_constraints}
\end{figure}
The boxes represent the \gls{ev} shape for the initial \gls{ev} state and the \gls{tv} safety rectangles at the given prediction step. The state of the current prediction step is marked with a bold cross, while other prediction steps are regular crosses. Initial vehicle states are indicated by a bold circle. For \gls{tv}1 in the same lane as the \gls{ev}, an inclined constraint is generated (case D). At each prediction step the constraint connects the initial \gls{ev} shape with the \gls{tv}1 safety rectangle at the predicted position. It is seen that the predicted \gls{SMPC} trajectory for the \gls{ev} stays above the constraint line. It is to note that only the respective predicted state must satisfy the illustrated constraint. Predicted states farther in the future satisfy respective constraints depending on a \gls{tv} safety rectangle for a predicted \gls{tv} position farther ahead. For \gls{tv}2 case E is active, also resulting in an inclined constraint. Both \gls{tv}4 and \gls{tv}5 are two lanes left of the \gls{ev}, yielding cases G and H, resulting in horizontal constraints to the right side of the \glspl{tv}. \gls{tv}3 is not shown in Figure~\ref{fig:reg_smpc_constraints} due to clarity as it is farther behind the other vehicles at this time instance.

The \gls{FT} constraints at step $h = 22$ for prediction steps, $k = 1$ and $k = 7$ are shown in Figure~\ref{fig:reg_ft_constraints}.
\begin{figure}
\centering
\includegraphics[width = \columnwidth]{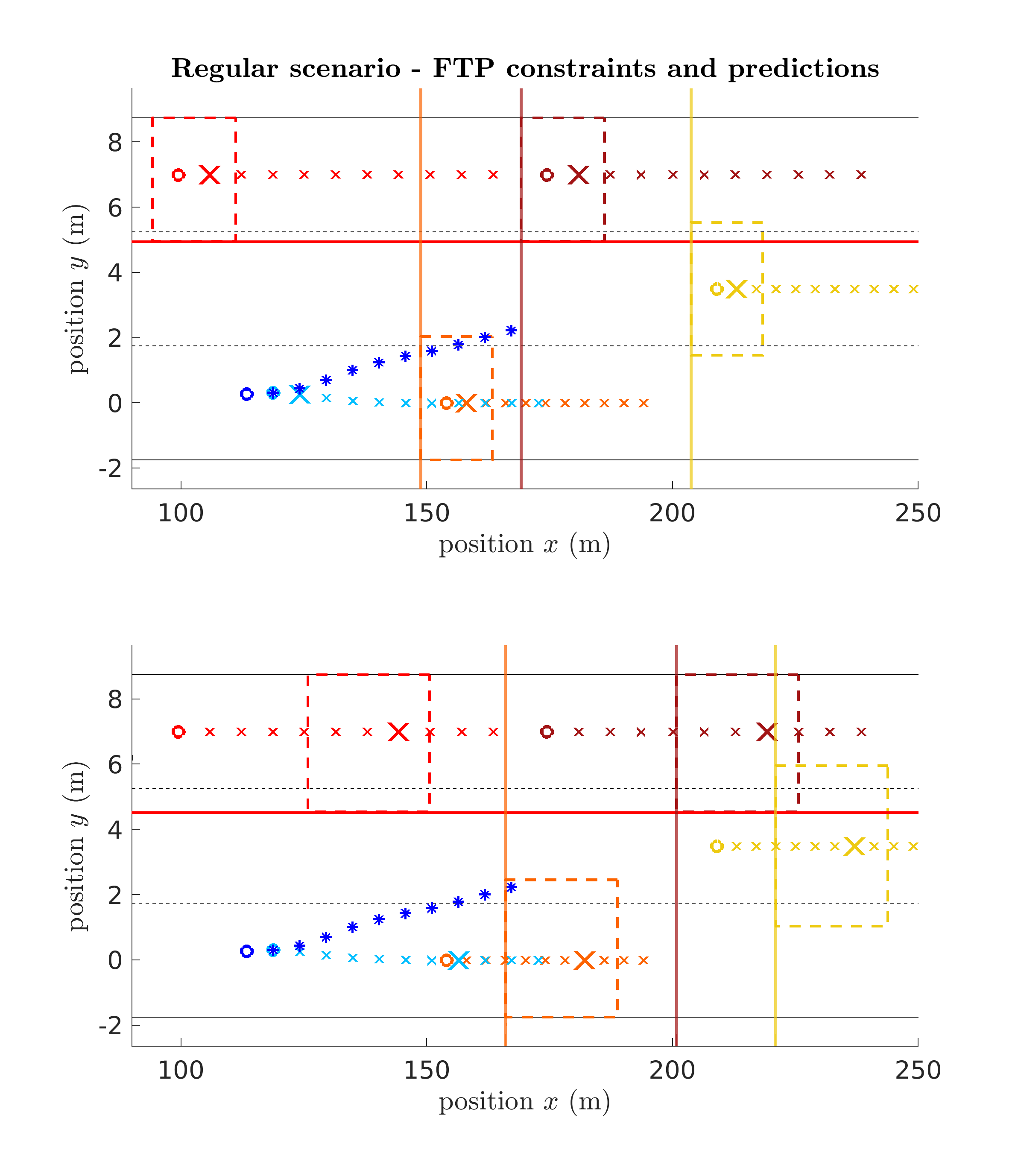}
\caption{\gls{FT} constraints for the regular scenario at time step $h=22$ and prediction steps $k = 1$, $k = 7$. The \gls{ev} is shown in blue. \Glspl{tv} as well as respective reachable sets and constraints have the same color. Initial states are marked by a circle. The initial \gls{FT} state starts after the first \gls{SMPC} input is applied. The prediction step displayed is indicated by a bold cross, other predicted states are represented by smaller crosses. For reference, the planned \gls{SMPC} trajectory is given by dark blue asterisks with a dark blue circle indicating the initial \gls{ev} state.}
\label{fig:reg_ft_constraints}
\end{figure}
The constraints are more conservative compared to the \gls{SMPC} constraints. The reachable \gls{tv} sets extend further to the back than the front, as maximal deceleration is larger than maximal acceleration. Additionally, the convex hull of reachable sets over two consecutive steps is considered. Constraints for \gls{tv}1 and \gls{tv}2 are built according to cases D${}^*$ and B${}^*$, respectively. Both constraints for \gls{tv}4 and \gls{tv}5 are generated given case H${}^*$. The planned \gls{SMPC} trajectory is shown in dark blue for reference, while the \gls{FT} trajectory is shown in light blue. While the \gls{SMPC} trajectory moves towards the center lane to overtake \gls{tv}1, the \gls{FT} trajectory finds a vehicle motion which, for the final prediction step, remains in the current lane with $\phi = 0$ and enough distance to \gls{tv}1, i.e., a safe terminal state. As the \gls{FT} optimal control problem yields a solution, i.e., a safe trajectory, the first input $\bm{u}_{\text{SMPC},0}$ of the planned \gls{SMPC} trajectory is then applied.

\subsubsection{Comparison to \gls{SMPC} and \gls{FT}}

Throughout the entire simulation both the \gls{SMPC} and \gls{FT} optimal control problems remain feasible. Therefore, the \gls{SMPC} inputs are always applied. Only applying an \gls{SMPC} algorithm without \gls{FT} would therefore yield the same result for this regular scenario. 

Unlike \gls{SMPC}, applying only \gls{FT} results in a different solution. As the constraints are more conservative compared to \gls{SMPC}, the \gls{ev} never changes lanes to overtake. As indicated by the \gls{FT} prediction in Figure~\ref{fig:reg_ft_constraints}, the \gls{FT} constraints keep the \gls{ev} in its current lane. There are situations where the \gls{FT} solution leads to a lane change, however, these situations are rare and only occur to avoid an obstacle, not to actively overtake it. 

We will use the following metric to compare the performance of \gls{SMPCFT} and \gls{FT}. Based on the cost function of the optimal control problem, the applied inputs and resulting states for the entire simulation are analyzed according to 
\begin{IEEEeqnarray}{c}
J_\text{sim} = \sum_{k=1}^{N_\text{sim}} \norm{\Delta \bm{\xi}_k}_{\bm{Q}} + \norm{\bm{u}_{k-1}}_{\bm{R}} + \norm{\Delta\bm{u}_{k-1}}_{\bm{S}} \label{eq:overallcost}
\end{IEEEeqnarray}
with the simulation steps $N_\text{sim}$.

The overall cost for \gls{SMPCFT} is $J_\text{sim} = 11.32$, while the overall \gls{FT} cost is $J_\text{sim} = 4.03\mathrm{e}{4}$. As expected, the cost comparison shows that the \gls{SMPCFT} approach yields a more efficient behavior than a safe \gls{FT} approach. In this case increased efficiency results from keeping the velocity close to the reference velocity.

\subsubsection{Risk Parameter Variation}

In the previously discussed simulation, the risk parameter was chosen to be $\beta = 0.8$. Here, we briefly analyze the effect of varying risk parameters on the \gls{ev} performance. The risk parameters analyzed range from $\beta = 0.8$ to $\beta = 0.999$. The overall simulation cost, according to \eqref{eq:overallcost}, for each risk parameter is given in Table~\ref{tab:riskparam}.
\begin{table}
\begin{center}
\caption{Risk Parameter Analysis}
\label{tab:riskparam}
\begin{tabular}{c  c  c  c  c c}
\toprule
risk parameter $\beta$ & $0.8$ & $0.9$ & $0.95$ & $0.99$ & $0.999$\\
\midrule
cost $J_\text{sim}$ & $11.21$ & $11.35$ & $11.58$ & $11.34$ & $11.31$ \\
\bottomrule
\end{tabular}
\end{center}
\end{table}
Intuitively one could expect increasing cost with higher risk parameters. The overall costs of the simulation results show that the \gls{SMPC} behavior and costs for a regular scenario are very similar. However, it can be beneficial regarding the cost to choose a larger risk parameter. While this slightly increases conservatism, inputs are changed more smoothly. In all five examples the \gls{ev} behavior is almost similar.

	\subsubsection{Varying Simulation Settings}
In the previous analysis, only one vehicle configuration is considered. In order to show that the \gls{SMPCFT} method is suitable for various scenarios, we ran 1000 simulations, each consisting of 125 simulation steps, with randomly selected initial vehicle positions and velocities for each simulation run. The \gls{ev} is located on one of the three lanes, i.e., $d_0 \in \{0, 3.5, 7\}$, with initial longitudinal position $s_0 = 0$ and velocity $v=27$. The five \glspl{tv} are randomly placed on one of the three lanes with an initial longitudinal position $x^\text{TV}_0 \in [-100, 200]$, constant velocity $v_x^\text{TV} \in [20, 32]$, and constant $v_y^\text{TV} = 0$. It is ensured that all vehicles positioned on the same lane have an initial longitudinal distance $\Delta x \geq 50$ and that \gls{tv} velocities are chosen such that \glspl{tv} do not collide with each other. 
The \gls{SMPCFT} method successfully handled all 1000 simulation runs and no collisions occurred.

\subsection{Emergency Highway Scenario}
After having shown the efficient \gls{SMPCFT} planning for a regular highway scenario, we now illustrate the safety property of the proposed algorithm in an emergency scenario. The initial vehicle states are the same as in the regular scenario, see Table~\ref{tab:rs_setup}. However, in this emergency scenario the \glspl{tv} change their velocities and lateral positions. Starting at time step $h=20$ \gls{tv}5 initiates an emergency braking maneuver with maximal deceleration until reaching a complete halt. This causes \gls{tv}4 to change lanes to the center lane in order to avoid \gls{tv}5. \gls{tv}1 reduces its velocity to $v_x^\text{TV1} = \SI{10}{\metre\per\second}$. After having passed \gls{tv}5, \gls{tv}4 moves to the left lane to then pass the slower \gls{tv}2. \gls{tv}1 also increases its velocity to $v_x^\text{TV1} = \SI{20}{\metre\per\second}$.

In the following \gls{SMPC} without \gls{FT} is analyzed first. Then the solution of the proposed \gls{SMPCFT} algorithm is presented.

\subsubsection{\gls{SMPC}}

Applying only \gls{SMPC} results in optimistic \gls{ev} trajectory planning, while not considering highly unlikely events. Even though \gls{tv}4 is slowly moving to the center lane, the \gls{ev} still moves to the center lane to overtake \gls{tv}1, as the predicted collision probability with \gls{tv}4 remains below the specified risk parameter. However, at step $h=25$ \gls{tv}4 continues to increase its lateral velocity towards the center lane and \gls{tv}1 decreases its velocity, therefore disallowing the \gls{ev} to return to the right lane. There exists no feasible solution to the \gls{SMPC} optimal control problem anymore which satisfies the chance constraint with the desired risk parameter. This causes the \gls{ev} to collide with \gls{tv}4 due to the concatenation of multiple unlikely and disregarded events. The collision sequence is illustrated in Figure~\ref{fig:emg_smpc}.
\begin{figure}
\centering
\includegraphics[width = \columnwidth]{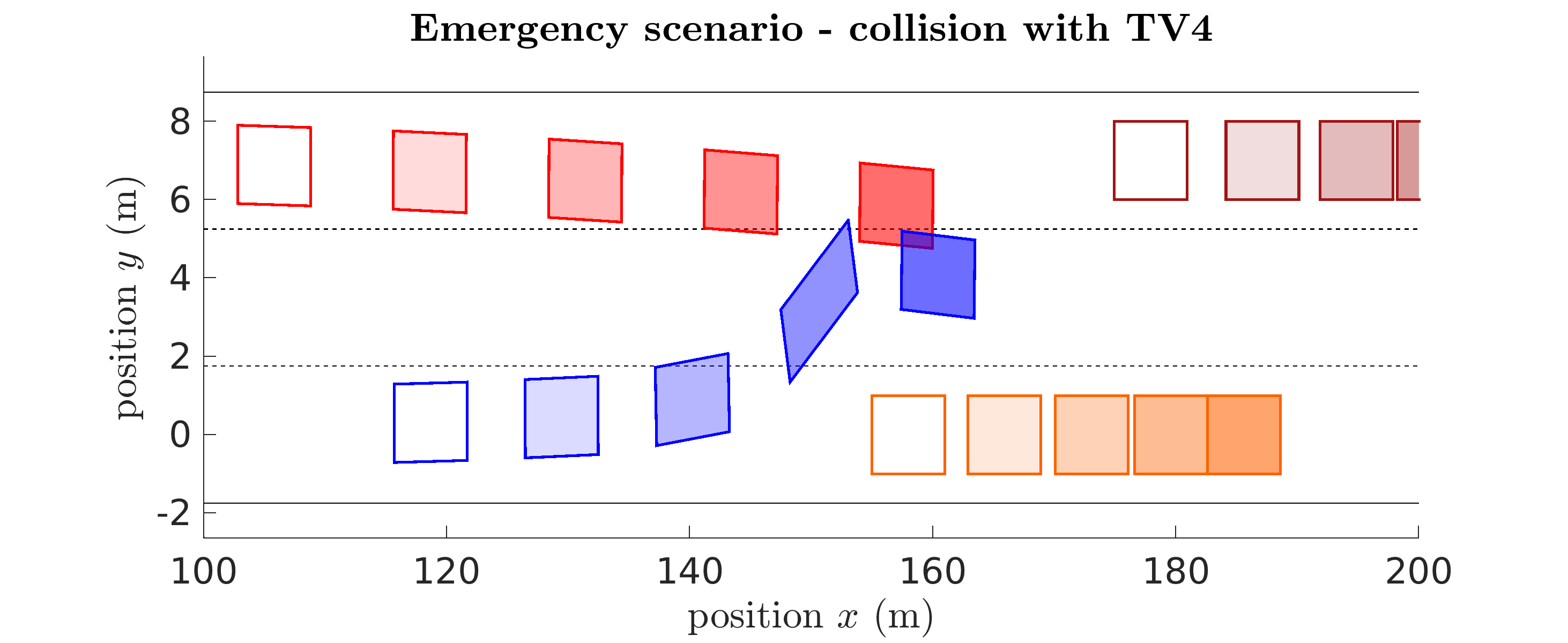}
\caption{Shots of the emergency scenario collision applying only \gls{SMPC}. Fading boxes show past states. The \gls{ev} is shown in blue.}
\label{fig:emg_smpc}
\end{figure}
While \gls{SMPC} performs well in regular scenarios without unlikely uncertainty realizations, these rare situations cause major safety issues. 

\subsubsection{\gls{SMPCFT}}

We now show how the proposed \gls{SMPCFT} method handles the emergency scenario. The \gls{ev} states and inputs are given in Figure~\ref{fig:emg_smpcft_states}. Gray areas represent sequences illustrated in Figure~\ref{fig:emg_smpcft}.
\begin{figure}
\centering
\includegraphics[width = \columnwidth]{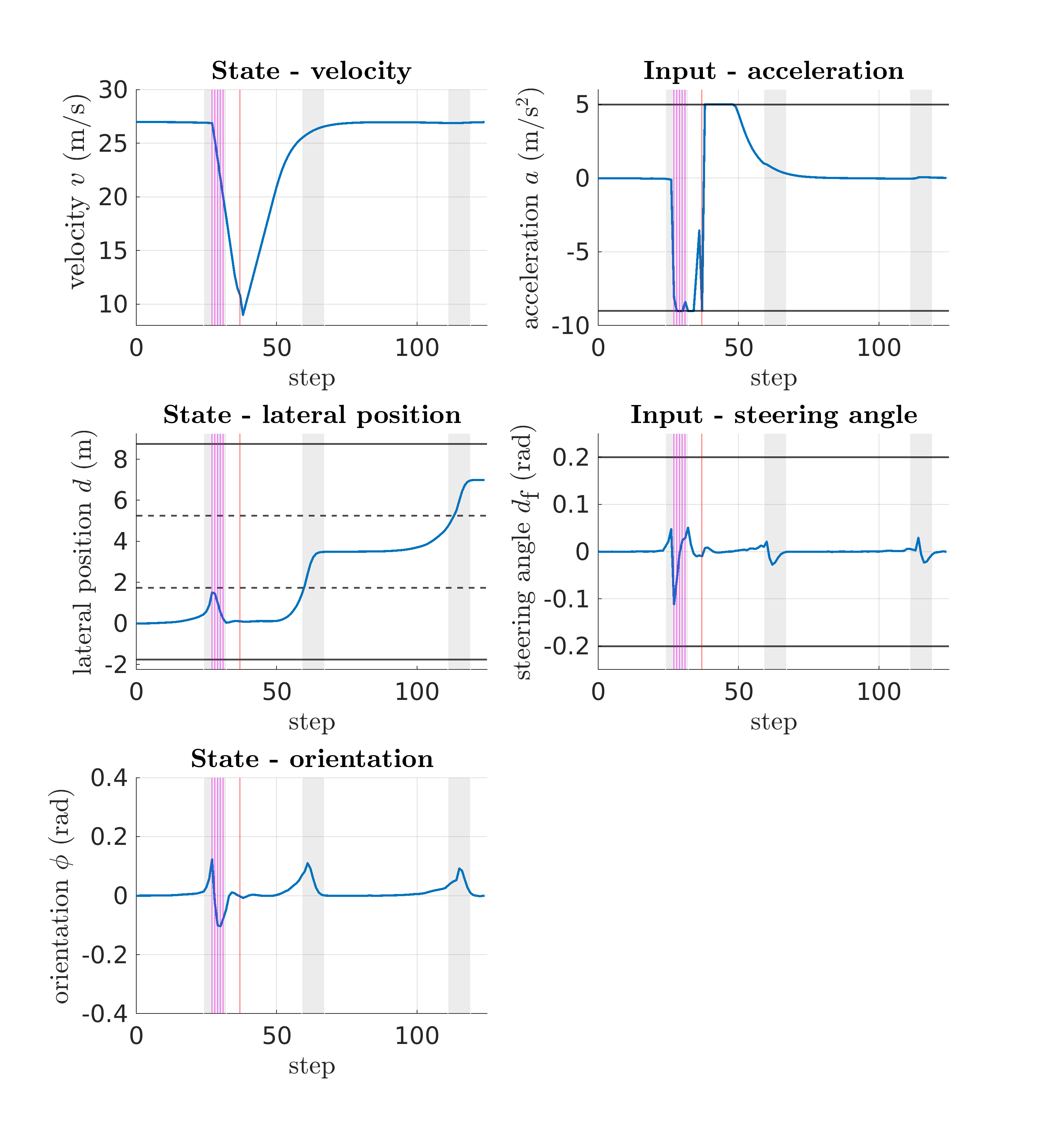}
\caption{\gls{SMPCFT} states and inputs for the emergency scenario. Pink vertical lines represent infeasible \gls{SMPC} and feasible \gls{FT} solutions, red vertical lines show infeasible \gls{FT} solutions. Vehicle motion in the gray areas is illustrated in Figure~\ref{fig:emg_smpcft}.}
\label{fig:emg_smpcft_states}
\end{figure}
\begin{figure}
\centering
\includegraphics[width = \columnwidth]{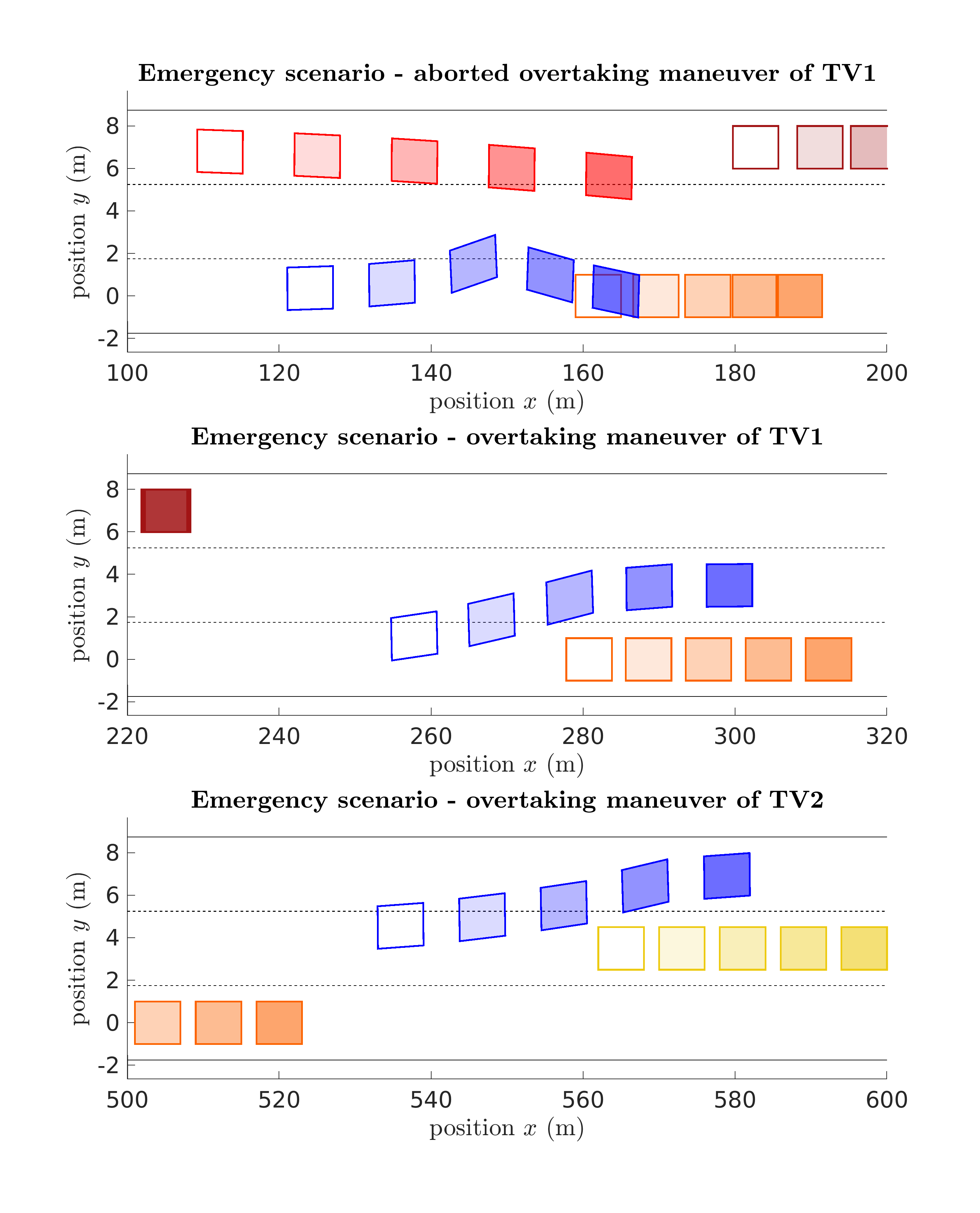}
\caption{Shots of the emergency scenario with \gls{SMPCFT}. Fading boxes show past states. The \gls{ev} is shown in blue.}
\label{fig:emg_smpcft}
\end{figure}

Initially the \gls{ev} attempts to switch lanes and overtake \gls{tv}1. However, at step $h=27$ the \gls{SMPC} is unable to find a solution. The \gls{FT} problem is still solved successfully and the first planned \gls{FT} input is applied. For the next four steps the \gls{SMPC} optimal control problem remains infeasible, indicated by the pink lines in Figure~\ref{fig:emg_smpcft_states}, and the \gls{FT} inputs are applied, which are obtained by successfully solving the \gls{FT} \glspl{ocp}. 
The \gls{ev} slows down and returns to the right lane, as illustrated in the first shot of Figure~\ref{fig:emg_smpcft}. At step $h=37$ the \gls{SMPC} problem is feasible and the \gls{ev} plans to overtake \gls{tv}1 again. However, as \gls{tv}4 is still too close, the \gls{FT} is unable to find a new safe backup trajectory if the next planned \gls{SMPC} input were applied, i.e., the \gls{FT} \gls{ocp} becomes infeasible. Therefore, while the \gls{SMPC} problem is still feasible, the safe input sequence obtained at the previous time step $h=36$ is applied to the \gls{ev}, as indicated by the red line in Figure~\ref{fig:emg_smpcft_states}. The \gls{ev} remains in the right lane until \gls{tv}4 is far enough away to safely change to the center lane. 
Eventually the \gls{ev} passes \gls{tv}2 by smoothly switching to the left lane with a small steering angle change. The average computation time for solving the \gls{SMPC} and \gls{FT} optimal control problems are $\SI{0.15}{\second}$ and $\SI{0.22}{\second}$, respectively. The values are higher compared to the regular scenario, as the computation time for infeasible optimal control problems is significantly larger. 

In this rare emergency situation, the \gls{ev} inputs lead to a less smooth motion. However, this is acceptable as safety is guaranteed even in this challenging situation. 

It is also possible to only apply \gls{FT} in this emergency scenario. While this leads to safe vehicle behavior throughout the simulation, the \gls{ev} does not overtake \gls{tv}1 and \gls{tv}2. The combination of \gls{SMPC} and \gls{FT}, however, enables the \gls{ev} to safely handle this emergency scenario while passing slower \glspl{tv} when possible. Comparing the cost yields the following result. Applying \gls{FT} to the emergency scenario yields a cost of $J_\text{sim} = 4.28\mathrm{e}{4}$, while the \gls{SMPCFT} cost is $J_\text{sim} =3.34\mathrm{e}{4}$.

The \gls{SMPC} and \gls{FT} constraints at step $h=24$ and prediction step $k = 4$ closely before the aborted lane change maneuver are displayed in Figure~\ref{fig:emg_smpcft_constraints}.
\begin{figure}
\centering
\includegraphics[width = \columnwidth]{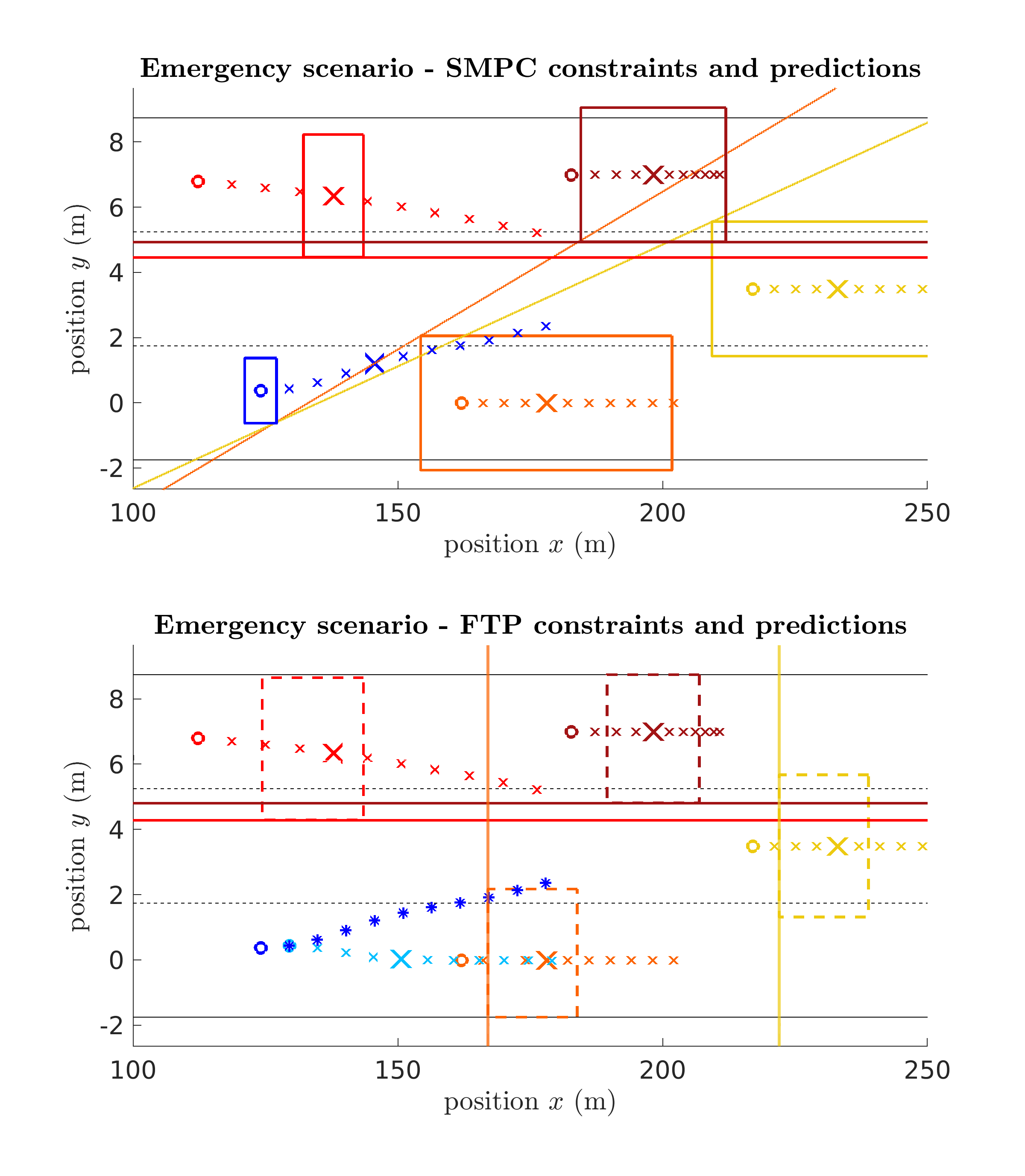}
\caption{\gls{SMPCFT} constraints for the emergency scenario at time step $h=24$ and prediction step $k = 4$. \gls{SMPC} constraints are illustrated in the top, \gls{FT} constraints in the bottom. The \gls{ev} shape and planned trajectory are shown in blue. \Glspl{tv} as well as respective safety rectangles, reachable sets, and constraints have the same color. Initial states are marked by a circle. The initial \gls{FT} state starts after the first \gls{SMPC} input is applied. The prediction step displayed is indicated by a bold cross, other predicted states are represented by smaller crosses. For reference, in the \gls{FT} plot the planned \gls{SMPC} trajectory is given by dark blue asterisks with a dark blue circle indicating the initial \gls{ev} state.}
\label{fig:emg_smpcft_constraints}
\end{figure}
Whereas the \gls{SMPC} prediction plans to steer the \gls{ev} into the center lane, the planned \gls{FT} trajectory remains in the right lane. Once the \gls{FT} is unable to still find a safe trajectory remaining in the right lane, given the first \gls{SMPC} step $\bm{u}_{\text{SMPC},0}$, the \gls{FT} takes control. 

In summary, the simulation scenarios in this section have shown the benefits of the proposed \gls{SMPCFT} method. The \gls{SMPC} part optimistically plans trajectories which are executed as long as there always exists a safe backup trajectory, computed by \gls{FT}. In regular scenarios \gls{SMPCFT} provides benefits known from \gls{SMPC}. In emergency scenarios the safety guarantee of \gls{FT} holds while the \gls{ev} is still more efficient compared to applying pure \gls{FT}.

\section{Discussion}
\label{sec:discussion}

In \gls{MPC} applications, the prediction horizon is a design choice. While in general long prediction horizons are useful, here it is beneficial to select a relatively short \gls{SMPC} horizon. This decreases the chance of the \gls{SMPC} problem becoming infeasible and it is not necessary to employ a long horizon, as \gls{FT} is used to guarantee safety. In some \gls{FT} approaches it is required that the vehicle comes to a standstill at the end of the fail-safe trajectory. Here, we only require a certain distance to vehicles ahead and zero orientation offset with respect to the road for the terminal state. This enables the use of a relatively short \gls{FT} horizon. However, if the \gls{FT} horizon is selected too short, lane changes are not possible anymore, as the \gls{FT} requires a certain amount of prediction steps to fulfill the terminal constraint. For very short \gls{FT} horizons lane change maneuvers are therefore aborted before the lane change actually takes place.

It is possible to get oscillating behavior between applied \gls{SMPC} inputs and the activation of \gls{FT}. In other words, in step one the \gls{SMPC} input is applied, which potentially causes the \gls{FT} to intervene in the next step. Then, a safe state is again achieved, leading to another, potentially over-aggressive \gls{SMPC} input, again requiring \gls{FT} in the subsequent step. This can be avoided by designing the \gls{SMPC} controller and its constraints less aggressively, as done in the simulation study.

Regarding the simulation, simulating each scenario once is adequate. While the \gls{tv} is assumed to behave probabilistically by the \gls{ev}, the actual \gls{tv} behavior here is deterministic. And whereas \gls{SCMPC} depends on drawn samples, which vary between simulations, the applied \gls{SMPC} approach uses a chance constraint reformulation that always yields the same constraint, given the same uncertainty distribution.

In the emergency scenario multiple \glspl{tv} change velocities or lanes. This scenario was chosen such that the \gls{SMPC} method causes a collision, which usually does not happen even for highly unlikely \gls{tv} trajectories. The chance constraint within \gls{SMPC} does allow a small probability of constraint violation, however, in most cases the iterative structure of \gls{MPC} handles potential future constraint violations. Furthermore, constraint violations do not necessarily cause collisions, as the safety area around a \gls{tv} is larger than the actual \gls{tv} shape.

Comparing the planned \gls{SMPC} trajectories for the \gls{ev} at two consecutive time steps without any major environment changes, one would assume that the planned trajectory remains similar. However, this is not the case. The constraints with respect to other \glspl{tv} are generated based on the \gls{ev} state at the beginning of the optimal control problem in order to formulate linear constraints. 
Therefore, in the next step, the constraint generation is based on an updated initial \gls{ev} state, resulting in a slightly different planned \gls{SMPC} trajectory compared to the previously planned trajectory. This could be addressed by using \gls{ev} predictions for the constraint generation, however, this would require nonlinear constraints.

The applied vehicle inputs in the emergency scenario lead to relatively high steering angles. This is not ideal for a smooth vehicle motion. Even though this behavior is acceptable in rare cases, the motion could be optimized by defining more cases for the constraint generation.

The individual \gls{SMPC} and \gls{FT} algorithms in this work are possible controller realizations, specifically designed for highway scenarios with multiple \glspl{tv}. The properties of the combined \gls{SMPCFT} method are not restricted to the suggested \gls{SMPC} and \gls{FT} trajectory planners described in Section~\ref{sec:detailsSMPC} and Section~\ref{sec:detailsFST}, respectively. Other \gls{SMPC} or \gls{FT} approaches can be applied.

In dense traffic or unclear traffic situations, humans often do not wait until the desired vehicle motion is entirely realizable. Instead, humans often slowly initiate maneuvers, causing other vehicles to react. For example, cutting into a lane is often preceded by slight motion towards the other lane so that other vehicles leave extra space. Therefore, it is possible to execute the lane change maneuver successfully, even though it was not possible to safely plan the entire lane change maneuver initially. While perfectly mimicking this human approach by automated vehicles is challenging due to safety reasons, the proposed \gls{SMPCFT} framework enables automated vehicle motion that comes close to this efficient human behavior.

\section{Conclusion}
\label{sec:conclusion}

In this work we presented a safe and efficient \gls{SMPCFT} method for self-driving vehicles. While \gls{SMPC} is used to plan optimistic, efficient vehicle trajectories, a fail-safe trajectory planning (FTP) MPC problem ensures that only those \gls{SMPC} inputs are applied which keep the vehicle in a safe state. The advantages of the proposed \gls{SMPCFT} algorithm are shown in a simulation study, where comparisons are done to pure \gls{SMPC} and pure \gls{FT} methods.

The efficiency of the \gls{SMPCFT} method depends on the proposed constraint generation. Extending and refining the case differentiation will have a positive effect on efficiency. Considering urban automated driving, the \gls{SMPCFT} approach remains valid, however, the case differentiation must be adapted to fit the urban traffic environment.

The presented \gls{SMPCFT} method is suitable to be applied to further safety-critical transportation applications, such as currently developed air taxis. However, it is also possible to extend the application area to non-transportation applications, such as human-robot collaboration, where uncertainty is always present while safety must still be guaranteed.

\appendices

\section{Linearized and Discretized System Matrices}
\label{sec:appendixd}

The linearized, time-discrete system matrices $\bm{A}_\text{d}$ and $\bm{B}_\text{d}$ in \eqref{eq:d_evsys} are given by
\begin{IEEEeqnarray}{rl}
\IEEEyesnumber \label{eq:lindic_matrix}
\bm{A}^\text{d} &= 
\setlength{\delimitershortfall}{1pt} 
\begin{bmatrix}
1 & 0 & -T v \sin z_1 & T \cos z_1 - \frac{ z_2 \sin z_1 }{2 z_4}\\
0 & 1 & T v \cos z_1 & T \sin z_1 + \frac{ z_2 \cos z_1 }{2 z_4} \\
0 & 0 & 1 & \frac{ T \tan{\delta}  }{ z_4}\\
0 & 0 & 0 & 1
\end{bmatrix} \IEEEeqnarraynumspace \IEEEyessubnumber\\
\bm{B}^\text{d} &= 
\setlength{\delimitershortfall}{1pt}
\begin{bmatrix}
\frac{T^2 \cos z_1}{2} & -\frac{T^2 v z_7 \sin z_1 }{2} - \frac{z_8 \sin z_1 }{z_9} \\
\frac{T^2 \sin z_1}{2} & \frac{T^2 v z_7 \cos z_1  }{2} + \frac{z_8 \cos z_1  }{z_9}\\
\frac{ T^2 \tan{\delta}  }{2 z_4} & T z_7 \\
T & 0
\end{bmatrix} \IEEEyessubnumber
\end{IEEEeqnarray}
with
\begin{IEEEeqnarray}{rl}
\IEEEyesnumber \label{eq:lindic_matrix_app}
z_1 &= \phi + \arctan\lr{\frac{l_\text{r} \tan \delta}{l_\text{r} + l_\text{f}}} \IEEEyessubnumber\\
z_2 &= T^2 v \tan{\delta} \IEEEyessubnumber\\
z_3 & = \lr{l_\text{r} \tan \delta}^2 \IEEEyessubnumber\\
z_4 &= \lr{l_\text{r} + l_\text{f}} \lr{\frac{z_3}{\lr{l_\text{r} + l_\text{f}}^2}+1}^{\frac{1}{2}} \IEEEyessubnumber\\
z_5 & = v \lr{\lr{\tan \delta}^2 + 1} \IEEEyessubnumber\\
z_6 &= \lr{l_\text{r} + l_\text{f}}^3 \lr{\frac{z_3}{\lr{l_\text{r} + l_\text{f}}^2}+1}^{\frac{3}{2}} \IEEEyessubnumber\\
z_7 &= \frac{z_5}{z_4} - \frac{z_3z_5}{z_6} \IEEEyessubnumber\\
z_8 &= T l_\text{r} z_5  \IEEEyessubnumber \\
z_9 &= \lr{l_\text{r} + l_\text{f}} \lr{\frac{z_3}{\lr{l_\text{r} + l_\text{f}}^2}+1}. \IEEEyessubnumber
\end{IEEEeqnarray}

\section{Proof of Theorem 1}
\label{sec:appendixa}

\begin{proof}
Recursive feasibility is proved by induction by showing that $\Gamma_h \neq \emptyset \Rightarrow \Gamma_{h+1} \neq \emptyset$ for all $h \in \mathbb{N}$. 
 
At time step $h=0$ it holds that $\gls{chiX}_0^{\bm{U}_{\text{safe},\text{init}}} \in \Gamma_0$, i.e., an initially safe trajectory exists according to Assumption~\ref{ass:FTinitial}. If at step $h=0$ the \gls{FT} \gls{ocp} can be solved, a new safe input set $\bm{U}_{\text{safe},0}$ is obtained according to \eqref{eq:usafe1} or \eqref{eq:usafe2}. This new safe input set $\bm{U}_{\text{safe},0}$ remains valid at step $h=1$ and ensures that a safe trajectory exists, i.e., $\gls{chiX}_1^{\bm{U}_{\text{safe},0}} \in \Gamma_1$. If at step $h=0$ the \gls{FT} \gls{ocp} is infeasible, the shifted previous safe input set 
remains valid, i.e., $\bm{U}_{\text{safe},0} = \bm{U}^\leftarrow_{\text{safe},\text{init}}$ according to Section~\ref{sec:safeinputsequence}. In this case the shifted safe input set $\bm{U}_{\text{safe},0} = \bm{U}^\leftarrow_{\text{safe},\text{init}}$ guarantees that $\gls{chiX}_1^{\bm{U}_{\text{safe},0}} \in \Gamma_1$. Therefore, $\Gamma_0 \neq \emptyset \Rightarrow \Gamma_{1} \neq \emptyset$.

For $h=1$ it holds that $\gls{chiX}_1^{\bm{U}_{\text{safe},0}} \in \Gamma_1$. A feasible \gls{FT} \gls{ocp} yields the new safe input sequences $\bm{U}_{\text{safe},1}$, such that there exists a safe trajectory $\gls{chiX}_{2}^{\bm{U}_{\text{safe},{1}}} \in \Gamma_{2}$. If the \gls{FT} \gls{ocp} is infeasible, reusing the still valid previous safe input set $\bm{U}_{\text{safe},0}$, i.e., setting $\bm{U}_{\text{safe},1} = \bm{U}^\leftarrow_{\text{safe},0}$, ensures that $\gls{chiX}_{2}^{\bm{U}_{\text{safe},1}} \in \Gamma_{2}$.

For time step $h \geq 2$ it holds that $\gls{chiX}_h^{\bm{U}_{\text{safe},h-1}} \in \Gamma_h$. If the \gls{FT} \gls{ocp} is feasible, this yields the new safe input sequences $\bm{U}_{\text{safe},h}$, such that there exists a safe trajectory $\gls{chiX}_{h+1}^{\bm{U}_{\text{safe},{h}}} \in \Gamma_{h+1}$. If the \gls{FT} \gls{ocp} is infeasible, the previous safe input set $\bm{U}_{\text{safe},h-1}$ is still valid and choosing $\bm{U}_{\text{safe},h} = \bm{U}^\leftarrow_{\text{safe},h-1}$ ensures that $\gls{chiX}_{h+1}^{\bm{U}_{\text{safe},h}} \in \Gamma_{h+1}$.

Therefore, $\gls{chiX}_{h+1}^{\bm{U}_{\text{save},{h}}} \in \Gamma_{h+1}$ holds for all $h \in \mathbb{N}$, i.e., the proposed method is safe and recursively feasible.
\end{proof}

\section{Proof of Lemma 1}
\label{sec:appendixc}

\begin{proof}
As longitudinal and lateral motion are uncorrelated, the covariance matrix is given in terms of a block diagonal matrix
\begin{equation}
		\bm{\Sigma}^{e}_k = \begin{bmatrix} \bm{\Sigma}^e_{x,k} & 0 \\ 0 & \bm{\Sigma}^e_{y,k} \end{bmatrix},
\end{equation}
with
\begin{equation}
        \begin{aligned}
		\bm{\Sigma}^{e}_{x,k} &= \begin{bmatrix} \sigma^2_{x,k} & \sigma^2_{xv_x,k} \\ \sigma^2_{xv_x,k} & \sigma^2_{v_x,k} \end{bmatrix}, \\
		\bm{\Sigma}^{e}_{y,k} &= \begin{bmatrix} \sigma^2_{y,k} & \sigma^2_{yv_y,k} \\ \sigma^2_{yv_y,k} & \sigma^2_{v_y,k} \end{bmatrix},
		\end{aligned}
		\label{aeq4.20}
\end{equation}
and each direction is computed independently. We will show that the error for the position coordinate $x$ is distributed with a probability density function only depending on $\sigma^2_{x,k}$. Let $\bm{\xi}_{x,k} = [x_k,~v_{x,k}]^\top$ be the state vector projected onto its longitudinal coordinates with the estimated states $\hat{\bm{\xi}}_{x,k} = [\hat{x}_k,~\hat{v}_{x,k}]^\top$ and $\bm{e}_{x,k} = \hat{\bm{\xi}}_{x,k} - \bm{\xi}_{x,k}$, then $\bm{e}_{x,k}$ is Gaussian distributed with the bivariate probability function
\begin{equation}
		f(\bm{e}_{x,k}) = \frac{1}{2\pi\sqrt{ \mathrm{det} \left(\bm{\Sigma}^{e}_{x,k} \right)}} \exp \left(-\frac{1}{2} \bm{e}_{x,k}^\top \left(\bm{\Sigma}^{e}_{x,k} \right)^{-1} \bm{e}_{x,k} \right). \IEEEeqnarraynumspace
		\label{aeq4.21}
\end{equation}
We obtain the marginal probability density function by integrating $f(\bm{e}_{x,k})$ over the longitudinal velocity, i.e.,
\begin{equation}
	f_x(\bm{e}_{x,k}) = \int_{-\infty}^{\infty} f(\bm{e}_{x,k}) \mathrm{d}v_{x,k},
	\label{aeq4.22}
\end{equation}
which yields
\begin{equation}
	f_x(\bm{e}_{x,k}) = \frac{1}{\sqrt{2 \pi \sigma^2_{x,k}}} \exp \left(- \frac{1}{2 \sigma^2_{x,k}} x_k^2 \right).
	\label{aeq4.24}
\end{equation}
As \eqref{aeq4.24} corresponds to the univariate Gaussian distribution of the $x$-coordinate with only $\sigma_{x,k}$, the $x$-direction is computed independently of $\sigma_{v_{x},k}$ and $\sigma_{xv_x,k}$, which proves the lemma. The same proof holds for the lateral direction.
\end{proof}

\section{Constraint Generation}
\label{sec:appendixb}

Here, we give a complete overview of the cases considered. The cases and conditions for \gls{SMPC} are found in Table~\ref{tab:smpc_constraints}. The values $c^\text{EV/TV}_{i,x/y,k}$, $i \in \{1,2,3,4\}$ indicate the corner $x$-position and $y$-position of the \gls{ev} shape or \gls{tv} safety rectangle, according to Figure~\ref{fig:corners}. 
\begin{figure}
\centering
\includegraphics[width = \columnwidth]{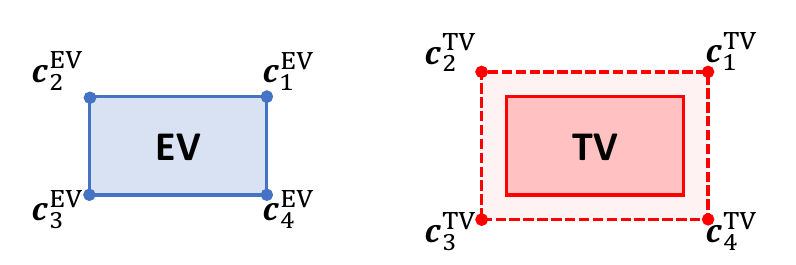}
\caption{Corner description for the \gls{ev} shape and \gls{tv} safety rectangle.}
\label{fig:corners}
\end{figure}
While the main idea of the cases is already given in Section~\ref{sec:constraintgeneration}, some extra details were not previously mentioned for clarity. The function \gls{ffclose} was considered to be a simple constant, however, the full function consists of a constant part \gls{dclose} and a variable part $\max\{0,\sign{\Delta v} \Delta v N \Delta t\}$, i.e., the necessary distance between the \gls{ev} and \gls{tv} depends on the velocity difference. A larger velocity difference results in greater difference of distance covered by the \gls{ev} and \gls{tv} within the prediction time $N \Delta t$. This yields 
\begin{IEEEeqnarray}{C}
\gls{ffclose} = \gls{dclose} + \max\{0,\sign{\Delta v} \Delta v N \Delta t\}.
\end{IEEEeqnarray}
As the general plan for the \gls{ev} is to overtake only on the left, case E is extended slightly compared to Section~\ref{sec:constraintgeneration}. When the \gls{ev} is too close to the \gls{tv}, based on a left-lane margin $r_\text{llm}$, a vertical constraint behind the \gls{tv} replaces the inclined constraint (case E${}_3$). It is only planned to overtake \glspl{tv} on a lane left of the \gls{ev} if the \gls{ev} velocity is larger than the \gls{tv} velocity (case E${}_2$). For cases D and E the constraint slope is bounded such that it does not lie within with the \gls{ev} shape or the \gls{tv} safety rectangle. 
\begin{table*}
\center
\caption{\gls{SMPC} Constraints}
\label{tab:smpc_constraints}
\begin{tabular}{c  l l  c c c}
\toprule
case & longitudinal conditions (pos. and vel.) & lateral conditions (pos.) & $q_{x,k}$ & $q_{y,k}$ & $q_{t,k}$\\
\midrule
A & $|\gls{dxEVTV}| \geq r_\text{lar}$ & - & 0 & 0 & 0\\
\addlinespace
B & $\gls{ffclose} < -\lr{\gls{dxEVTV}} < r_\text{lar}$ & - & 1 & 0 & $-c^\text{TV}_{2,x,k}$\\
\addlinespace
C & $\gls{ffclose} < \lr{\gls{dxEVTV}} < r_\text{lar}$ & - & -1 & 0 & $c^\text{TV}_{1,x,k}$\\
\addlinespace
D & $ 0 \leq -\lr{\gls{dxEVTV}} \leq \gls{ffclose}$ & $\gls{evlane} = \gls{tvlane}$ & $\max \left\lbrace 0, \frac{c^\text{EV}_{4,y,0} - c^\text{TV}_{2,y,k}}{c^\text{EV}_{4,x,0} - c^\text{TV}_{2,x,k}} \right\rbrace$ & -1 & $c^\text{TV}_{2,y,k} - s_{x,k} c^\text{EV}_{4,x,0}$\\
\addlinespace
E & \begin{tabular}{@{}l@{}} $-\lr{\gls{dxEVTV}} \leq \gls{ffclose}$, \\ $s_0 + 0.5 w_\text{veh} + r_\text{llm} \leq \gls{xtv0}$, \\ $v_0 - v_{x,0}^\text{TV} > 0$ \end{tabular} & $\gls{evlane} + w_\text{lane} = \gls{tvlane}$ & $\max \left\lbrace 0, \frac{c^\text{EV}_{4,y,0} - c^\text{TV}_{2,y,k}}{c^\text{EV}_{4,x,0} - c^\text{TV}_{2,x,k}} \right\rbrace$ & -1 & $c^\text{TV}_{2,y,k} - s_{x,k} c^\text{EV}_{4,x,0}$\\
\addlinespace
E${}_2$ & \begin{tabular}{@{}l@{}} $-\lr{\gls{dxEVTV}} \leq \gls{ffclose}$, \\ $s_0 + 0.5 w_\text{veh} + r_\text{llm} \leq \gls{xtv0}$, \\ $v_0 - v_{x,0}^\text{TV} > 0$ \end{tabular} & $\gls{evlane} + w_\text{lane} = \gls{tvlane}$ & 1 & 0 & $-c^\text{TV}_{2,x,k}$\\
\addlinespace
E${}_3$ & \begin{tabular}{@{}l@{}} $0 \leq -\lr{\gls{dxEVTV}} \leq \gls{ffclose}$, \\ $s_0 + 0.5 w_\text{veh} + r_\text{llm} > \gls{xtv0}$ \end{tabular} & $\gls{evlane} + w_\text{lane} = \gls{tvlane}$ & 0 & 1 & $-c^\text{TV}_{3,y,k}$\\
\addlinespace
F & $|\gls{dxEVTV}| \leq \gls{ffclose}$ & $\gls{evlane} > \gls{tvlane}$ & 0 & -1 & $c^\text{TV}_{2,y,k}$\\
\addlinespace
G & $0 \leq -\lr{\gls{dxEVTV}} \leq \gls{ffclose}$ & $\gls{evlane} + 2w_\text{lane} \leq \gls{tvlane}$ & 0 & 1 & $-c^\text{TV}_{3,y,k}$\\
\addlinespace
H & $0 < \lr{\gls{dxEVTV}} \leq \gls{ffclose}$ & $\gls{evlane} \leq \gls{tvlane}$ & 0 & 1 & $-c^\text{TV}_{4,y,k}$\\
\addlinespace
J & $0 < \lr{\gls{dxEVTV}} \leq \gls{ffclose}$ & $\gls{evlane} = \gls{tvlane}$ & 0 & 0 & 0\\
\bottomrule
\end{tabular}
\end{table*}

The cases and conditions for \gls{FT} are similar to the \gls{SMPC} cases. A complete description is given in Table~\ref{tab:ft_constraints}. The main idea of the \gls{FT} cases are described in Section~\ref{sec:constraintgeneration_ft}. Here, we give a detailed description of the cases where three placeholder \gls{tv} predictions are considered (cases J${}^*$ and C${}^*$). The cases with a \gls{tv} prediction in the same lane as the \gls{ev} are denoted by J${}^*_\text{S}$, C${}^*_\text{S}$, while cases with a \gls{tv} prediction in a lane to the left or right of the \gls{tv} are denoted by J${}^*_\text{L}$, C${}^*_\text{L}$ and J${}^*_\text{R}$, C${}^*_\text{R}$, respectively. The slopes for cases C${}^*_\text{L}$ and C${}^*_\text{R}$ are limited to $q_{x,k} \leq 0$ and $q_{x,k} \geq 0$, respectively. If the values for $q_{x,k}$ exceed the respective limits, the cases C${}^*_\text{L,lim}$ and C${}^*_\text{R,lim}$ are applied. The cases F${}_2^*$ and H${}_2^*$ represent scenarios where the \gls{ev} center is not in the \gls{tv} lane, but the \gls{ev} shape is already in the \gls{tv} lane. In these cases, where the \gls{ev} is behind the \gls{tv}, vertical constraints behind the \gls{tv} are built. The cases F${}^*$ and H${}^*$ are split into F${}_a^*$, F${}_b^*$ and H${}_a^*$, H${}_b^*$, respectively.
\begin{table*}
\center
\caption{\gls{FT} Constraints}
\label{tab:ft_constraints}
\begin{tabular}{c  l l  c c c}
\toprule
case & longitudinal conditions (pos. and vel.) & lateral conditions (pos.) & $q_{x,k}$ & $q_{y,k}$ & $q_{t,k}$\\
\midrule
A${}^*$ & $|\gls{dxEVTV}| \geq r_\text{lar}$ & - & 0 & 0 & 0\\
\addlinespace
B${}^*$ & $\gls{ffcloseft} < -\lr{\gls{dxEVTV}} < r_\text{lar}$ & - & 1 & 0 & $-c^\text{TV}_{2,x,k}$\\
\addlinespace
D${}^*$ & $ 0 \leq -\lr{\gls{dxEVTV}} \leq \gls{ffcloseft}$ & $\gls{evlane} = \gls{tvlane}$ & 1 & 0 & $-c^\text{TV}_{2,x,k}$\\
\addlinespace
F${}_a^*$ & $0 < \lr{\gls{dxEVTV}} \leq \gls{ffcloseft}$ & $\gls{evlane} > \gls{tvlane}$& 0 & -1 & $c^\text{TV}_{2,y,k}$\\
\addlinespace
F${}_b^*$ & $0 \leq -\lr{\gls{dxEVTV}} \leq \gls{ffcloseft}$ & \begin{tabular}{@{}l@{}} $\gls{evlane} > \gls{tvlane}$, \\ $d_0 \geq \gls{tvlane} + 0.5w_\text{lane} + 0.5w_\text{veh}$ \end{tabular} & 0 & -1 & $c^\text{TV}_{2,y,k}$\\
\addlinespace
F${}_2^*$ & $0 \leq -\lr{\gls{dxEVTV}} \leq \gls{ffcloseft}$ & \begin{tabular}{@{}l@{}} $\gls{evlane} > \gls{tvlane}$, \\ $d_0 < \gls{tvlane} + 0.5w_\text{lane} + 0.5w_\text{veh}$ \end{tabular} & 1 & 0 & $-c^\text{TV}_{2,x,k}$\\
\addlinespace
H${}_a^*$ & $0 < \lr{\gls{dxEVTV}} \leq \gls{ffcloseft}$ & $\gls{evlane} < \gls{tvlane}$ & 0 & 1 & $-c^\text{TV}_{3,y,k}$\\
\addlinespace
H${}_b^*$ & $0 \leq -\lr{\gls{dxEVTV}} \leq \gls{ffcloseft}$ & \begin{tabular}{@{}l@{}} $\gls{evlane} < \gls{tvlane}$, \\ $d_0 \leq \gls{tvlane} - 0.5w_\text{lane} - 0.5w_\text{veh}$ \end{tabular} & 0 & 1 & $-c^\text{TV}_{3,y,k}$\\
\addlinespace
H${}_2^*$ & $0 \leq -\lr{\gls{dxEVTV}} \leq \gls{ffcloseft}$ & \begin{tabular}{@{}l@{}} $\gls{evlane} < \gls{tvlane}$, \\ $d_0 > \gls{tvlane} - 0.5w_\text{lane} - 0.5w_\text{veh}$ \end{tabular} & 1 & 0 & $-c^\text{TV}_{2,x,k}$\\
\addlinespace
J${}^*_\text{S}$ & $0 < \lr{\gls{dxEVTV}} \leq \gls{ffcloseft}$ & $\gls{evlane} = \gls{tvlane}$ & -1 & 0 & $c^\text{TV}_{1,x,k}$\\
\addlinespace
J${}^*_\text{L}$ & $0 < \lr{\gls{dxEVTV}} \leq \gls{ffcloseft}$ & $\gls{evlane} = \gls{tvlane}$ & 0 & 1 & $-c^\text{TV}_{4,y,k}$\\
\addlinespace
J${}^*_\text{R}$ & $0 < \lr{\gls{dxEVTV}} \leq \gls{ffcloseft}$ & $\gls{evlane} = \gls{tvlane}$ & 0 & -1 & $-c^\text{TV}_{1,y,k}$\\
\addlinespace
C${}^*_\text{S}$ & $\gls{ffcloseft} < \lr{\gls{dxEVTV}} < r_\text{lar}$ & $\gls{evlane} = \gls{tvlane}$ & -1 & 0 & $c^\text{TV}_{1,x,k}$\\
\addlinespace
C${}^*_\text{L}$ & $\gls{ffcloseft} < \lr{\gls{dxEVTV}} < r_\text{lar}$ & $\gls{evlane} = \gls{tvlane}$ & $ \frac{c^\text{EV}_{3,y,0} - c^\text{TV}_{1,y,k}}{c^\text{EV}_{3,x,0} - c^\text{TV}_{1,x,k}} $ & -1 & $c^\text{TV}_{1,y,k} - s_{x,k} c^\text{EV}_{3,x,0}$\\
\addlinespace
C${}^*_\text{L,lim}$ & $\gls{ffcloseft} < \lr{\gls{dxEVTV}} < r_\text{lar}$ & $\gls{evlane} = \gls{tvlane}$ & -1 & 0 & $c^\text{TV}_{1,x,k}$\\
\addlinespace
C${}^*_\text{R}$ & $\gls{ffcloseft} < \lr{\gls{dxEVTV}} < r_\text{lar}$ & $\gls{evlane} = \gls{tvlane}$ & $- \frac{c^\text{EV}_{2,y,0} - c^\text{TV}_{4,y,k}}{c^\text{EV}_{2,x,0} - c^\text{TV}_{4,x,k}}$ & 1 & $-\lr{c^\text{TV}_{4,y,k} - s_{x,k} c^\text{EV}_{2,x,0}}$\\
\addlinespace
C${}^*_\text{R,lim}$ & $\gls{ffcloseft} < \lr{\gls{dxEVTV}} < r_\text{lar}$ & $\gls{evlane} = \gls{tvlane}$ & 1 & 0 & $-c^\text{TV}_{4,x,k}$\\
\bottomrule
\end{tabular}
\end{table*}

\section*{Acknowledgment}

The authors thank Daniel Althoff, Matthias Althoff, and Christian Pek for valuable discussions.

\ifCLASSOPTIONcaptionsoff
  \newpage
\fi

\bibliography{./references/root}
\bibliographystyle{unsrt}

\end{document}